\documentclass[journal]{IEEEtran}
\usepackage[utf8]{inputenc}
\usepackage{cite}
\usepackage{amsmath,amssymb,amsfonts,amsthm}
\usepackage{float} 
\usepackage{graphicx}
\usepackage{subfig}
\usepackage{graphicx}
\newtheorem{theorem}{Theorem}

\newtheorem{definition}{Definition}
\newtheorem{assumption}{Assumption}
\theoremstyle{definition}
\newtheorem{remark}{\normalfont\bfseries Remark}
\theoremstyle{definition}
\newtheorem{example}{\normalfont\bfseries Example}
\usepackage{xcolor}

\newcommand{\delayu}{\tau_u}
\newcommand{\delaysensor}{\tau_y}
\newcommand{\diff}{\,\mathrm{d}}
\newcommand{\drupper}{\bar{a}}
\newcommand{\drlower}{\underline{a}}
\newcommand{\eobound}{\Gamma}
\newcommand{\EKF}{extended class $\mathcal{K}_{\infty}$ function}

\usepackage{pifont}

\title{Robust Safety for Mixed-Autonomy Traffic with\\ Delays and Disturbances}
\author{{Chenguang Zhao, Huan Yu}
\thanks{*Huan Yu is the corresponding author. Email: {\it huanyu@ust.hk}}
 \thanks{Chenguang Zhao and Huan Yu are with the Hong Kong University of Science and Technology (Guangzhou), Thrust of Intelligent Transportation, Nansha, Guangzhou, 511400, Guangdong, China. Huan Yu is also affiliated with the Hong Kong University of Science and Technology, Department of Civil and Environmental Engineering, Hong Kong SAR, China}}

\begin{document}
\maketitle

\begin{abstract}
Various control strategies and field experiments have been designed for connected and automated vehicles (CAVs) to stabilize mixed traffic that contains both CAVs and Human-driven Vehicles (HVs). The effect of these stabilizing CAV control strategies on traffic safety is still under investigation.    In an effort to prioritize safety over stability, a safety-critical filter via control barrier functions (CBFs) can be designed by modifying the stabilizing nominal control input in a minimal fashion and imparting collision-free driving behaviors for CAVs and HVs. However, such formal safety guarantees can be violated if there are delays in the actuation and communication channels of the CAV. Considering both actuator and sensor delays, and disturbances, we propose robust safety-critical traffic control (RSTC) design to ensure ``robust safety" of the mixed traffic. While predictor-based CBF has been developed to compensate for the actuator delay, uncertain speed disturbances from the head vehicle  cause prediction error and require novel robust CBF design. Besides, safety-critical control  with sensor delay also remains an open question. 
In RSTC, a state predictor with bounded error is designed, and robust CBF constraints are constructed to guarantee safety under actuator delay and disturbances. When there is a sensor delay, a state observer is designed and integrated with a predictor-based CBF to ensure robust safety. Numerical simulations demonstrate that the proposed RSTC avoids rear-end collisions for two unsafe traffic scenarios in the presence of actuator, sensor delays and disturbances.
\end{abstract}


\section{Introduction}

The potential of connected and automated vehicles (CAVs) in mitigating traffic congestion through stabilization has continuously gained attention over the recent years, including theoretical analysis~\cite{cui2017stabilizing,yu2018stabilization} and field experiments~\cite{stern2018dissipation,jin2018experimental}. To stabilize traffic via CAVs, controllers have been designed under different communication typologies~\cite{kesting2008adaptive,orosz2016connected,wang2021leading}, and via various control methods such as feedback control~\cite{wang2023general}, $\mathcal{H}_{\infty}$ control~\cite{zhou2020stabilizing}, and  optimal control~\cite{wang2020controllability,le2022cooperative,wang2021optimal,jin2020analysis}. However, scarce studies have been conducted on how such stabilizing CAV controllers affect traffic safety.

In real-world traffic scenarios, how to formally guarantee the safety of traffic systems is imperative. For example, one major factor affecting public acceptance of autonomous driving is its safety performance, particularly regarding the possibility of rear-end collisions~\cite{cunningham2019buy}. To guarantee traffic safety, representative technologies include reachability analysis~\cite{althoff2014online,mitchell2005time}, model predictive control~\cite{gong2018cooperative,feng2021robust}, and  control barrier function (CBF)~\cite{ames2014control,xiao2021bridging,zhao2023safety}. Compared with the other approaches, CBF provides certified safety by directly synthesizing a safety-critical controller from modifying user-selected controllers, and thus offers some degree of freedom in design. In the authors' previous work~\cite{zhao2023safety}, a safety-critical traffic control (STC) has been developed using CBF, which provides safety guarantee for pre-designed traffic stabilizing nominal controller for mixed-autonomy systems.

Previous works have focused on control problems of delay-free traffic systems, while in this article, delayed traffic systems with disturbances are considered. In practice, CAV controllers will face intrinsic delays from multiple sources. For example, onboard sensors measure vehicle speed and gap, which are then filtered to reduce measurement noise. This process of sensing and filtering brings delays due to discrete sampling of sensors, radar or lidar filtering~\cite{wang2018delay,xiao2011practical}. The CAV collects information from the surrounding traffic environment or from following vehicles through vehicle-to-infrastructure  or vehicle-to-vehicle wireless communication, which causes transmission delays due to the scheduling algorithms to send packets, the computation of onboard computers, and packet drops~\cite{jin2014dynamics,beregi2021connectivity}. These delays can be lumped as a sensor delay in the measurement~\cite{ma2022string}. In addition to the sensor delay, another type of delay is the actuator delay in the controller.  When the vehicle control system executes the acceleration command from the controller, there are delays from  the engine response, the throttle actuator, or the brake actuator~\cite{xiao2011practical}. Besides delays, practical traffic systems are also subject to disturbances that evolve independently from the system and the control input  and thus bring uncertainty and challenge to controller design. In this paper, the disturbance comes from the speed of the leading vehicle ahead of the controlled CAV.

The stability of traffic systems is jeopardized in the presence of delays. Stability analysis in \cite{jin2014dynamics} and \cite{beregi2021connectivity} shows that for a feedback controller, its stability region, i.e., the range of feedback gain that stabilizes traffic,  shrinks when there are sensor and actuator delays, respectively. While some works have analyzed the effect of delays on stability and designed delay-robust stabilizing controllers~\cite{beregi2021connectivity,wang2022design,bekiaris2023robust}, a formal safety guarantee for delayed mixed-autonomy traffic is still lacking. Simulations in~\cite{molnar2022input} show that delays may cause safety violations, such as rear-end collisions, even when the delayed system is still stable. In this paper, we mainly focus on designing robust safety-critical CAV controllers to guarantee safety in mixed traffic.

\begin{figure*}[t!]
    \centering
    \includegraphics[width=0.65\textwidth]{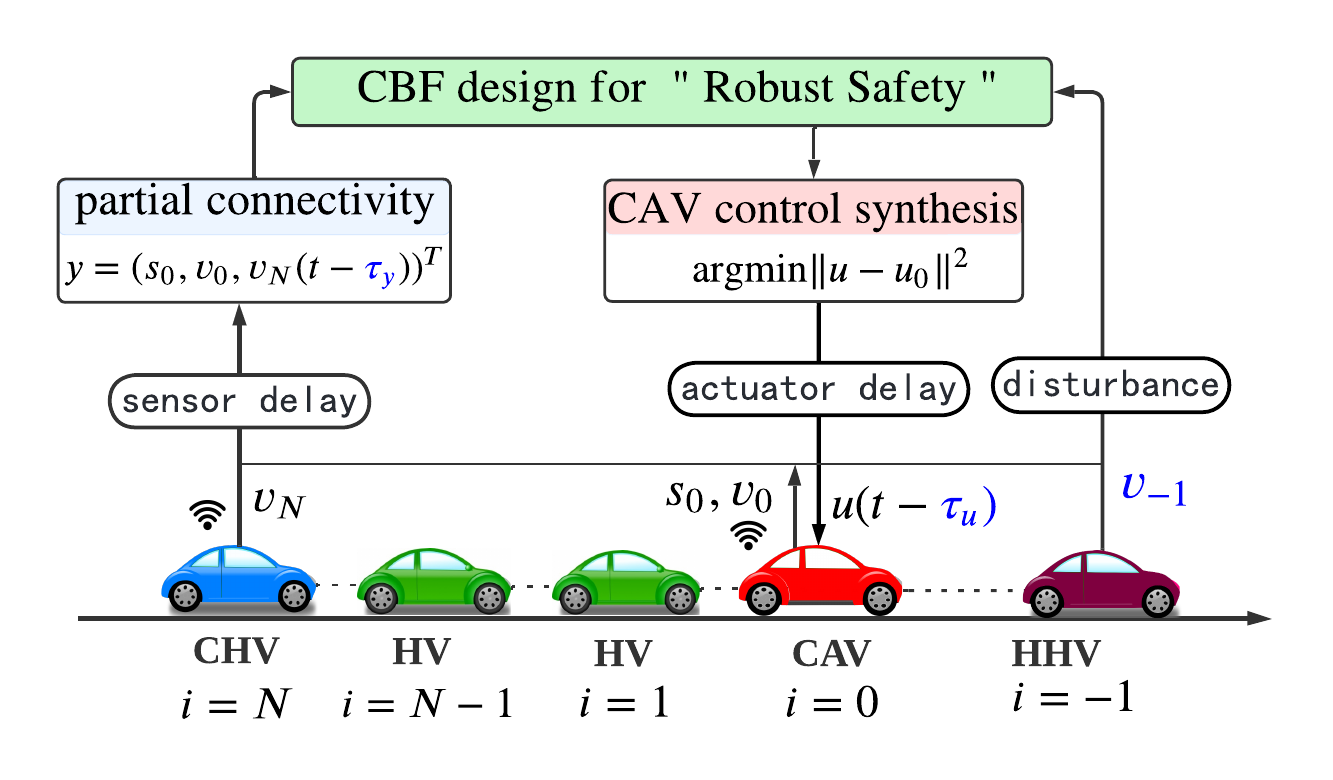}
    \caption{The proposed robust safety-critical traffic control (RSTC) framework, which achieves rear-end collision-free "robust safety" in mixed autonomy traffic  by controlling a CAV in the existence of actuator delay, sensor delay, and disturbances.}
    \label{fig:framework}
\end{figure*}

To ensure safety for systems with actuator delay, CBF is integrated with a state predictor that predicts future state value over the delay interval given the current and historical dynamics~\cite{abel2021safety,singletary2020control}. For systems with disturbances, the notion of  input-to-state safety (ISSf) has been developed and integrated with CBFs to characterize safety under disturbances~\cite{kolathaya2018input,krstic2021inverse}. We consider in this paper a delayed mixed-autonomy system with disturbances that come from the speed of a leading vehicle ahead of the controlled CAV.  Neither the aforementioned predictor-based CBF approach nor the ISSf-CBF approach can be directly applied to such a system. The main challenge remains in that the dynamics of this external disturbance are unknown, which induces prediction errors that could violate the naive predictor-based CBF safety constraints~\cite{molnar2022safety}. In~\cite{molnar2022input}, predictor-based  ISSf-CBF is designed, which allows safety violation and only guarantees forward invariance of a safe set larger than the original one without disturbances. In~\cite{molnar2022safety}, robust CBF constraints are designed to guarantee safety under the worst-case prediction. The analysis in~\cite{molnar2022safety} focuses on the specific scenario of one CAV following one HV, which ignores the risk from following vehicles, such as when the following vehicle suddenly accelerates. Therefore, robust safety guarantee for mixed autonomy against actuator delay and disturbances remains an open question.

For systems without sensor delay, observer-based CBF has been developed in~\cite{agrawal2022safe,wang2022observer}. A state observer is designed to estimate un-measurable states, and the estimated state is then used to construct CBF constraints that guarantee safety considering the estimation error. When there is a sensor delay in the measurement, how to combine a delay-compensating observer~\cite{watanabe1981observer} with CBFs to achieve safety-critical control under delayed partial measurement, to the best of the authors' knowledge, has not been investigated. 
Therefore, this article provides robust safety guarantee via a novel CBF design that accounts for the actuator delay, sensor delay, and disturbances.

To summarize, existing safety-critical controllers designed for CAVs may still cause rear-end collisions due to the three challenges that arise from real-world traffic systems: actuator delay, sensor delay, and external disturbances. 
To bridge this gap, this paper develops robust safety-critical traffic control (RSTC) as shown in Fig.~\ref{fig:framework} to impart formal safety guarantee to mixed-autonomy systems. To compensate for the actuator delay, we design a state predictor and prove its prediction error bound  under the assumption of a bounded derivative of the disturbance, which is the acceleration of the head vehicle. Robust predictor-based CBF constraints are then designed for  both the CAV and the following vehicles. Considering the sensor delay, a predictor-observer is designed to estimate the system state, and CBF constraints are constructed considering both the estimation error and prediction error.  The contribution of this paper is two-fold.
\begin{itemize}
    \item The theoretical novelty lies in providing the first robust CBF design for mixed-autonomy systems with actuator delay, sensor delay, and disturbances. 
    \item For application, the proposed RSTC framework guarantees ``robust safety" for mixed traffic. The design acts as a safety filter that is flexible to integrate with various existing CAV controllers for rear-end collision avoidance.
\end{itemize}

The remainder of this paper is organized as follows. In Section \ref{sec:model}, we formulate the mixed autonomy traffic system with actuator and sensor delays. In Section \ref{sec:preliminary CBF}, we introduce preliminary knowledge of  CBF and give examples of how CBF is utilized to design safety-critical traffic controllers. Section \ref{sec:RSTC} designs CBF constraints to guarantee  mixed traffic safety with delays. The proposed RSTC is validated and analyzed by numerical simulation in Section \ref{sec:simulation}.

\section{Mixed-autonomy traffic with actuator and sensor delays} \label{sec:model}

We consider the longitudinal control of a CAV in a mixed vehicle chain as shown in Fig. \ref{fig:framework}. The CAV follows a head human-driven vehicle (HHV) and leads $N$ following human-driven vehicles (HVs). We index the HHV as $-1$, the CAV as $0$, and the following HVs as $1$ to $N$.

The car-following dynamics of the HVs are described as
\begin{align}
\dot s_i(t) &= v_{i-1}(t) - v_{i}(t), \label{eq:CF HDV s} \\
\dot{v}_{i}(t) &=F_i\left(s_{i}(t), \dot{s}_{i}(t),  v_{i}(t)\right), \label{eq:CF HDV v}
\end{align}
with $s_i(t) \in \mathbb{R}$ being the spacing  between HV-$i$ and its leading vehicle $i-1$, $v_i(t) \in \mathbb{R}$  and $v_{i-1}(t) \in \mathbb{R}$ being the speed of its own and the leader vehicle, respectively. The function $F_i:\mathbb{R}^3\to \mathbb{R}$ describes the human driver's driving decision based on $s_i$, $v_i$, and $v_{i-1}$. The function $F_i$ can be taken from  representative human driver models, such as the optimal velocity model (OVM) or the  intelligent driver model (IDM).

{\bf Actuator delay:} For the CAV, we use a controller to control its motion, and its longitudinal dynamics are governed by
\begin{align}
    \dot{{s}}_{0}(t) &= v_{-1}(t)- {v}_{0}(t), \label{eq:CF CAV s}\\
    \dot{{v}}_{0}(t) &=u(t-\delayu), \label{eq:CF CAV v}
\end{align}
where $v_{-1}(t) \in \mathbb{R}$ is the velocity of the head vehicle, $u(t) \in \mathbb{R}$ is the control input. In practical autonomous driving systems, there is a delay between the controller command and actual vehicle's acceleration due to delays in the engine response, throttle actuator, or brake actuator. We formulate those delays as the actuator delay $\delayu>0$ in the control input.

At the equilibrium states, all vehicles drive at a uniform speed $v^*$, and each vehicle $i$ keeps a constant gap $s_i^*$ decided by $F_i(s_i^*,0,v^*) = 0$. We take the perturbations around the equilibrium states as 
\begin{align}
    \tilde{s}_{i}(t)&=s_{i}(t)-s_i^{\star}, \\
    \tilde{v}_{i}(t)&=v_{i}(t)-v^{\star},
\end{align}
and linearize the car-following model \eqref{eq:CF HDV s}-\eqref{eq:CF HDV v}   as
\begin{align}
    \dot{\tilde{s}}_{i}(t) &= \tilde{v}_{i-1}(t) - \tilde{v}_{i}(t), \label{eq:CF_lin HDV s}\\
    \dot{\tilde{v}}_{i}(t) &= a_{i1} \tilde{s}_{i}(t) - a_{i2} \tilde{v}_{i}(t) + a_{i3} \tilde{v}_{i-1}(t), \label{eq:CF_lin HDV v}
\end{align}
where $a_{i1}=\frac{\partial F_i}{\partial s_i}, a_{i2}=\frac{\partial F_i}{\partial \dot{s}_i}-\frac{\partial F_i}{\partial v_i}, a_{i3}=\frac{\partial F}{\partial \dot{s}_i}$ are parameters evaluated at the steady states $v^*$ and $s_i^*$. For the CAV at equilibrium states, it will drive at the same equilibrium speed $v^*$, and keep a constant gap $s^*$. We linearize the car-following model for CAV \eqref{eq:CF CAV s}-\eqref{eq:CF CAV v} as
\begin{align}
    \dot{\tilde{s}}_{0}(t) &= \tilde{v}_{-1}(t) - \tilde{v}_{0}(t), \label{eq:CF_lin CAV s}\\
    \dot{\tilde{v}}_{0}(t) &= u(t-\delayu).\label{eq:CF_lin CAV v}
\end{align}
For the mixed autonomy system, we have the state variable as 
\begin{align}\label{eq:state sv tilde}
x =[\tilde{s}_{0},\;  \tilde{v}_{0},\; \tilde{s}_{1},\;  \tilde{v}_{1},\cdots,  \tilde{s}_{N},\;  \tilde{v}_{N}
]^\top \in \mathbb{R}^{n},
\end{align}
with $n = 2N+2$. The linearized system model is
\begin{equation}\label{eq:system}
    \dot{x}(t) = Ax(t)  + Bu(t-\delayu) + Dr(t),
\end{equation}
where $r(t) = \tilde{v}_{-1}(t)$ is the disturbance from the head vehicle's speed, and the model coefficients $A\in\mathbb{R}^{n\times n}$, $B\in\mathbb{R}^{n}$, $D\in\mathbb{R}^{n}$ are 
\begin{equation}\label{eq:model ABD}
\begin{split}
    A \! = \!\begin{bmatrix}
    0 & -1     &        &         &    &        &        &\\
    0 & 0      &        &         &    &        &        &\\
    0 & 1      & 0      & -1      &    &        &        &\\
    0 & a_{13} & a_{11} & -a_{12} &    &        &        &\\
      &        &        & \ddots  &    &        &        &\\
      &        &        &         & 0  & 1      & 0      & -1 \\
      &        &        &         & 0  & a_{N3} & a_{N1} & -a_{N2}
    \end{bmatrix} \!, \\
    B = \begin{bmatrix}
        0 & 1 & 0 &\cdots & 0 
        \end{bmatrix}^{\top} \!, 
   D = \begin{bmatrix}    1 & 0 &\cdots &0    \end{bmatrix}^{\top}\!\!\!.
\end{split}
\end{equation}

{\bf Sensor delay:}  In practice, some following HVs may not be equipped with communication devices, so their gap and speed information will be unknown to the CAV. In this case, the CAV only gets partial information of the system state $x(t)$. Besides, there are delays in the measurement that includes sensor detection, filter processing, or wireless communication. In this paper, we consider the sensor delay that mainly comes from communication.  We assume that $0<n_y\le n$ state variables are available to the CAV. The measurement $y(t) \in \mathbb{R}^{n_y}$ is 
\begin{align}\label{eq:yc general}
    y(t) = \sum_{j=1}^{J} C_j x(t-\tau_{y,j}),
\end{align}
where $\tau_{y,j}\ge 0$ are sensor delays and $C_j\in \mathbb{R}^{n_y\times n}$ are observation matrices. For the system to be observable,  the least information required is the speed of the CAV, the gap of the CAV, and the speed of the last vehicle HV-$N$ \cite{watanabe1981observer,wang2021leading}. In this most challenging case, the measurement $y(t) \in \mathbb{R}^3$ becomes 
\begin{equation}\label{eq:y}
    y(t) =
    \begin{bmatrix}
    \tilde{s}_{0}(t) & \tilde{v}_{0}(t) & \tilde{v}_{N}(t-\delaysensor)
    \end{bmatrix}^{\top},
\end{equation}
with $\delaysensor>0$ being the sensor delay that comes from the transmission of the velocity measurement of the last connected vehicle HV-$N$ to the CAV.  We write $y(t)$ in the general form of \eqref{eq:yc general} as 
\begin{align}\label{eq:yC}
    y(t) = C_1 x(t) + C_2 x(t-\delaysensor),
\end{align}
where $\delaysensor>0$ is the sensor delay, and the two observation matrices  $C_1 \in \mathbb{R}^{3\times n}$ and $C_2 \in \mathbb{R}^{3\times n}$ are
\begin{equation}
    C_1 = \begin{bmatrix}
    1  & 0  & \cdots &  0     \\
    0  & 1  & \cdots &  0      \\
    0  & 0  & \cdots & 0   
    \end{bmatrix},
    C_2= \begin{bmatrix}
    0  & 0  & \cdots &  0     \\
    0  & 0  & \cdots &  0      \\
    0  & 0  & \cdots & 1  
    \end{bmatrix}.
\end{equation}

\section{Preliminaries on Control Barrier Function} \label{sec:preliminary CBF}

In this section, we introduce some basic preliminaries on safety-critical control via CBF.

\subsection{CBF for delay-free system}

Consider an affine control system with the state $x\in \mathcal{D}\subset \mathbb{R}^{n_1}$ and control input $u\in \mathcal{U}\subset \mathbb{R}^{n_2}$,
\begin{equation}\label{eq:introCBF x}
    \dot{x}(t) = f(x(t)) + g(x(t))u(t),
\end{equation}
with $f:\mathbb{R}^{n_1}  \to \mathbb{R}^{n_1}$ and $g:\mathbb{R}^{n_1} \to \mathbb{R}^{n_2}$ being locally Lipschitz. To ensure safety, we use CBF defined as follows.

\begin{definition}[CBF \cite{ames2019control}] \label{def:CBF}
For the system \eqref{eq:introCBF x}, let 
\begin{equation}\label{eq:introCBF safe set}
    \mathcal{S} = \left\{ x\in \mathcal{D} : h(x)\ge 0\right\}
\end{equation}
be the superlevel set of a continuously differentiable function $h:\mathcal{D}\to \mathbb{R}$, the function $h$  is called a control barrier function for the system \eqref{eq:introCBF x} on $\mathcal{S}$  if there exists an \EKF \  $\alpha$ such that 
\begin{align}
    \sup_{u\in \mathcal{U}} \dot{h}(x,u)  \geq -\alpha(h(x)), \quad \forall x\in \mathcal{D}, 
\end{align}
where the time derivative of $h$ is
\begin{align}
    \dot{h}(x,u) = L_{f} h(x)+L_{g} h(x) u,
\end{align}
with $L_fh = \frac{\partial h(x)}{\partial x}  f(x)$ and $L_gh  = \frac{\partial h(x)}{\partial x}  g(x)$ being the Lie derivatives. An \EKF \ is a function $\alpha:\mathbb{R} \to \mathbb{R}$ that is strictly increasing with $\alpha(0) = 0$.
\end{definition}

CBF guarantees the safety of the system as stated in Theorem \ref{theorem:CBF}.

\begin{theorem}[Safety guarantee for delay-free system \cite{ames2019control}]\label{theorem:CBF}
    If $h$ is a CBF, then any Lipschitz controller $u$ that satisfies
    \begin{align}
        L_fh(x) + L_gh(x)u+\alpha(h(x))\ge 0
    \end{align}
    renders the set $\mathcal{S}$ \eqref{eq:introCBF safe set}  forward invariant, which means if $x(0) \in \mathcal{S}$, then $x(t) \in \mathcal{S}$ for all $t\ge 0$.
\end{theorem}

\begin{example}[STC for delay-free mixed autonomy systems]\label{example:STC}
The delay-free mixed autonomy system is written as \eqref{eq:system}, 
\begin{align}
    \dot x(t) = Ax(t) + Bu(t) + Dr(t).
\end{align}
A nominal controller designed as
\begin{align}
    u_0(t) = Kx(t) + \alpha_3 r(t)
\end{align}
achieves string stability with proper choice of feedback gain $K$~\cite{wang2021leading}. But it may cause rear-end collisions since safety is ignored when stabilizing traffic. In   our previous work~\cite{zhao2023safety}, STC is proposed to ensure safety of  mixed traffic via CBF. We adopt the constant time headway (CTH) spacing policy for safety as
\begin{equation}\label{eq:safe s}
    s_i \geq \psi_i v_i, \quad \forall i=0,1,\cdots,N,
\end{equation}
with $\psi_i>0$ being a safe time-headway. With CTH, we define the safe set for each vehicle as
\begin{equation}\label{eq:safe set}
    \mathcal{S}_{i}=\left\{x \in \mathbb{R}^{n}: h_{i}(x) \geq 0\right\},
\end{equation}
with the safety function
\begin{align}\label{eq:h}
    h_{i}(x) & =  s_i - \psi_i  v_i.
\end{align}
\begin{itemize}
    \item To guarantee CAV safety, the control input should satisfy 
    \begin{equation}\label{eq:CBFintro CAV}
        L_fh_0(x) + L_gh_0(x) u + \alpha(h(x)) \ge 0,
    \end{equation}
    with $f(x) = Ax + Dr$, $g(x) = B$.
    \item For HV safety, the function $h_i$ for HV-$i$ has a relative degree $i+1$ with respect to the control input $u$. High relative degree CBF requires a more complex formulation and is more sensitive to parameters \cite{ames2019control,xiao2022high}. Therefore, we design reduced-degree CBF for the following vehicles with a relative degree one as $h_{i}^{\mathrm{r}}$:
    \begin{align}\label{eq:hir}
        h_{i}^{\mathrm{r}}(x) = h_i(x) - \eta_i h_0(x),
    \end{align}
    with $\eta_i>0$ being a positive coefficient. We see that if $h_{i}^{\mathrm{r}}(x)\ge0$ and $ h_0(x)\ge0$, then $h_i(x)\ge 0$, which implies that the safe criterion for HVs is met. And the CBF constraints for HV becomes
    \begin{equation}\label{eq:CBFintro HV}
        L_fh_{i}^{\mathrm{r}} (x) + L_gh_{i}^{\mathrm{r}}(x) u + \alpha(h_{i}^{\mathrm{r}}(x)) \ge 0.
    \end{equation}
\end{itemize}
The proposed STC~\cite{zhao2023safety} synthesizes a safety-critical controller by solving a QP:
\begin{align}
    & u = \underset{u \in \mathbb{R}, \sigma_i \geq 0}{\operatorname{argmin}} \;  |u-u_0|^{2} + \sum_{i=1}^N p_i\sigma_i^2 \label{eq:QP CBFintro traffic}\\
  \text{s.t.}\;\; & 
 L_fh_0(x) + L_gh_0(x) u   +  \alpha(h_0(x))  \ge 0  \notag \\
    & L_f h_{1}^{\mathrm{r}} (x)  + L_g h_{1}^{\mathrm{r}} (x)u  +\alpha (h_{1}^{\mathrm{r}} (x) ) +\sigma_1 \ge 0 \notag\\
    & \qquad \vdots \notag\\
    & L_f h_{N}^{\mathrm{r}} (x) \! + \! L_g h_{N}^{\mathrm{r}} (x)u     +  \alpha(h_{N}^{\mathrm{r}} (x) )   +\sigma_N \ge 0 ,\notag 
\end{align}
where $\sigma_i$ are slack variables to ensure the feasibility of the QP,   and $p_i>0$ are penalty coefficients.
\end{example}

\subsection{CBF for system with actuator delay}

Consider the system with an actuator delay $\tau>0$,
\begin{align}\label{eq:introCBFdelay x}
    \dot x(t) = f(x(t)) + g(x(t)) u(t-\tau),
\end{align}
with $f$ and $g$ being the same as \eqref{eq:introCBF x}. To synthesize safety-critical controllers when there is an actuator delay, CBF is integrated with a state predictor to compensate for the actuator delay.

For the system \eqref{eq:introCBFdelay x}, the state over $[t,t+\tau]$ is
\begin{align}
    x(t+s) = P(s,x(t),u_t), \quad \forall s\in [0,\tau],
\end{align}
with
\begin{align}
    P(s,x,u_t) =& x(t) +\int_0^{s}  f\left(P(\theta,x(t),u_t)\right) \diff \theta \notag \\
    &+ \int_0^{s} g\left(P(\theta,x(t),u_t)\right) u_t(s-\tau) \diff \theta, 
\end{align}
where 
\begin{align}
u_t(s) = u(t+s), \quad s\in[-\tau,0),
\end{align}
is the historical input over $[t-\tau,t)$.
The predicted state value of $x(t+\tau)$ is
\begin{align}\label{eq:introCBFdelay predict}
    x_p(t) = P(\tau,x(t),u_t).
\end{align}
For the predicted system $x_p(t)$, it has a delay-free dynamics:
\begin{align}
    \dot{x}_p(t) = f(x_p(t)) + g(x_p(t))u(t).
\end{align}
Therefore, the  actuator delay $\delayu$ is compensated by the state predictor~\cite{molnar2022safety}. The CBF design for systems with actuator delay is given as follows.  
\begin{definition}[CBF with actuator delay \cite{molnar2022safety}] 
A continuous function $h:\mathcal{D}\to\mathbb{R}$  is a CBF for \eqref{eq:introCBFdelay x} if there exists an \EKF \ $\alpha$ such that $\forall x\in \mathcal{D}$, we have 
\begin{align}
\sup_{u\in \mathcal{U}} L_{f} h(x_p)+L_{g} h(x_p) u \geq -\alpha(h(x_p)) ,
\end{align}
with $x_p = P(\tau,x(t),u_t)$ being the predicted state by \eqref{eq:introCBFdelay predict}. 
\end{definition}

The CBF guarantees the safety of the system \eqref{eq:introCBFdelay x} as stated in Theorem \ref{theorem:CBF delay}. 
\begin{theorem}[Safety guarantee for system with actuator delay  \cite{molnar2022safety}]\label{theorem:CBF delay}
    If the initial historical input satisfy $P(s,x(0),u_0)\in \mathcal{S},\;\forall s\in[0,\tau]$,  and $h$ is a CBF for \eqref{eq:introCBFdelay x}, then any locally Lipschitz continuous controller $u$ satisfying 
    \begin{align}\label{eq:introCBFdelay CBF constraint}
        L_{f} h(x_p)+L_{g} h(x_p) u \geq -\alpha(h(x_p)),
    \end{align}
    renders the  safe set $\mathcal{S}$ \eqref{eq:introCBF safe set} forward invariant, that is, $x(t)\in \mathcal{S}$ for all $t>0$.
\end{theorem}

\begin{example}[STC for mixed autonomy systems with actuator delay and without disturbances]
For the mixed autonomy system \eqref{eq:system}, if there is no head vehicle ahead of the CAV, then there is no disturbance, i.e., $D=0$.  The system becomes 
\begin{align}
    \dot x(t) = Ax(t) + Bu(t-\delayu),
\end{align}
and the corresponding predictor has a closed-form solution as
\begin{align}\label{eq:introCBFdelay predict traffic}
    x_p(t) = e^{A\tau}x(t) 
     +\int_{-\delayu}^{0} e^{-A\theta} B u(t+\theta) \diff \theta.
\end{align}
Since there is no leading vehicle for the CAV, only HV safety needs to be guaranteed. CBF constraints are designed  following HVs based on Theorem~\ref{theorem:CBF delay} as:
\begin{align}\label{eq:introCBFdelay CBF constraint HV}
        L_{f} h_{i}^{\mathrm{r}} (x_p)+L_{g} h_{i}^{\mathrm{r}}(x_p) u \geq -\alpha(h(x_p)),
\end{align}
with $f(x) = Ax$ and $g(x) = B$. When there is the external disturbance $Dr(t)$ as in \eqref{eq:system}, the predictor~\eqref{eq:introCBFdelay predict traffic}  has prediction error, and the CBF constraints in \eqref{eq:introCBFdelay CBF constraint HV} no longer guarantees safety. In the next section, we will redesign the predictor and   CBF constraints  to ensure safety for the system  \eqref{eq:system}.
\end{example}

\section{Robust safety-critical traffic control under actuator delay, sensor delay, and disturbances} \label{sec:RSTC}

In this section, we design safe constraints and formulate a QP to solve a safety-critical control input for mixed autonomy with actuator delay and disturbances in section \ref{sec:subsec:safety constraints input delay}, and further  incorporate sensor delay in section \ref{sec:subsec:safety constraints input and sensor delay}.

\subsection{Safety under actuator delay and disturbances}\label{sec:subsec:safety constraints input delay}

Since the control input $u(0)$ is actuated to the system until $\delayu$ in the system  \eqref{eq:system}, we make Assumption \ref{assumption:initial safe} on the system state during $[0,\delayu)$.

\begin{assumption}\label{assumption:initial safe}
    The system is safe before the control input $u(0)$ is actuated, i.e., $x(t)\in \mathcal{S}_i$ holds for all $ t\in[0,\delayu) $ and for all $i=0,1,\cdots,N $.
\end{assumption}

For the disturbance $r(t)$, speed of the head HV,  its current value is available to the CAV by onboard sensors. But  its dynamics $\dot r(t)$, acceleration of the head HV, is unknown. To design safety constraints, we make Assumption \ref{assumption:bound rdot} on the disturbance. 
\begin{assumption}\label{assumption:bound rdot}
For the disturbance $r(t)$, its value is known. Its derivative  $\dot{r}(t)$ is unknown, but is bounded by two known bounds $\drlower<0$ and $\drupper>0$ as:
\begin{align}
    \drlower \le \dot r(t) \le \drupper, \quad \forall t\ge 0. \label{eq:bound rdot}
\end{align}
\end{assumption}
Since the derivative $\dot r$ means the acceleration of the head vehicle. Assumption \ref{assumption:bound rdot} thus holds in practice since the acceleration is always bounded considering the existing physical constraint of vehicles. 

For the mixed autonomy model \eqref{eq:system}, the future state  at $t+\delayu$ is determined by: 1) the current state $x(t)$, 2) the historical control input $u$ from $t-\delayu$ to $t$, and 3) the future disturbance $r$ from $t$ to $t+\delayu$ as:
\begin{align}\label{eq:predict x true}
    x(t+\delayu) =& e^{A\delayu}x(t) +\int_{-\delayu}^{0} e^{-A\theta} B u(t+\theta) \diff \theta  \notag \\
     + &\int_{0}^{\delayu} e^{A(\delayu-\theta)} D r(t+\theta) \diff \theta.
\end{align}
Since the future value of disturbance, $r(t+\theta)$ with $\theta\in [0,\delayu)$, is unknown, we design a predictor using current  $r(t)$. The predicted state $x_p(t)$ is
\begin{align}\label{eq:predict x}
    x_p(t) =& e^{A\delayu}x(t) + \int_{-\delayu}^{0} e^{-A\theta} B u(t+\theta) \diff \theta  \notag \\
     + &\int_{0}^{\delayu} e^{A(\delayu-\theta)} D r(t) \diff \theta.
\end{align}
Based on this predictor, we have Theorem \ref{theorem:safety} to guarantee safety of mixed autonomy system.

\begin{theorem}[Safety guarantee for mixed autonomy systems with actuator delay and disturbances]\label{theorem:safety}
Under Assumption \ref{assumption:initial safe} and Assumption \ref{assumption:bound rdot}, for the mixed autonomy system \eqref{eq:system}, if a Lipschitz continuous controller $u(t)$ satisfies 
\begin{align}
    L_fh_0(x_p(t)) + L_gh_0(x_p(t)) u(t)  \ge -\alpha_0( h_0(x_p(t))) + M_0(t), \label{eq:safe constraint CBF CAV theorem}
\end{align}
where $f(x(t)) = Ax(t)$, $g(x(t)) = B$, $x_p\left(t\right)$ is the predicted state value given by \eqref{eq:predict x}, $\alpha_0>0$ is an \EKF, the function $M_0(t):\mathbb{R}^{+}\to \mathbb{R}$ reflects the effect of actuator delay on safety and is defined as  
\begin{align}
    M_0(t) =  & \alpha_0(h_0(x_p(t))) - \alpha_0\left(h_0(x_p(t)) + \frac{1}{2} \drlower \delayu^2\right) \notag\\
      & -\frac{\partial h_0 (x_p(t))}{\partial x_p} D(r(t) + \drlower \delayu) ,\label{eq:M0}
\end{align}
then the safe set for CAV $\mathcal{S}_0$ \eqref{eq:safe set} is forward invariant. If $u$ also satisfies
\begin{align}
    &L_f h_{i}^{\mathrm{r}} (x_p(t)) + L_g h_{i}^{\mathrm{r}} (x_p(t))u(t)   \ge - \alpha_i(h_{i}^{\mathrm{r}} (x_p(t))) + M_i(t),  \label{eq:safe constraint CBF HDV theorem}
\end{align}
where $\alpha_i$ is an \EKF, and the function $M_i(t):\mathbb{R}^{+}\to \mathbb{R}$ is defined as
\begin{align}
    M_i(t) = &\alpha_i(h_{i}^{\mathrm{r}} (x_p(t))) - \alpha_i\left(h_{i}^{\mathrm{r}} (x_p(t)) - \frac{1}{2} \drlower \eta_i \delayu^2\right) \notag \\
    & - \frac{\partial h_{i}^{\mathrm{r}} (  x_p(t))}{\partial  x_p}  D (r(t) + \drupper \delayu )   \label{eq:Mi}, 
\end{align}
then the safe set for HV-$i$ $\mathcal{S}_i$ \eqref{eq:safe set} is forward invariant.

\end{theorem}

\begin{proof}
We prove the theorem in three steps: 
We first prove a bound on the prediction error $x(t+\delayu) - x_p(t)$. Based on the error bound, we then construct a robust safety creation $h_{iR}(x)$ such that  $h_{iR} \left(x_p(t)\right)\ge 0$ guarantees  $h_i \left(x\left(t+\delayu\right)\right)\ge 0$. Then we derive the dynamics of $x_p(t)$ and construct the CBF constraints to ensure $h_{iR} \left(x_p(t)\right)\ge 0$.  The detailed proofs are given in Appendix \ref{sec:appendix proof safety}.
\end{proof}

\begin{remark}[Effect of the actuator delay on CBF  constraints]\label{remark:actuator delay}

We note that CBF constraints in Theorem~\ref{theorem:safety} also work for delay-free systems. When  $\delayu=0$, we have $x_p = x$, and  CBF constraints for CAV~\eqref{eq:safe constraint CBF CAV theorem} and for HV~\eqref{eq:safe constraint CBF HDV theorem} reduce to the CBF designed for delay-free system as in~\eqref{eq:CBFintro CAV} and~\eqref{eq:CBFintro HV} respectively. From the CBF defined in Definition~\ref{def:CBF}, we see that both the function $h$ and its derivative $\dot h$ affect the CBF constraints. The actuator delay affects CBF constraints also from these two sides. Take the CAV safety constraint as an example, we see that since there is a prediction error on the future state, we ensure  safety  of a conservative $h_0(x_p) + \drlower\delayu^2/2$ instead of the original $h_0(x_p)$.  And when calculating $\dot h_0 
(x_p)$, the term $\drupper \delayu$ is introduced due to unknown future disturbance value over the delay interval.
\end{remark}

\begin{example}[RSTC for mixed autonomy systems with actuator delay and disturbances]
    Theorem \ref{theorem:safety} gives the constraints on the control input to ensure safety for mixed traffic. Similar to \eqref{eq:QP CBFintro traffic} in Example~\ref{example:STC}, given a nominal controller $u_0$, we can formulate a QP  with $N+1$ constraints to solve a safe control input: 
    \begin{align}
    & u = \underset{u \in \mathbb{R}, \sigma_i \geq 0}{\operatorname{argmin}} \;  |u-u_0|^{2} + \sum_{i=1}^N p_i\sigma_i^2 \label{eq:QP}  \\
        &\text{s.t.} \notag \\ 
        & L_fh_0(x_p) + L_gh_0(x_p) u   +  \alpha_0 (h_0(x_p) )+ M_0 \ge 0,  \notag \\
        & L_f h_{1}^{\mathrm{r}} (x_p) +L_g h_{1}^{\mathrm{r}} (x_p)u   + \alpha_1 (h_{1}^{\mathrm{r}} (x_p)) + M_1 +\sigma_1\ge 0, \notag\\
        & \qquad \vdots \notag\\
        & L_f h_{N}^{\mathrm{r}} (x_p) +L_g h_{N}^{\mathrm{r}} (x_p)u   + \alpha_N (h_{N}^{\mathrm{r}} (x_p)) + M_N +\sigma_N\ge 0 \notag,
    \end{align}
    where $x_p = x_p(t)$ is the predicted state by~\eqref{eq:predict x}.
\end{example}

\subsection{Safety under actuator delay, sensor delay, and disturbances}\label{sec:subsec:safety constraints input and sensor delay}

For the system \eqref{eq:system} with delayed measurement  in \eqref{eq:yc general} as
\begin{align*}
    y(t) =  \sum_{j=1}^{J} C_j x(t-\tau_{y,j}),
\end{align*}
a predictor   observer is designed to obtain the estimated state $\hat x$ as:
\begin{align}\label{eq:observer}
    \dot{\hat{x}} (t)  =&  A \hat x(t) + Bu(t-\delayu) + Dr(t) \notag \\
    &+ L\left(Y(t) - \bar{C} \hat x(t)\right), 
\end{align}
where $L\in \mathbb{R}^{n\times 
 n_y}$ is the gain matrix to be designed, the vector $Y(t) \in \mathbb{R}^{n_y}$ is given by 
\begin{align}
    Y(t) = &y(t) + \sum_{j=1}^{J} C_j e^{-A\tau_{y,j}} \!\! \int_{-\delayu-\tau_{y,j}}^{-\delayu} \!\! e^{-A(\theta+\delayu)}Bu(t+\theta)\diff \theta \notag \\
     + & \sum_{j=1}^{J} C_j e^{-A\tau_{y,j}}  \int_{-\tau_{y,j}}^{0} e^{-A\theta}Dr(t+\theta)\diff \theta, \label{eq:Ytransformed}
\end{align}
and matrix $\bar{C}\in  \mathbb{R}^{n_y\times n}$ is defined as:
\begin{align}\label{eq:Cbar}
    \bar C = \sum_{j=1}^{J} C_j e^{-A\tau_{y,j}}.
\end{align}
Since $Y(t)$ defined in \eqref{eq:Ytransformed} involves the historical value of $u(t)$ and $r(t)$, to analyze the stability of the estimation error
\begin{align}
    \epsilon(t) = x(t) - \hat{x}(t),
\end{align}
we first make the following assumption.

\begin{assumption}\label{assumption:initial estimation error}
    The historical estimation error $\epsilon(t) \in L_{\infty}[-\delayu-\delaysensor,0]$, which means the historical estimation error $\epsilon(t)$ is finite with $t\in [-\delayu-\delaysensor,0]$.
\end{assumption}

We have Theorem \ref{theorem:observer convergence} for the estimation error of the designed observer.

\begin{theorem}[Convergence of the predictor-observer] \label{theorem:observer convergence} 
Under Assumption~\ref{assumption:initial estimation error}, for the system \eqref{eq:system} with measurement  \eqref{eq:yc general}, if the gain matrix $L$ in the  observer  \eqref{eq:observer}  is chosen so that $A-L\bar{C}$ is Hurwitz, then for the estimation error $\epsilon$, the equilibrium $\epsilon = 0$ is exponentially stable in the sense that there exists $\Upsilon>$ and $\lambda>0$ such that
\begin{align}
    \Vert \epsilon(t) \Vert \le \Upsilon \Vert \epsilon(0)  \Vert e^{-\lambda t}.
\end{align}

\end{theorem}

\begin{proof}
By the system model \eqref{eq:system}, the system state $x(t)$ is decided by historical state at $x(t-\tau_{y,j})$, historical input $u$ from $t-\delayu-\tau_{y,j}$ to $t-\delayu$, and historical disturbance $r$ from $t-\tau_{y,j}$ to $t$ as:
\begin{align}
    x(t) =& e^{A\tau_{y,j}}x(t-\tau_{y,j}) + \int_{-\delayu-\tau_{y,j}}^{-\delayu} e^{-A(\theta+\delayu)}Bu(t+\theta)\diff \theta \notag \\
    & + \int_{-\tau_{y,j}}^{0} e^{-A\theta}Dr(t+\theta)\diff \theta.
\end{align}
The measurement $y(t)$ is re-written as
\begin{align}
    y(t) = &\sum_{j=1}^{J} C_j e^{-A\tau_{y,j}} x(t)   \notag \\
    & - \sum_{j=1}^{J} C_j e^{-A\tau_{y,j}} \int_{-\delayu-\tau_{y,j}}^{-\delayu} e^{-A(\theta+\delayu)}Bu(t+\theta)\diff \theta \notag \\
    & - \sum_{j=1}^{J} C_j e^{-A\tau_{y,j}}  \int_{-\tau_{y,j}}^{0} e^{-A\theta}Dr(t+\theta)\diff \theta.
\end{align}
Given $Y$ defined as~\eqref{eq:Ytransformed} and $\bar{C}$ defined as~\eqref{eq:Cbar}, we have 
\begin{align}\label{eq:yCbar}
    Y(t) & = \bar{C} x(t).
\end{align}
For the estimation error, we have
\begin{align}
    \dot \epsilon(t) = \left(A-L\bar{C}\right) \epsilon(t).
\end{align}
Since $A - L\bar{C}$ is Hurwitz, the equilibrium $\epsilon=0$ is exponentially stable. 
\end{proof}

Based on the estimated state $\hat{x}$ by the observer \eqref{eq:observer}, we design a predictor as
\begin{align}\label{eq:predict x observer}
    \hat x_p(t) =& e^{A\delayu} \hat x(t) +\int_{-\delayu}^{0} e^{-A\theta} B u(t+\theta) \diff \theta  \notag \\
     &+ \int_{0}^{\delayu} e^{A(\delayu-\theta)} D r(t) \diff \theta,
\end{align}
where $\hat{x}_p(t)$ predicts the system state at $t+\delayu$. We have Theorem \ref{theorem:safety observer} for the safety constraints of the system.

\begin{theorem}[Safety guarantee for mixed autonomy systems with actuator delay,  sensor delay and disturbances] \label{theorem:safety observer}
Under Assumptions \ref{assumption:initial safe}-\ref{assumption:initial estimation error}, for the mixed autonomy system~\eqref{eq:system}, the safe set for CAV $\mathcal{S}_0$~\eqref{eq:safe set} is forward invariant if a Lipschitz continuous controller $u$ satisfies
\begin{align}\label{eq:safe constraint CBF CAV theorem observer}
    & L_fh_0(\hat x_p(t)) + L_gh_0(\hat x_p(t)) \notag \\
    &\ge -\alpha_0( h_0(\hat x_p(t))) + \hat M_0(t) + Z_0(t),
\end{align}
where $\hat x_p \left(t\right)$ is the predicted state given by \eqref{eq:predict x observer}, $\alpha_0$ is an \EKF, the function $\hat M_0(t): \mathbb{R}^{+} \to \mathbb{R}$  is defined similarly to  $M_0(t)$ in \eqref{eq:M0} as:
\begin{align}
    \hat M_0(t) = &\alpha_0(h_0(\hat x_p(t))) -\alpha_0\left(h_0(\hat x_p(t)) +\frac{1}{2} \drlower\delayu^2\right) \notag \\
    &-\frac{\partial h_0(\hat x_p(t))}{\partial \hat x_p} D(r(t) + \drlower \delayu),
\end{align}
the function $Z_0(t): \mathbb{R}^{+} \to \mathbb{R}$ reflects the effect of observer error on safety and is defined as 
\begin{align}
    Z_0(t) =& \alpha_0\left(h_0(\hat x_p(t)) +\frac{1}{2} \drlower\delayu^2\right) \notag \\
    &-  \alpha_0\left(h_0(\hat x_p(t)) +\frac{1}{2} \drlower\delayu^2  - (1+\psi_0) \Vert e^{A\delayu}\Vert \Upsilon \bar{\epsilon} e^{-\lambda t}\right) \notag \\
    &- \frac{\partial h_{0}(\hat x_p(t))}{\partial \hat x_p}   e^{A\delayu}L(Y(t)-\bar{C}\hat x(t) ) \notag \\
    &- \lambda (1+\psi_0) \Vert e^{A\delayu}\Vert \Upsilon \bar{\epsilon} e^{-\lambda t}, \label{eq:Z0}
\end{align}
with   
$\Upsilon$ and $\lambda$  given in Theorem~\ref{theorem:observer convergence}, and $\bar{\epsilon}$ being an upper bound on the  initial estimation error $\Vert \epsilon(0) \Vert $.
If $u$ also satisfies 
\begin{align}
    &L_f h_{i}^{\mathrm{r}} \left(\hat x_p(t)\right) + L_g h_{i}^{\mathrm{r}} \left(\hat x_p(t)\right) u(t) \notag \\
    & \ge - \alpha_i (h_{i}^{\mathrm{r}} \left(\hat x_p(t)\right)) + \hat M_i(t) + Z_i(t) , \label{eq:safe constraint CBF HDV theorem observer}
\end{align}
with $\alpha_i>0$ being an \EKF, the function $\hat M_i(t): \mathbb{R}^{+}\to \mathbb{R}$ being
\begin{align}
    \hat M_i(t) = &\alpha_i(h_{i}^{\mathrm{r}} (\hat x_p(t))) - \alpha_i\left(h_{i}^{\mathrm{r}} (\hat x_p(t)) - \frac{1}{2} \drlower \eta_i\delayu^2\right) \notag \\
    & - \frac{\partial h_{i}^{\mathrm{r}} (  \hat x_p(t))}{\partial  \hat x_p} D (r(t) + \drupper \delayu ), 
\end{align}
and the function $Z_i(t):  \mathbb{R}^{+}\to \mathbb{R}$ being:
\begin{align}\label{eq:Zi}
    Z_i(t) &=  \alpha_i\left(h_{i}^{\mathrm{r}} (\hat x_p(t)) - \frac{1}{2} \drlower \eta_i \delayu^2\right) \notag \\
    &-  \alpha_i\left(h_{i}^{\mathrm{r}} (\hat x_p(t)) -  \frac{1}{2} \drlower \eta_i \delayu^2       - \nu_i \Vert e^{A\delayu}\Vert \Upsilon \bar{\epsilon} e^{-\lambda t}\right) \notag \\
    & - \frac{\partial h_{i}^{\mathrm{r}} (\hat x_p(t))}{\partial \hat x_p}   e^{A\delayu}L(Y(t)-\bar{C}\hat x(t) ) \notag \\
    & - \lambda \nu_i \Vert e^{A\delayu}\Vert \Upsilon \bar{\epsilon} e^{-\lambda t},
\end{align}
with $\nu_i =1-\eta_i + \psi_i - \eta_i\psi_0$, then the safe set for HV-$i$ $\mathcal{S}_i$ \eqref{eq:safe set} is also forward invariant.
\end{theorem}

\begin{proof} We also follow the idea and the three steps in the proof for Theorem~\ref{theorem:safety}. The detailed proof is in Appendix~\ref{sec:appendix proof safety observer}.
\end{proof}

\begin{remark}[Comparison of effects of actuator and sensor delays on CBF constraints]

For the two types of delays in practical traffic, actuator delay $\delayu$ and sensor delay $\delaysensor$, their effects on CBF constraints are reflected by $M_i(t)$  and $Z_i(t)$, respectively. 
We note that  $Z_i(t) \to 0 $  as $t\to \infty $, which indicates that the effect of partial feedback and sensor delay wears off. This is because the estimation error by the designed observer~\eqref{eq:observer} converges to zero as stated in Theorem~\ref{theorem:observer convergence}.  For the actuator delay, on the other side, it will have a permanent effect on CBF constraints as the term related to $\delayu$ in $M_i(t)$ remains constant over time. This is because we assume that the accurate value of $\dot r$, the acceleration of the head vehicle, is always unknown and can only be   bounded by fixed bounds $\drlower$ and $\drupper$  in Assumption~\ref{assumption:bound rdot}.
\end{remark}

\begin{example}[RSTC for mixed autonomy systems with actuator delay, sensor delay, and disturbances]
    Based on the safety constraints in Theorem \ref{theorem:safety observer}, we formulate a QP as \eqref{eq:QP observer} to solve a safety-critical control input.
    \begin{align}
        & u = \underset{u \in \mathbb{R}, \sigma_i \geq 0}{\operatorname{argmin}} \;  |u-u_0|^{2} + \sum_{i=1}^N p_i\sigma_i^2 \label{eq:QP observer}\\
        &\quad \text{s.t.} \notag\\
        & L_fh_0(\hat x_p) + L_gh_0(\hat x_p) u   +  \alpha_0 (h_0(\hat x_p)) + \hat{M}_0 -Z_0 \ge 0  \notag \\
        & L_f h_{1}^{\mathrm{r}} (\hat x_p)  + L_g h_{1}^{\mathrm{r}} (\hat x_p)u  +\alpha_1 (h_{1}^{\mathrm{r}} (\hat x_p)) +\hat  M_1 -Z_1 +\sigma_1 \ge 0 \notag\\
        & \qquad \vdots \notag\\
        & L_f h_{N}^{\mathrm{r}} (\hat x_p) \! + \! L_g h_{N}^{\mathrm{r}} (\hat x_p)u   \!  + \! \alpha_N( h_{N}^{\mathrm{r}} (\hat x_p)) \! + \! \hat M_N \! - \! Z_N \! +\sigma_N \! \ge \! 0 .\notag 
    \end{align}
\end{example}

\section{Numerical Simulation} \label{sec:simulation}

In this section, we run numerical simulations to validate the safety guarantee of the proposed RSTC. We first specify the simulation settings in section \ref{sec:subsec:simulation setting}. In section \ref{sec:subsec:simulation main result}, we show that the RSTC avoids rear-end collisions in safety-critical scenarios. And we further analyze the properties of RSTC in section \ref{sec:subsec:simulation analysis}.

\begin{figure*}[t!]
    \centering
    Nominal Controller \\
    \includegraphics[width=0.24\linewidth]{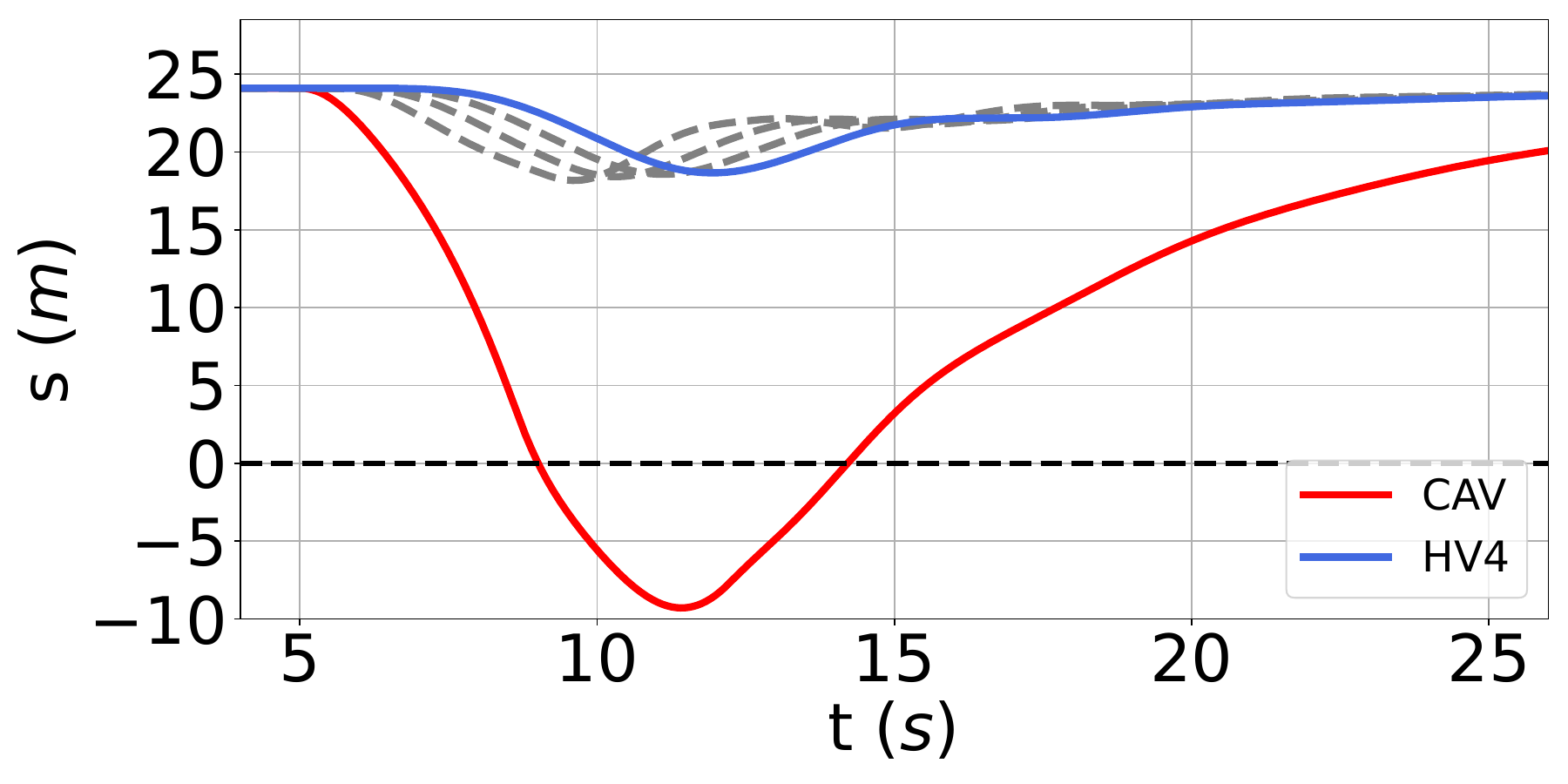}
    \includegraphics[width=0.24\linewidth]{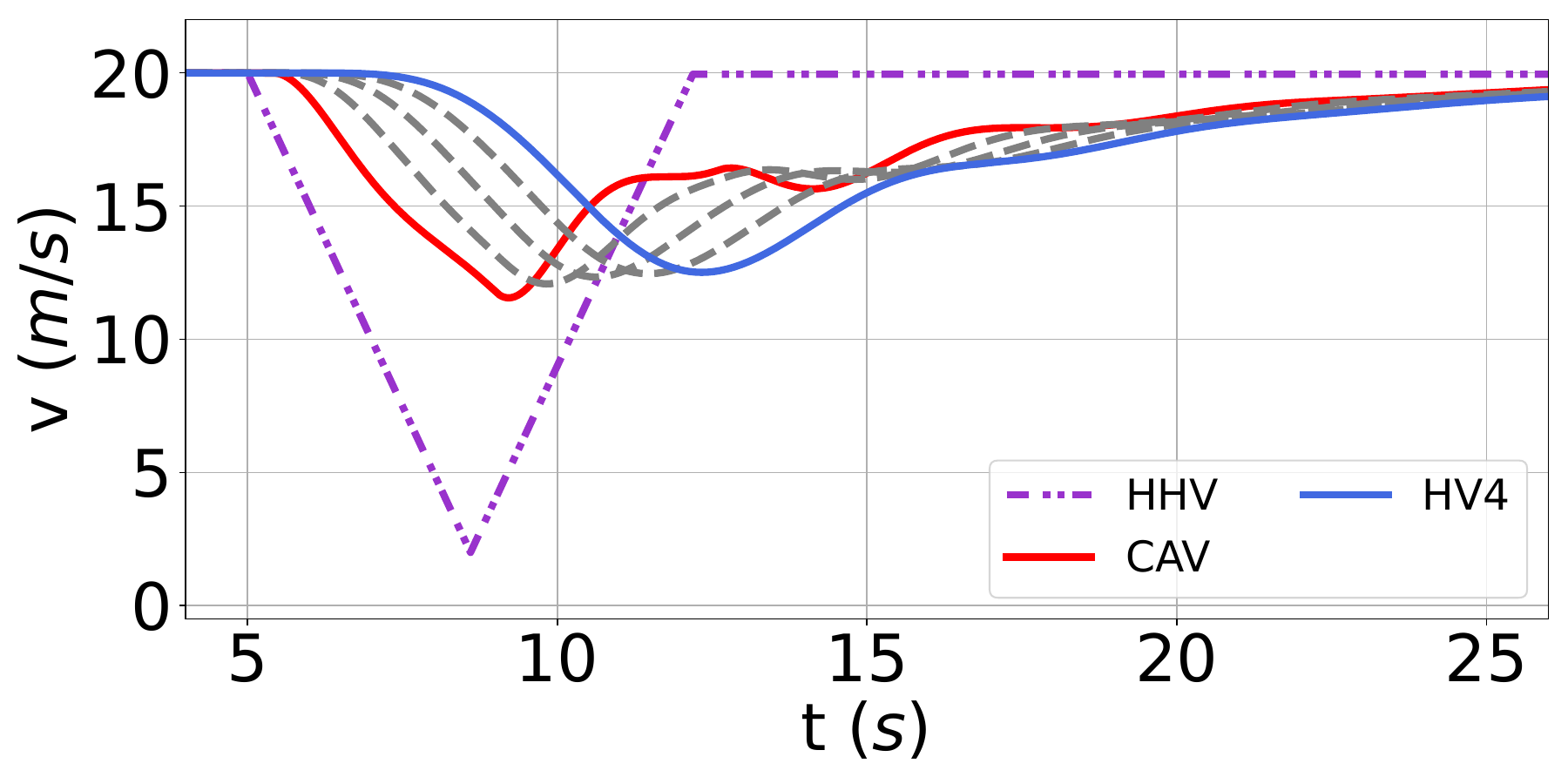}
    \includegraphics[width=0.24\linewidth]{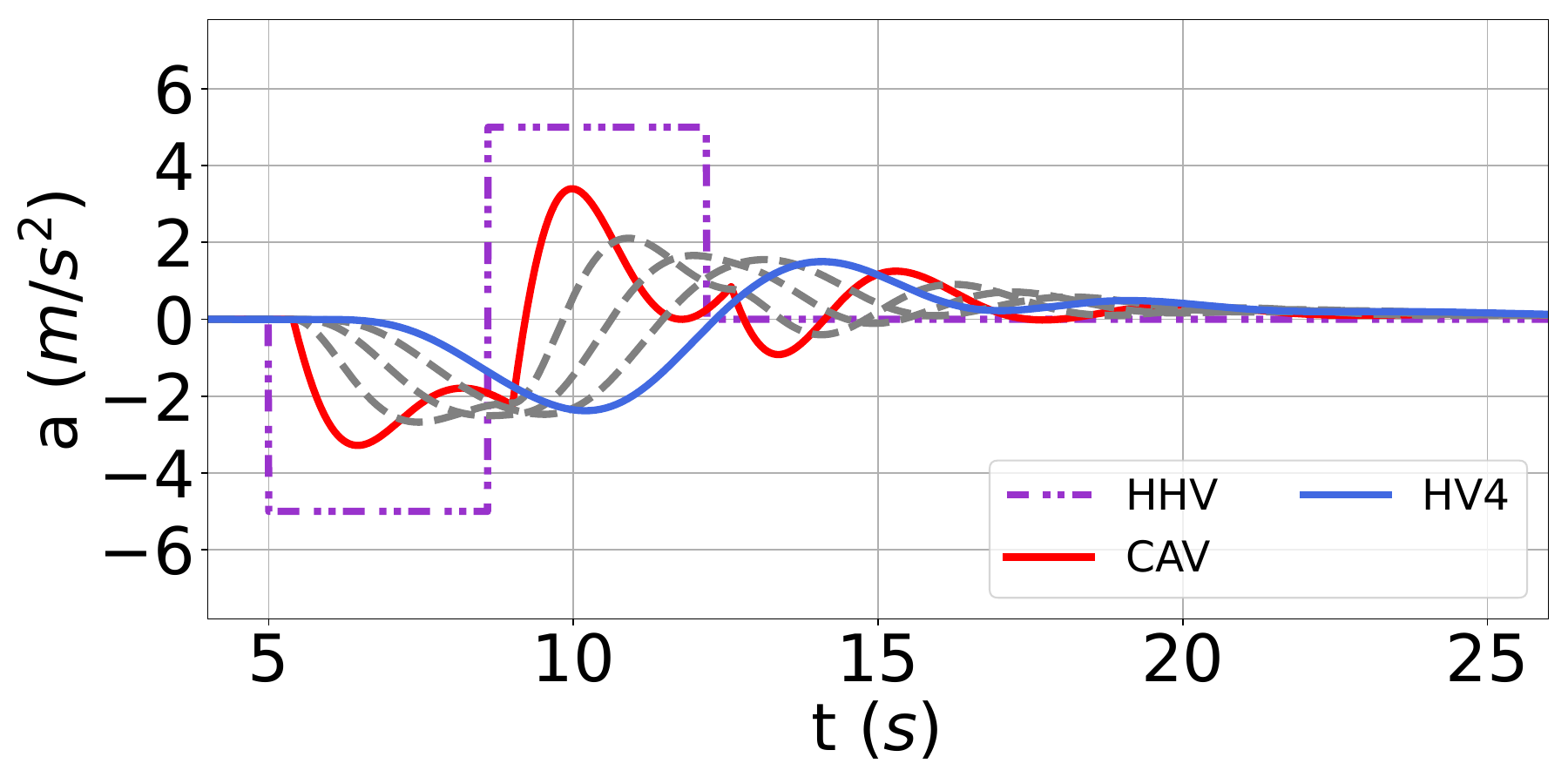}
    \includegraphics[width=0.24\linewidth]{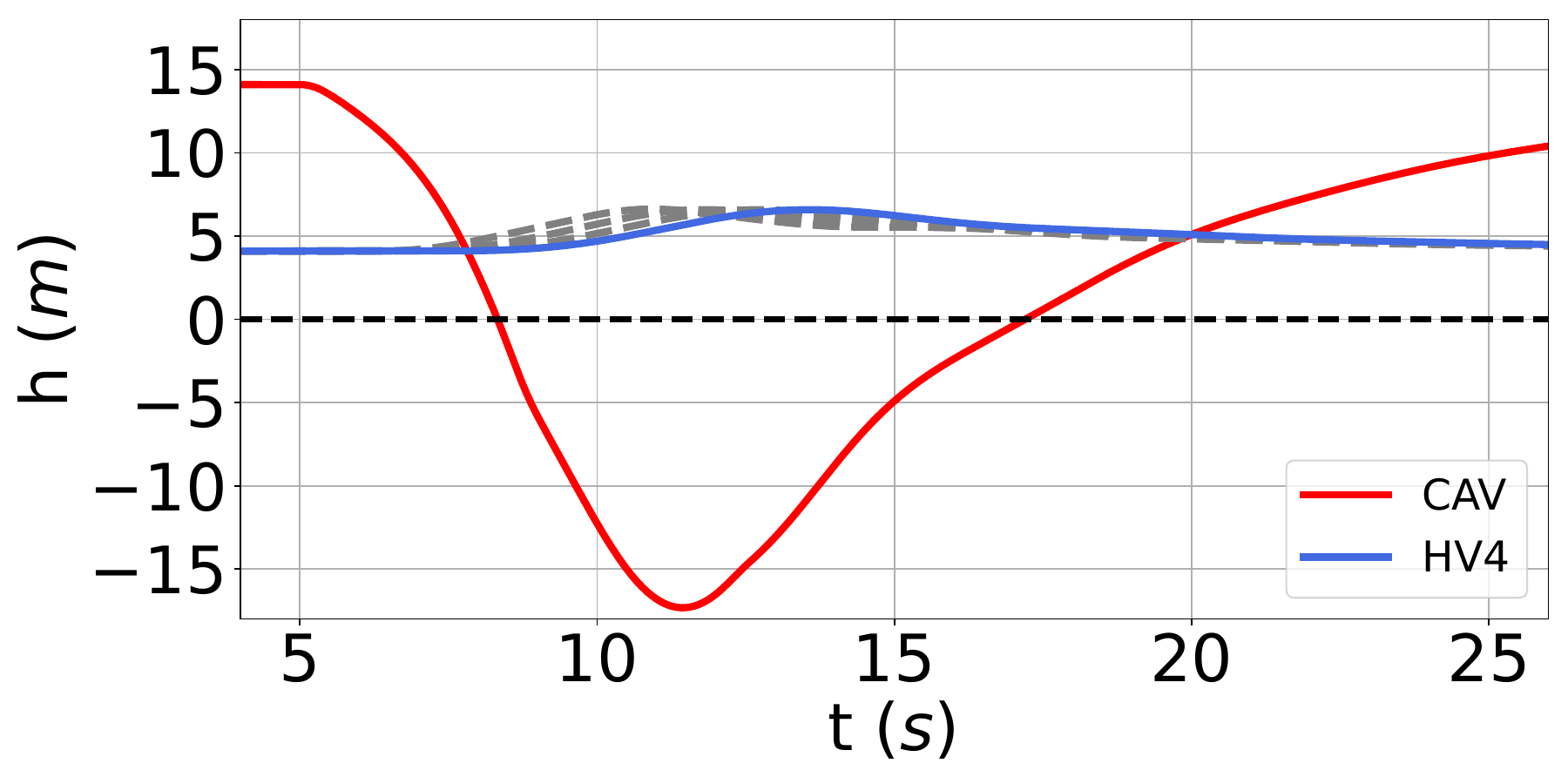}
    \\
    RSTC \\
    \includegraphics[width=0.24\linewidth]{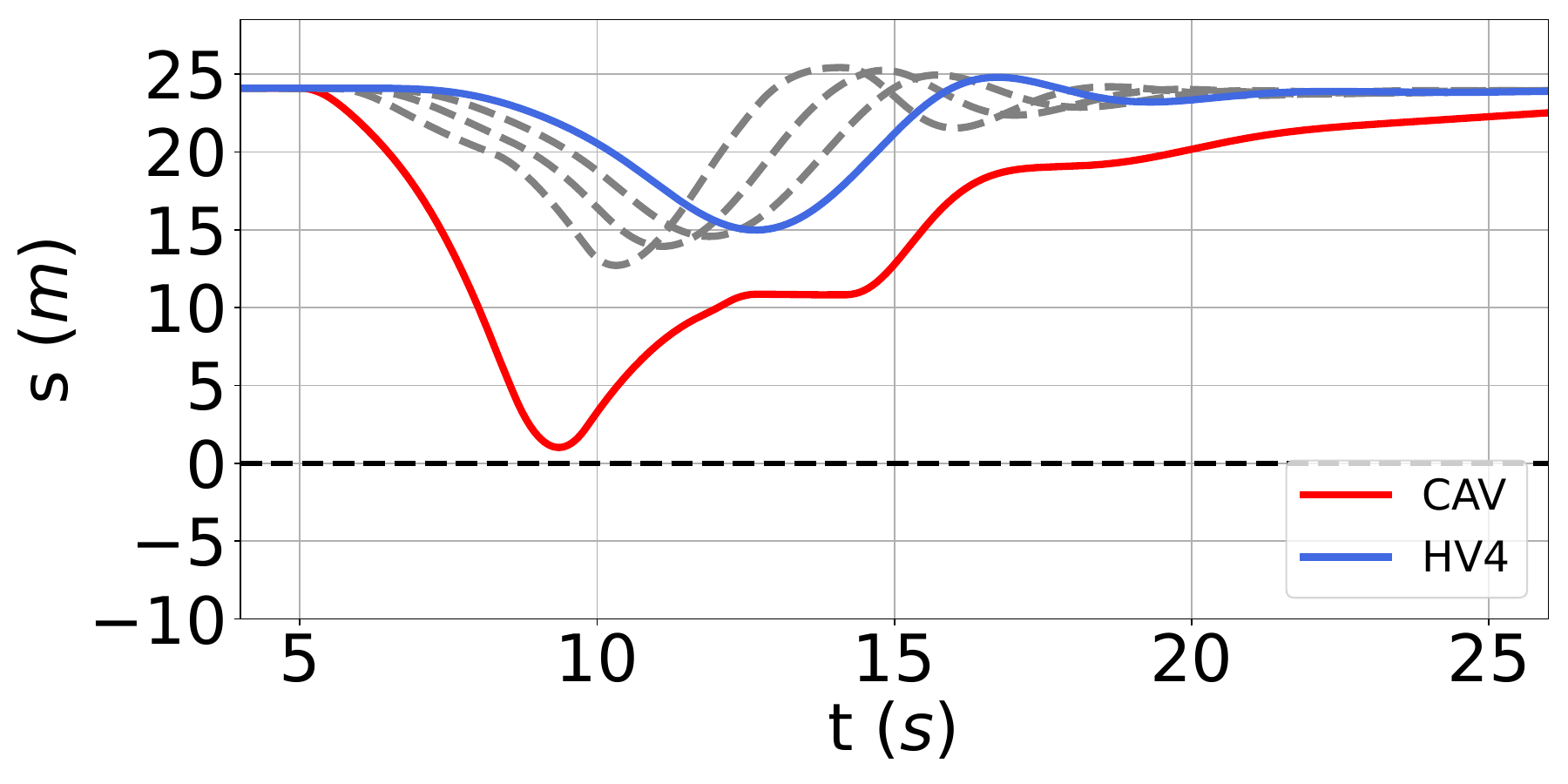}
    \includegraphics[width=0.24\linewidth]{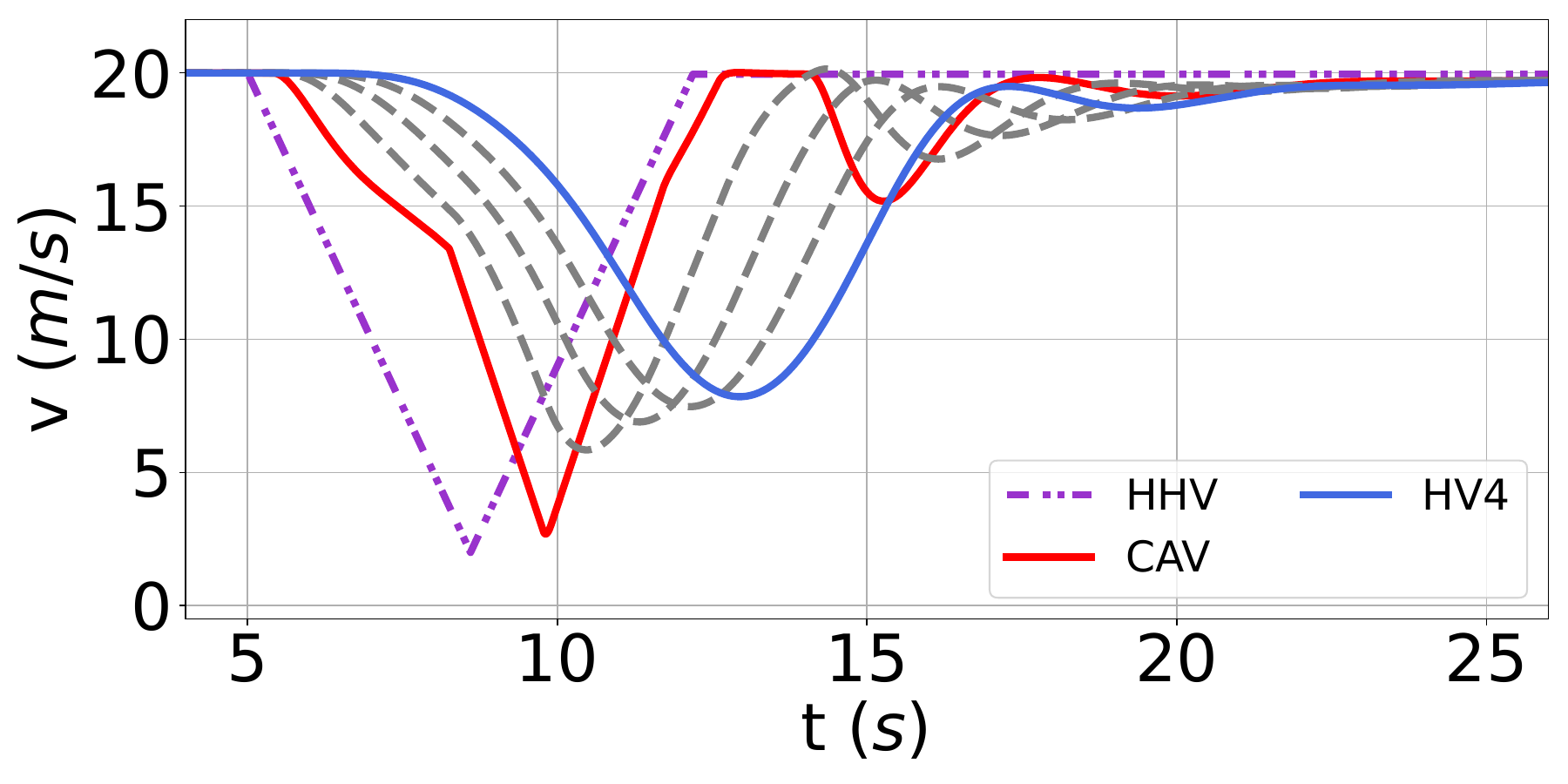}
    \includegraphics[width=0.24\linewidth]{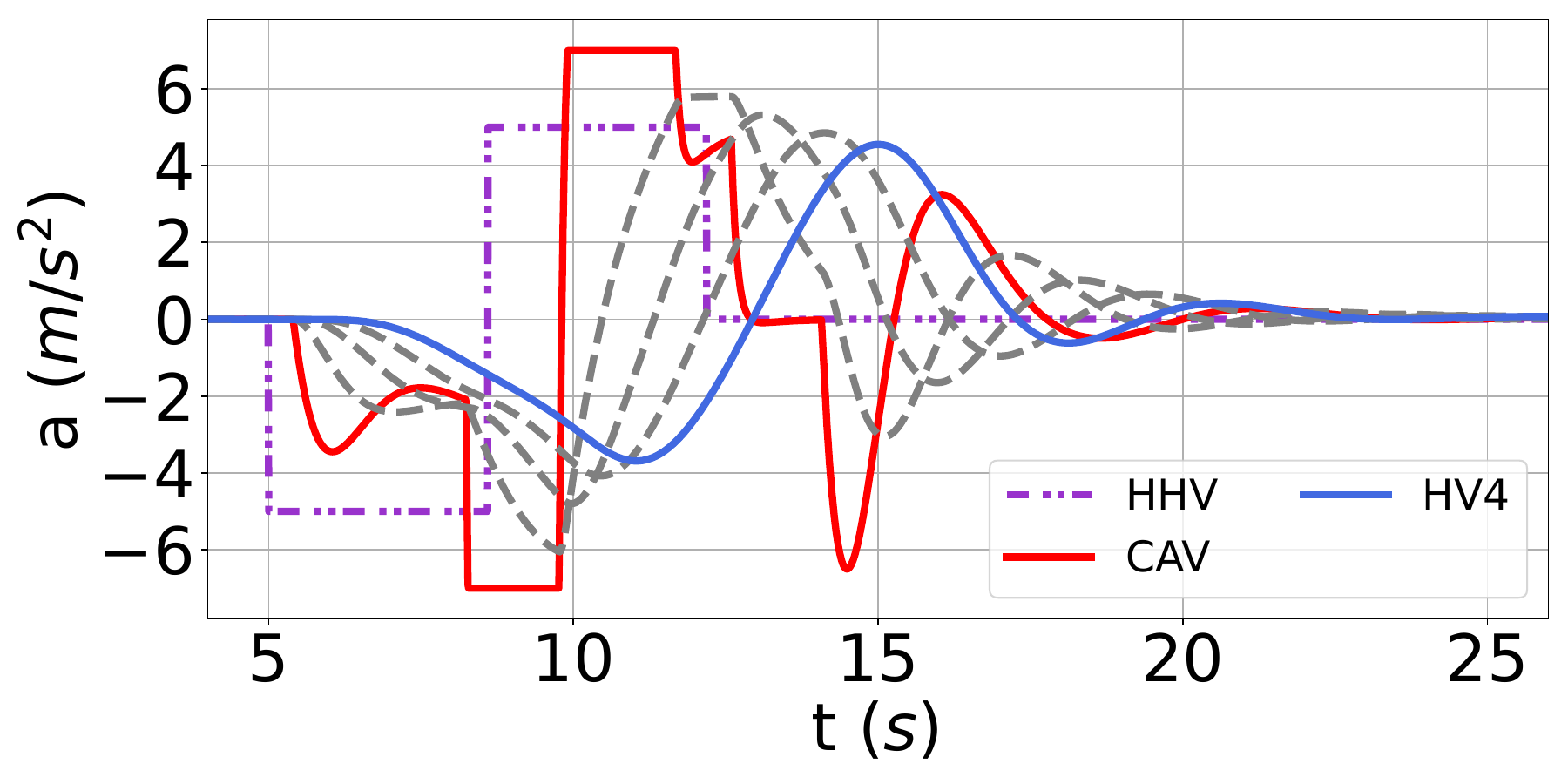}
    \includegraphics[width=0.24\linewidth]{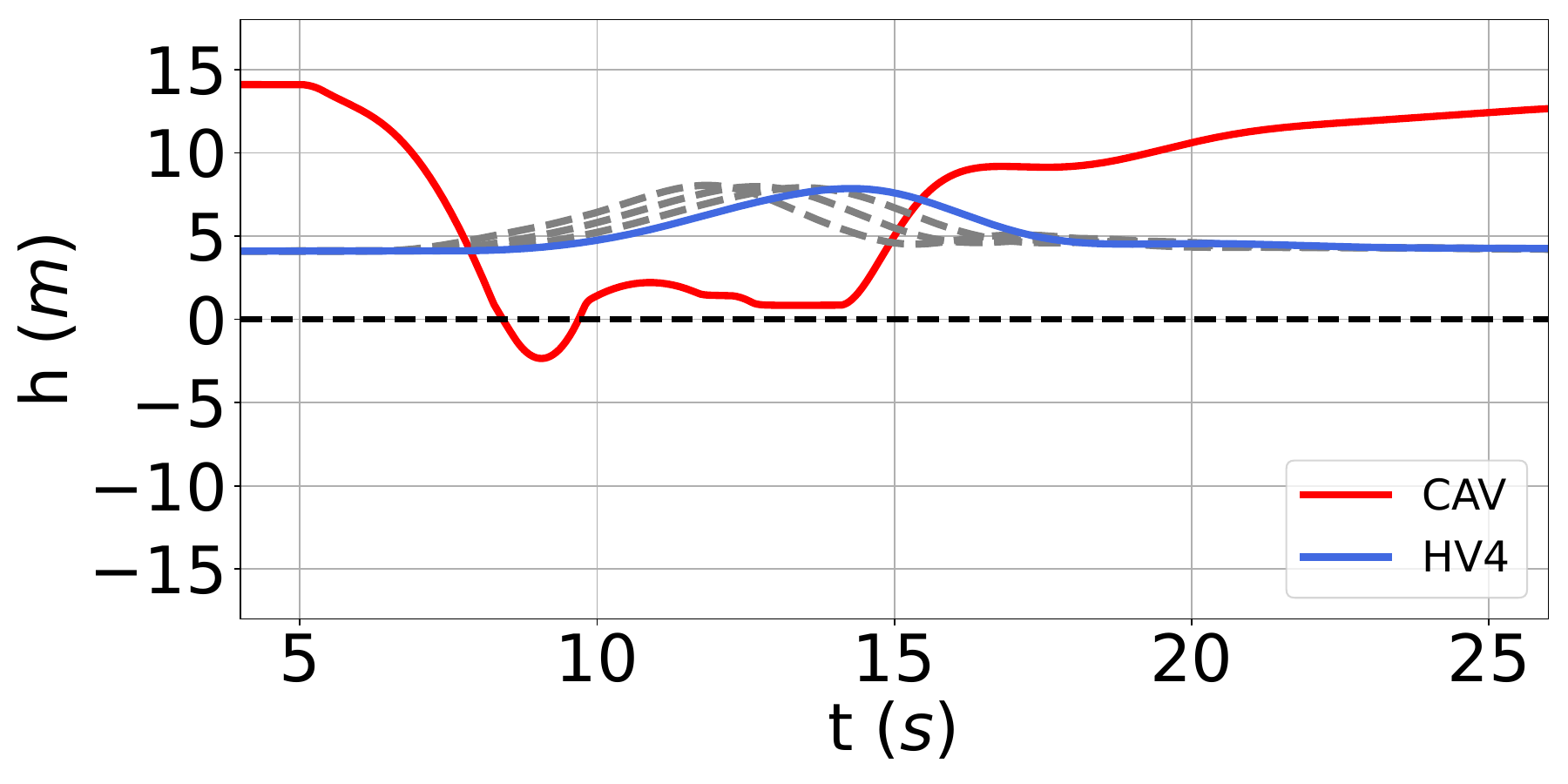}
    \caption{Numerical simulation of RSTC in Scenario 1. The first row gives the trajectory by the nominal controller \eqref{eq:nominal controller}, which stabilizes the traffic but causes a collision. The second row gives the trajectory by the proposed RSTC \eqref{eq:QP}, which guarantees safety.}
    \label{fig:trajectory scenario 1}
\end{figure*}

\begin{figure*}[t!]
    \centering
    Nominal Controller \\
    \includegraphics[width=0.24\linewidth]{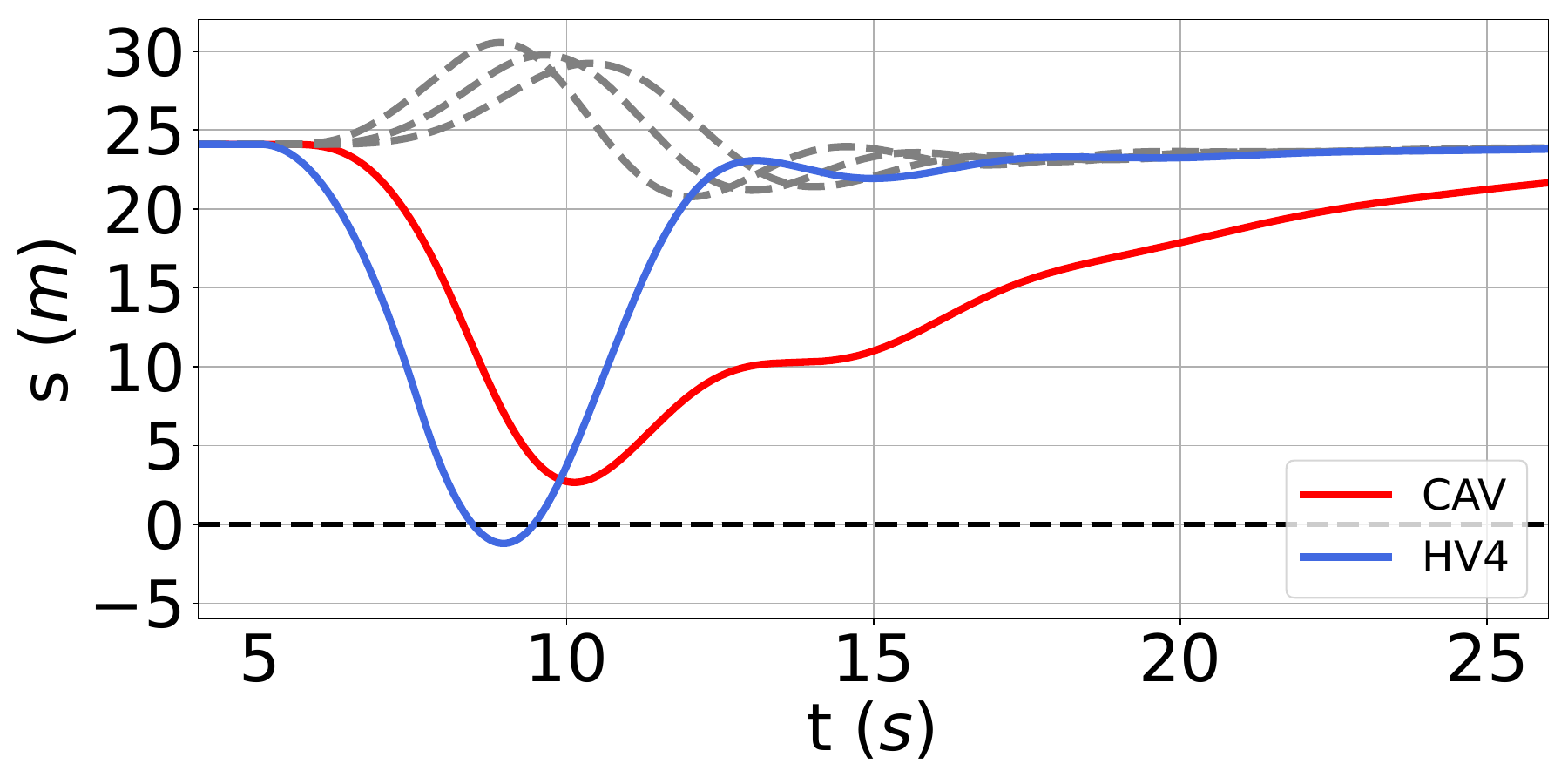}
    \includegraphics[width=0.24\linewidth]{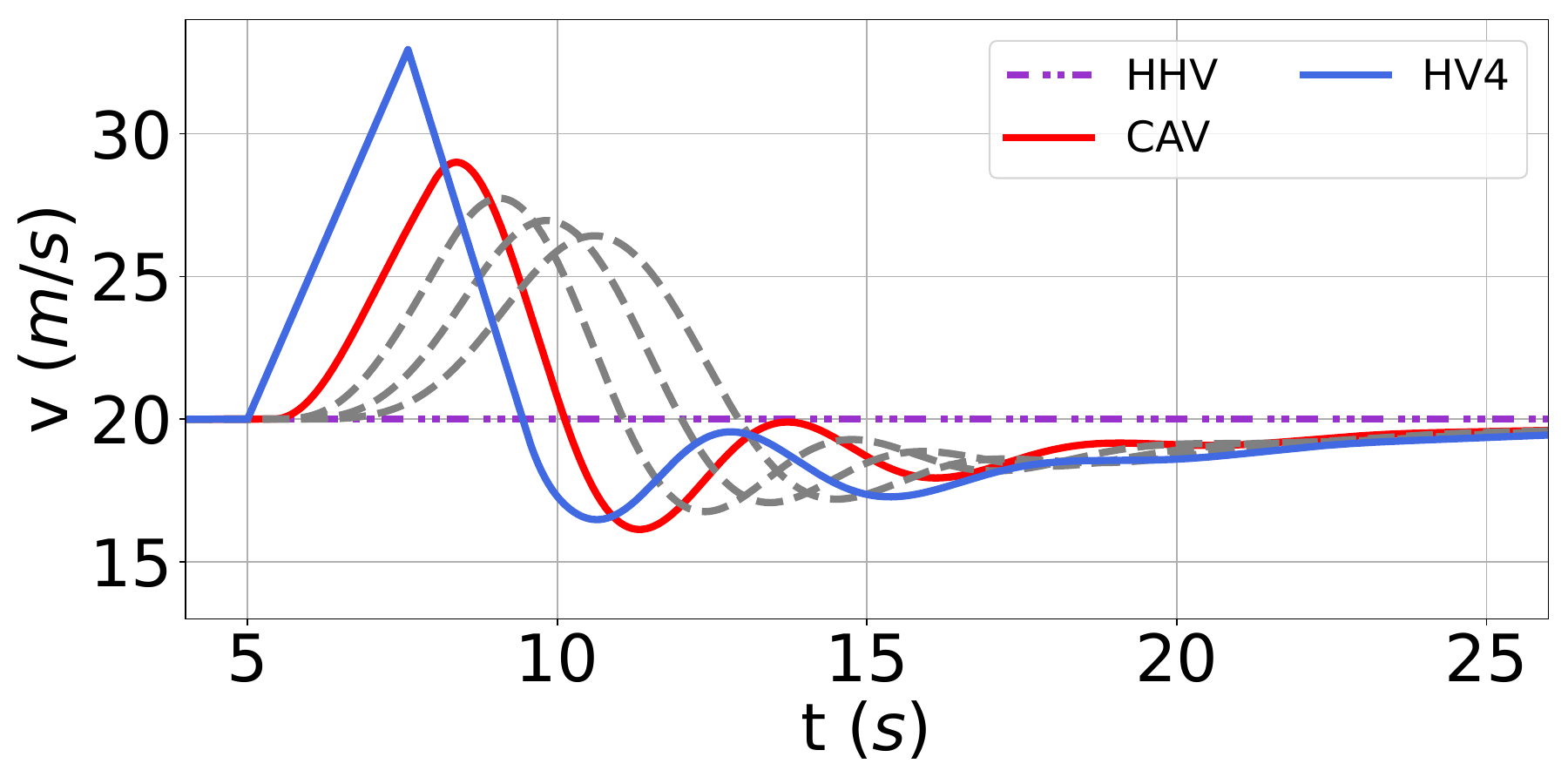}
    \includegraphics[width=0.24\linewidth]{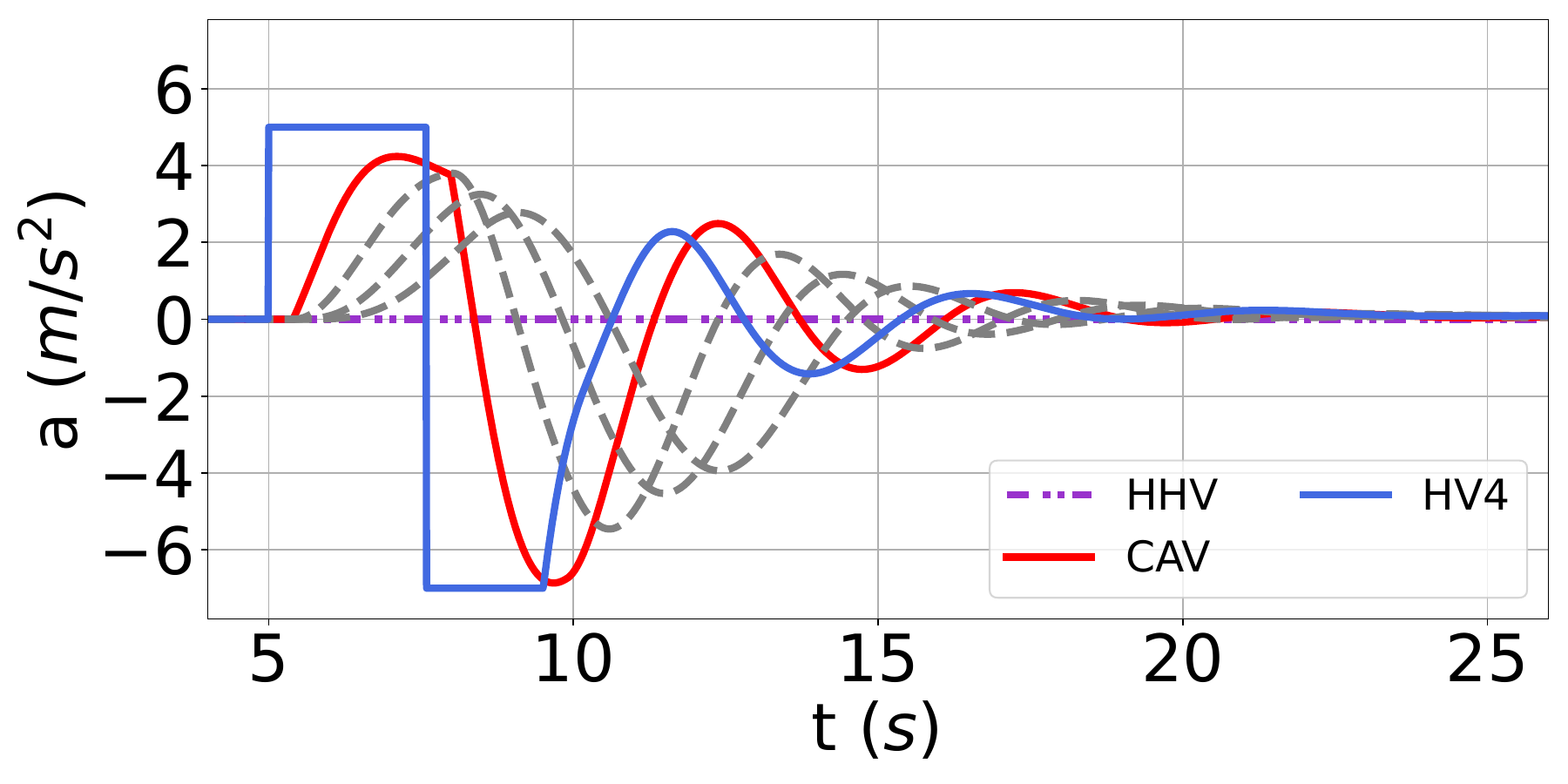}
    \includegraphics[width=0.24\linewidth]{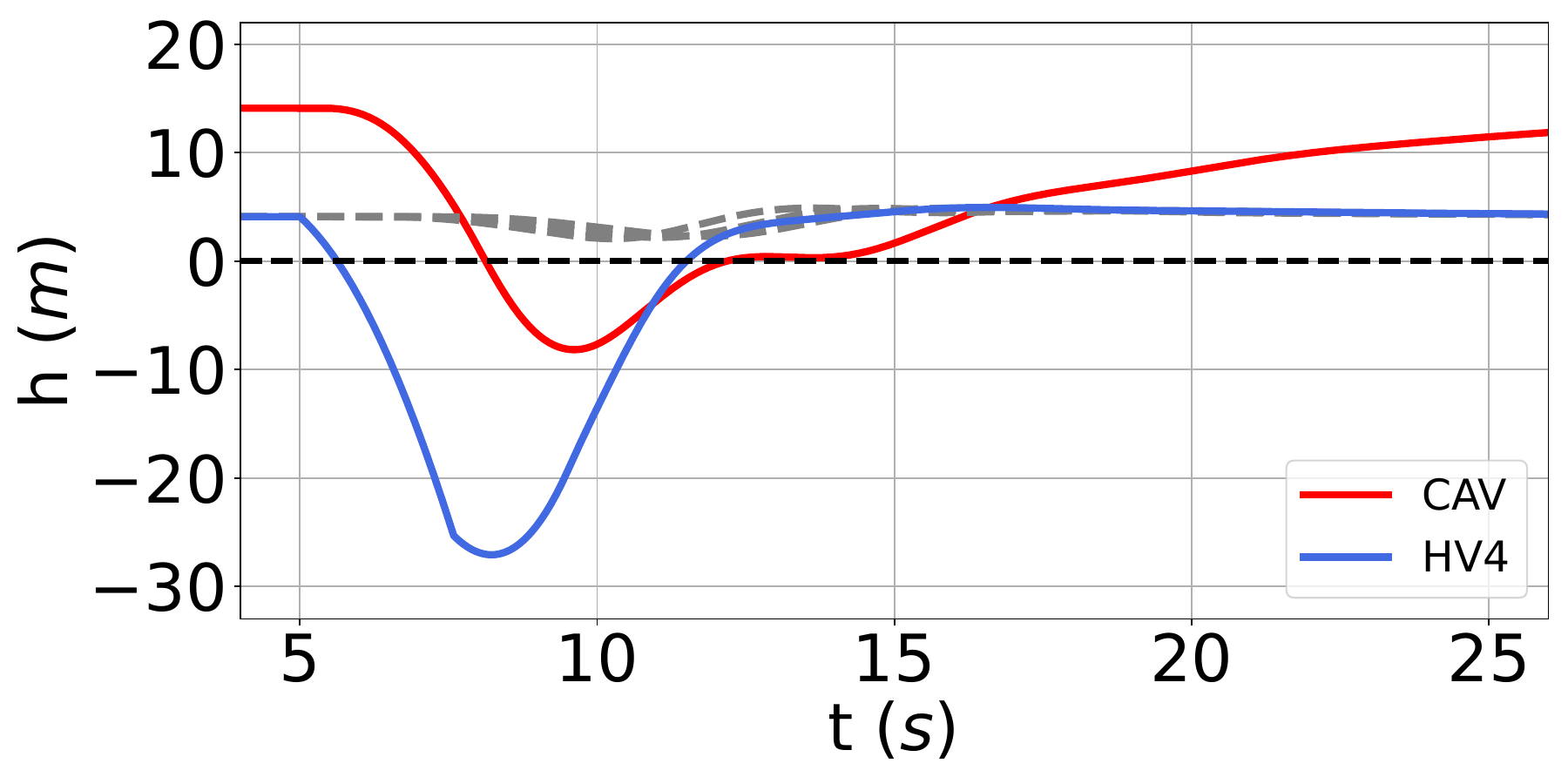}
    \\
    RSTC \\
    \includegraphics[width=0.24\linewidth]{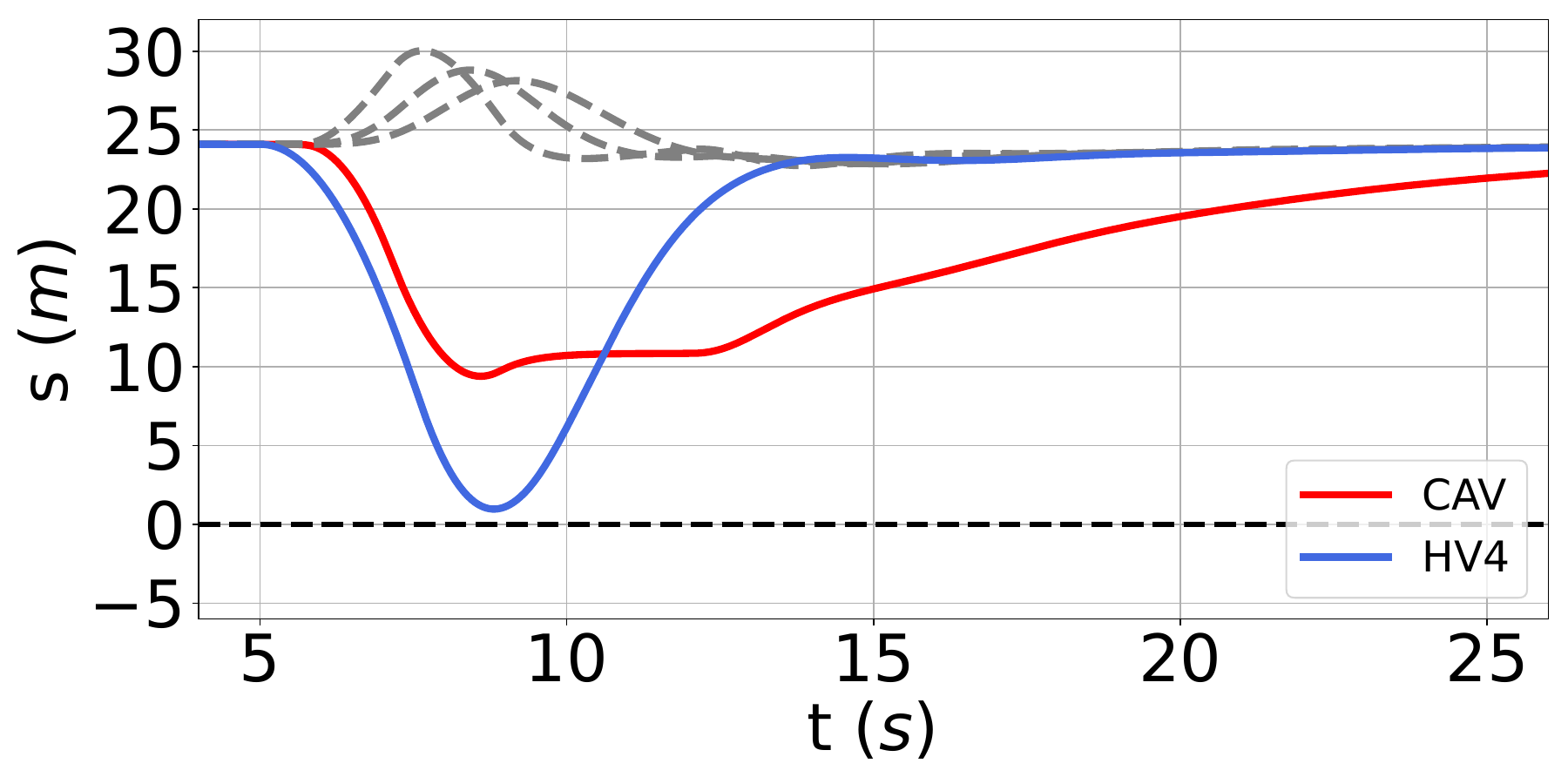}
    \includegraphics[width=0.24\linewidth]{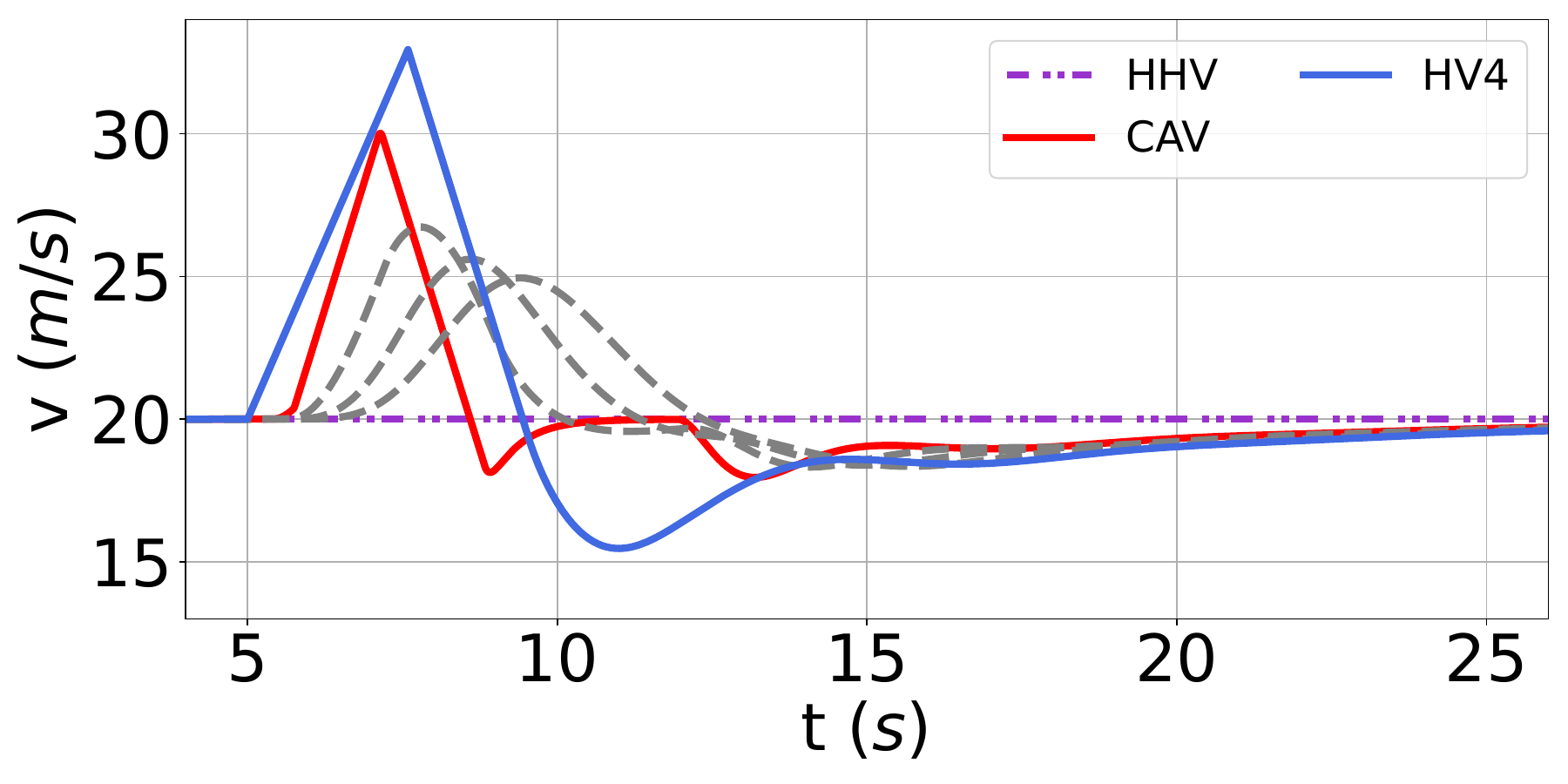}
    \includegraphics[width=0.24\linewidth]{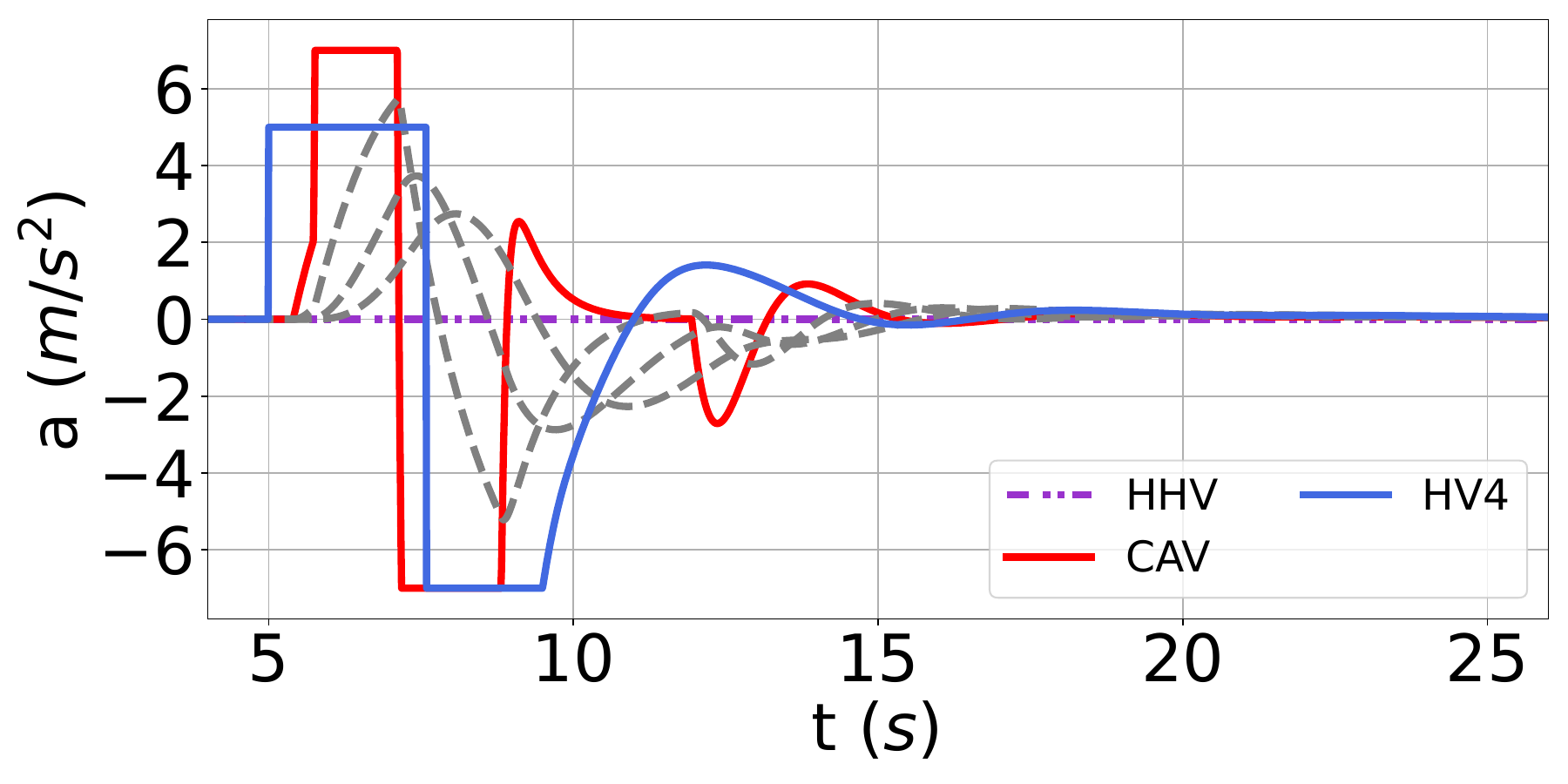}
    \includegraphics[width=0.24\linewidth]{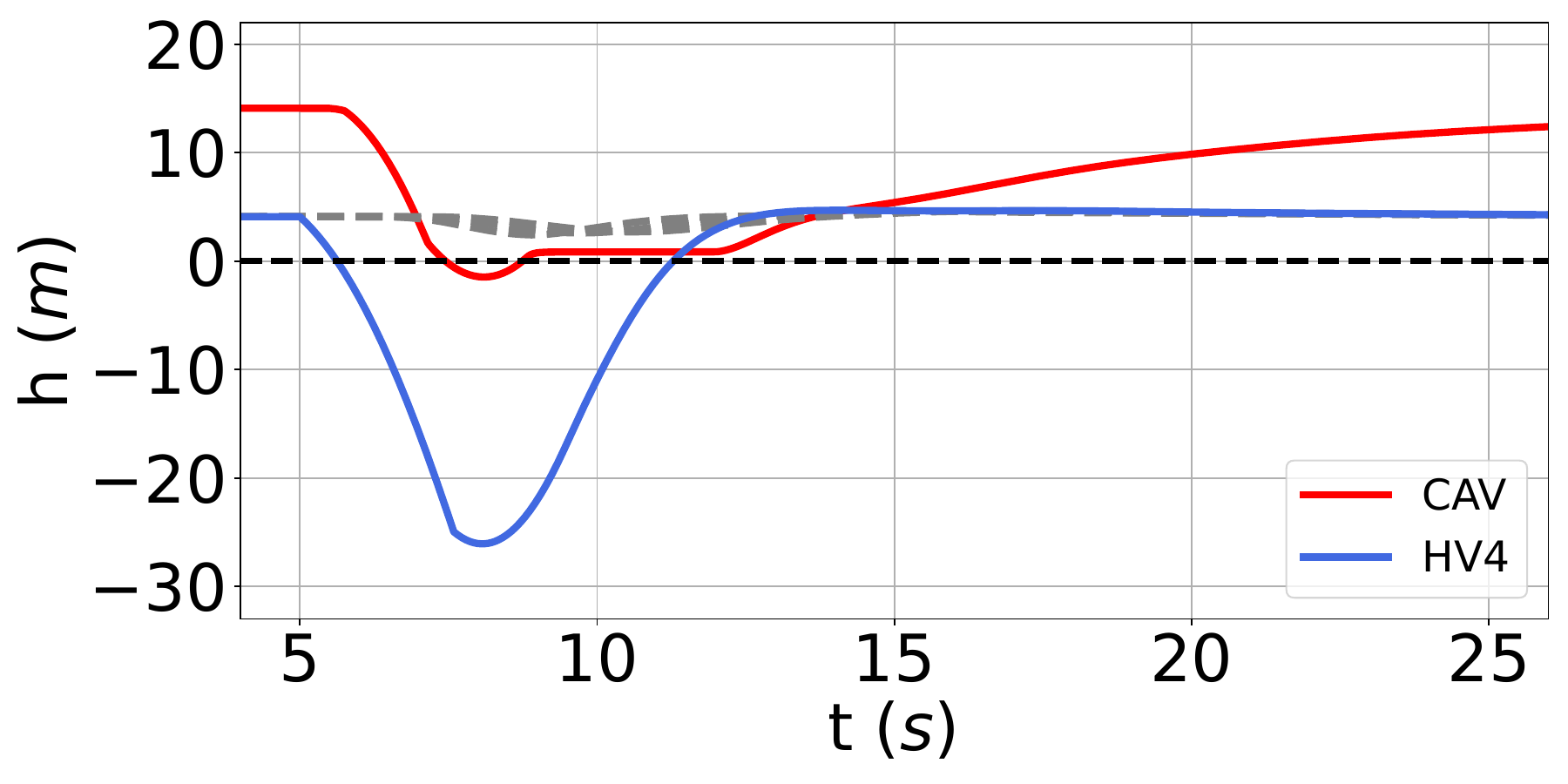}
    \caption{Numerical simulation of RSTC in Scenario 2. The trajectory of HV-4 (blue line) shows that RSTC also improves safety in this scenario.}
    \label{fig:trajectory scenario 2}
\end{figure*}

\subsection{Simulation setting}\label{sec:subsec:simulation setting}

We consider a vehicle chain of six vehicles, with one head vehicle, one CAV, and $N=4$ following vehicles. For the following HVs, we adopt the optimal velocity model (OVM)~\cite{bando1998analysis} as their driving strategy $F_i$:
\begin{equation}
    F_{i}(s_{i},\dot{s}_{i},v_{i})= \alpha \left(V\left(s_{i}\right)-v_{i}\right)+ \beta \dot{s}_{i} \label{eq:OVM}, 
\end{equation}
where $V(s)$ describes the spacing-dependent desired speed, the parameter $\alpha>0$ represents the driver's sensitivity to the mismatch between the desired speed $V(s_i)$ and current speed $v_i$, and the parameter $\beta>0$ reflects the driver's sensitivity to the speed gap with its leader. We take the desired speed-gap relationship as 
\begin{equation}
    V(s)=\left\{
    \begin{array}{ll}
    0, & s \leq s_{\mathrm{st}}, \\
    \frac{v_{\max }}{2}\left(1-\cos \left(\pi \frac{s-s_{\mathrm{st}}}{s_{\mathrm{go}}-s_{\mathrm{st}}}\right)\right), & s_{\mathrm{st}}<s<s_{\mathrm{go}}, \\
   v_{\max }, & s \geq s_{\mathrm{go}},
   \end{array}
   \right.
  \label{eq:Vs}
\end{equation}
where $s_{\mathrm{st}}$, $s_{\mathrm{go}}$, and $v_{\max}$ represent standstill spacing, free flow spacing, and the speed limit, respectively \cite{zhang2016motif}. In the simulation, we take the parameters as $\alpha = 0.6$, $\beta = 0.9$, $s_{\mathrm{st}} = 5$ m, $s_{\mathrm{go}} = 40$ m, and $v_{\max} = 35$ m/s. We set the equilibrium speed  as $v^* = 20$ m/s, and the equilibrium gap is decided by $V(s^*) = v^*$ 
as $s^* = 24$ m.

\begin{figure*}[t!]
    \centering
    Nominal Controller \\
    \includegraphics[width=0.24\linewidth]{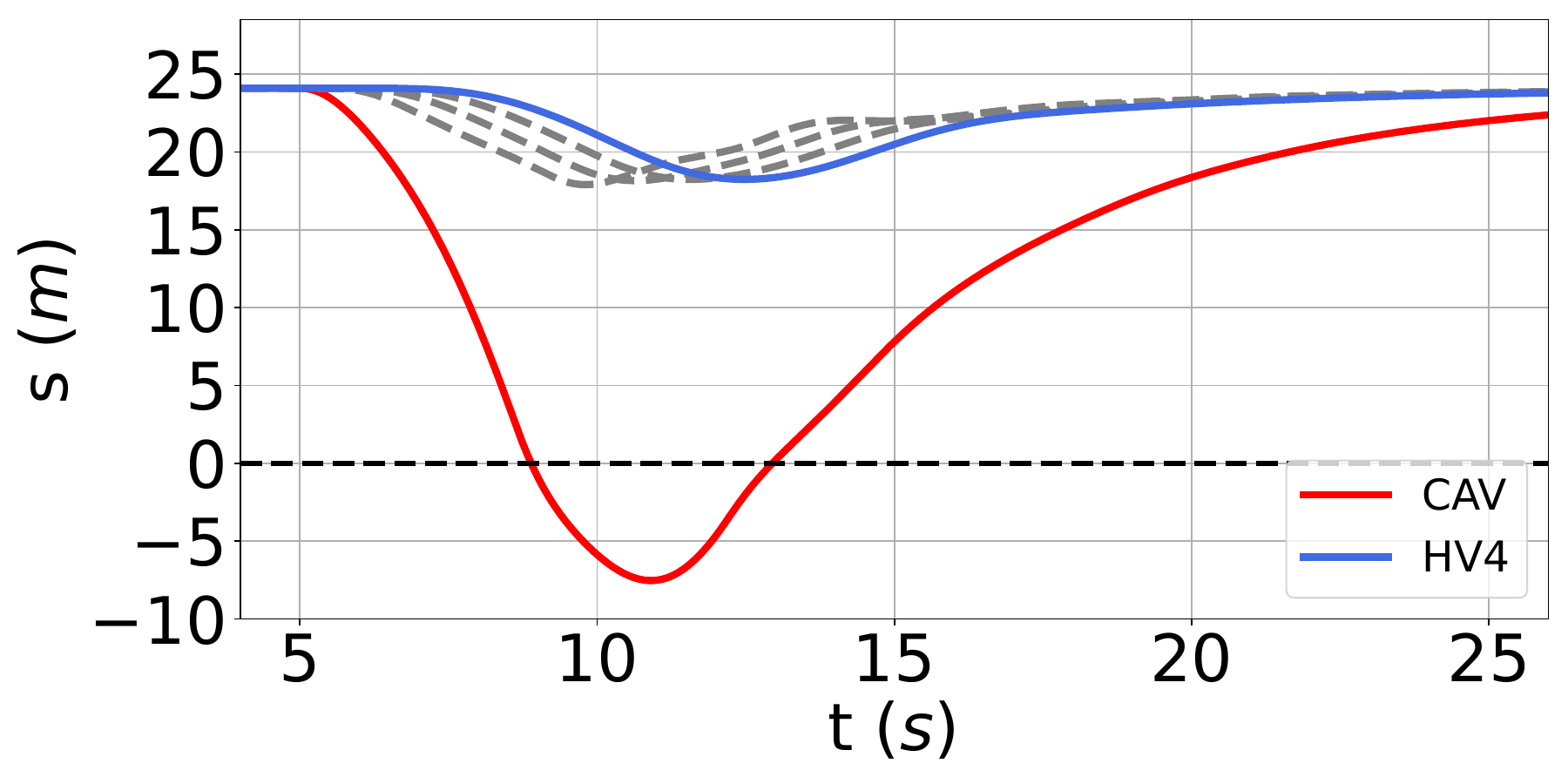}
    \includegraphics[width=0.24\linewidth]{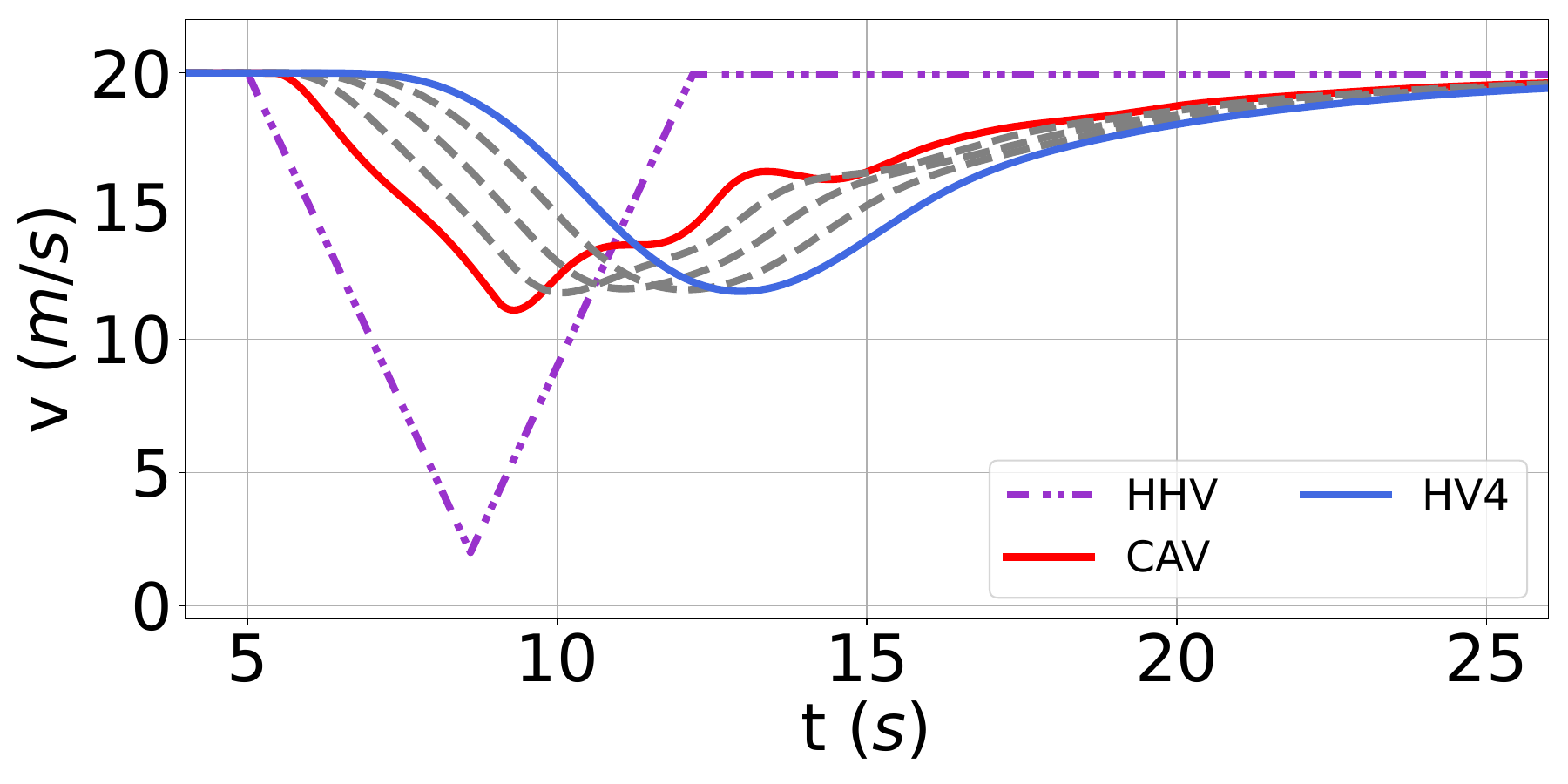}
    \includegraphics[width=0.24\linewidth]{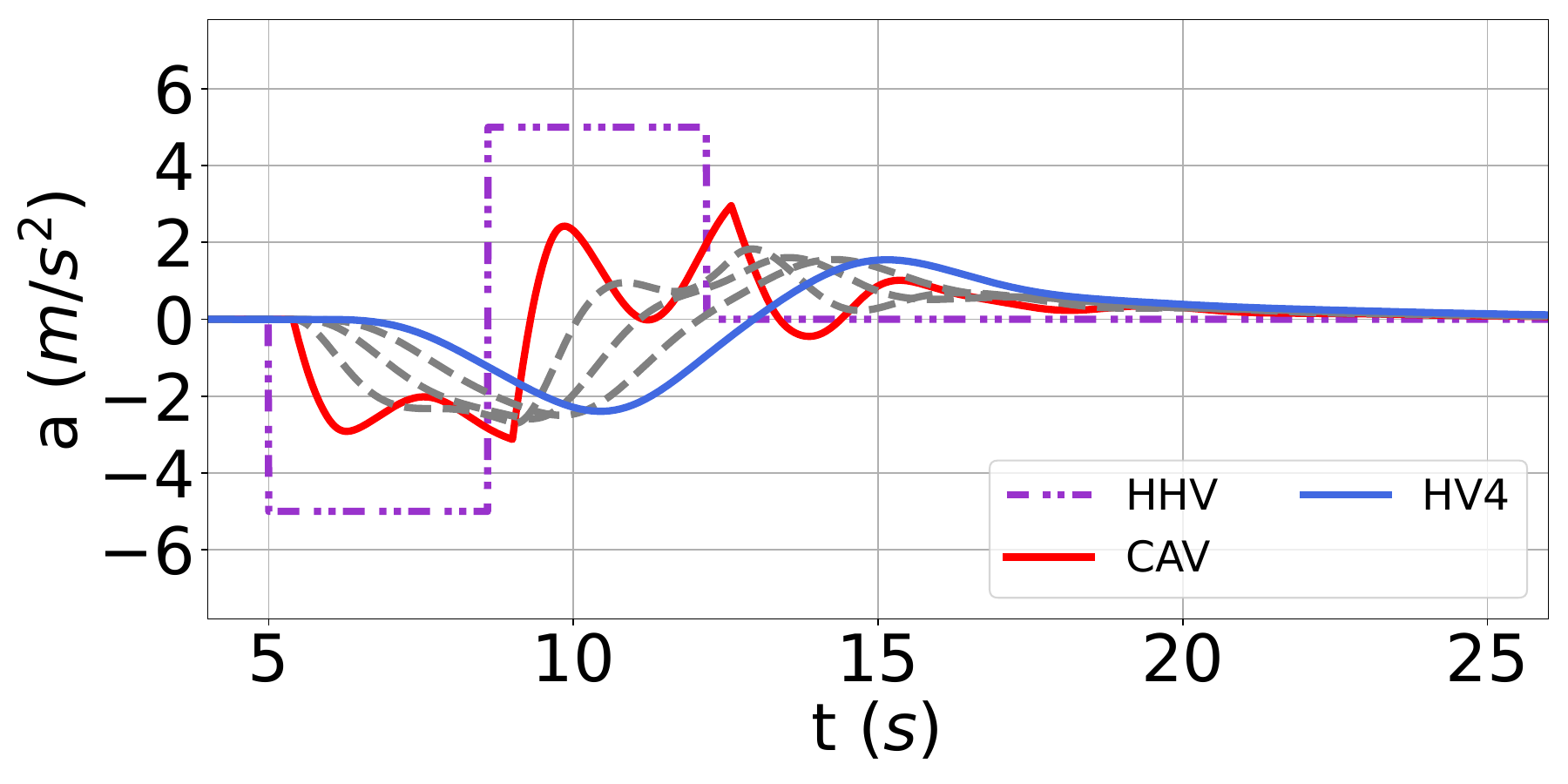}
    \includegraphics[width=0.24\linewidth]{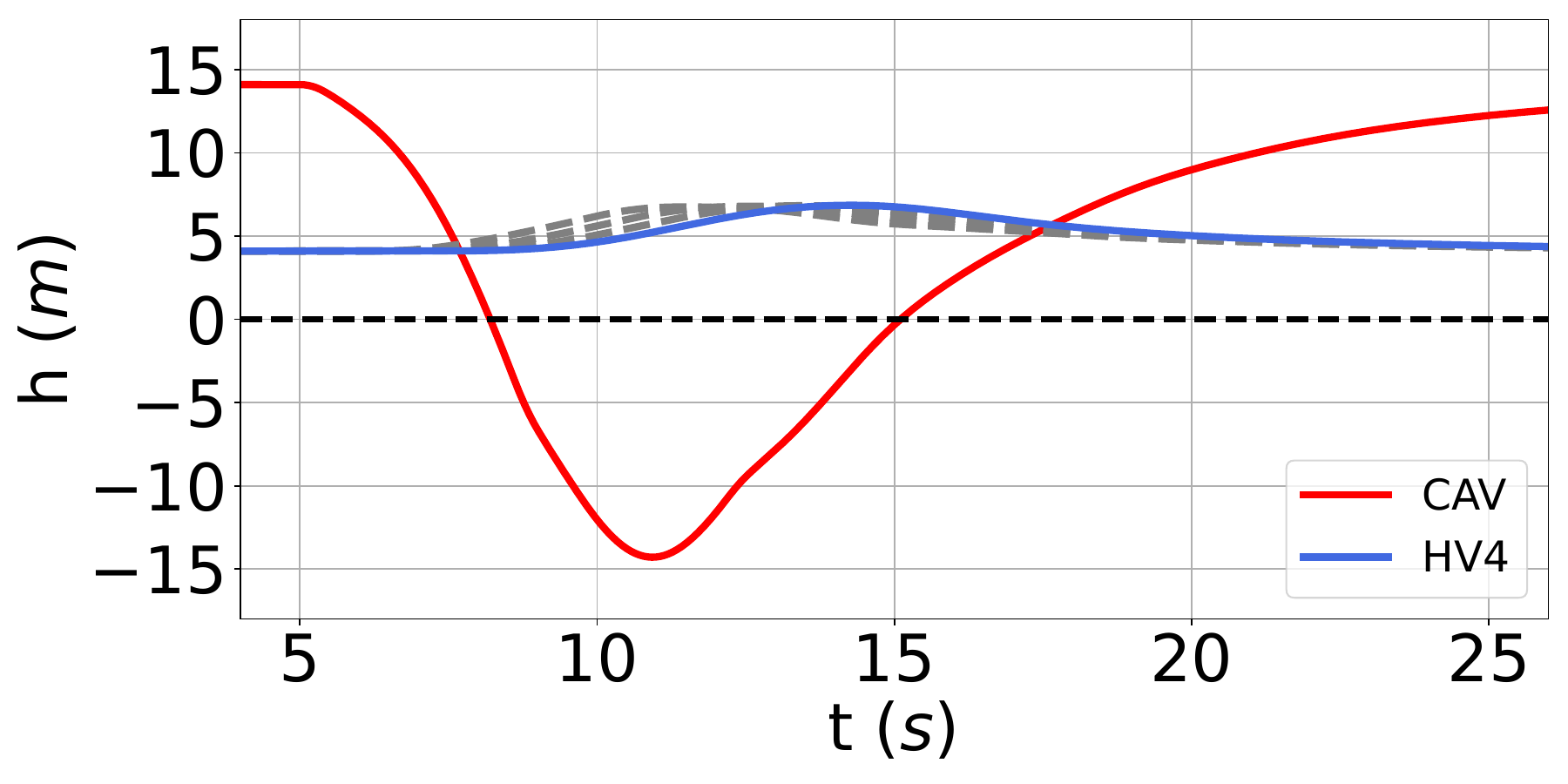}
    \\
    RSTC \\
    \includegraphics[width=0.24\linewidth]{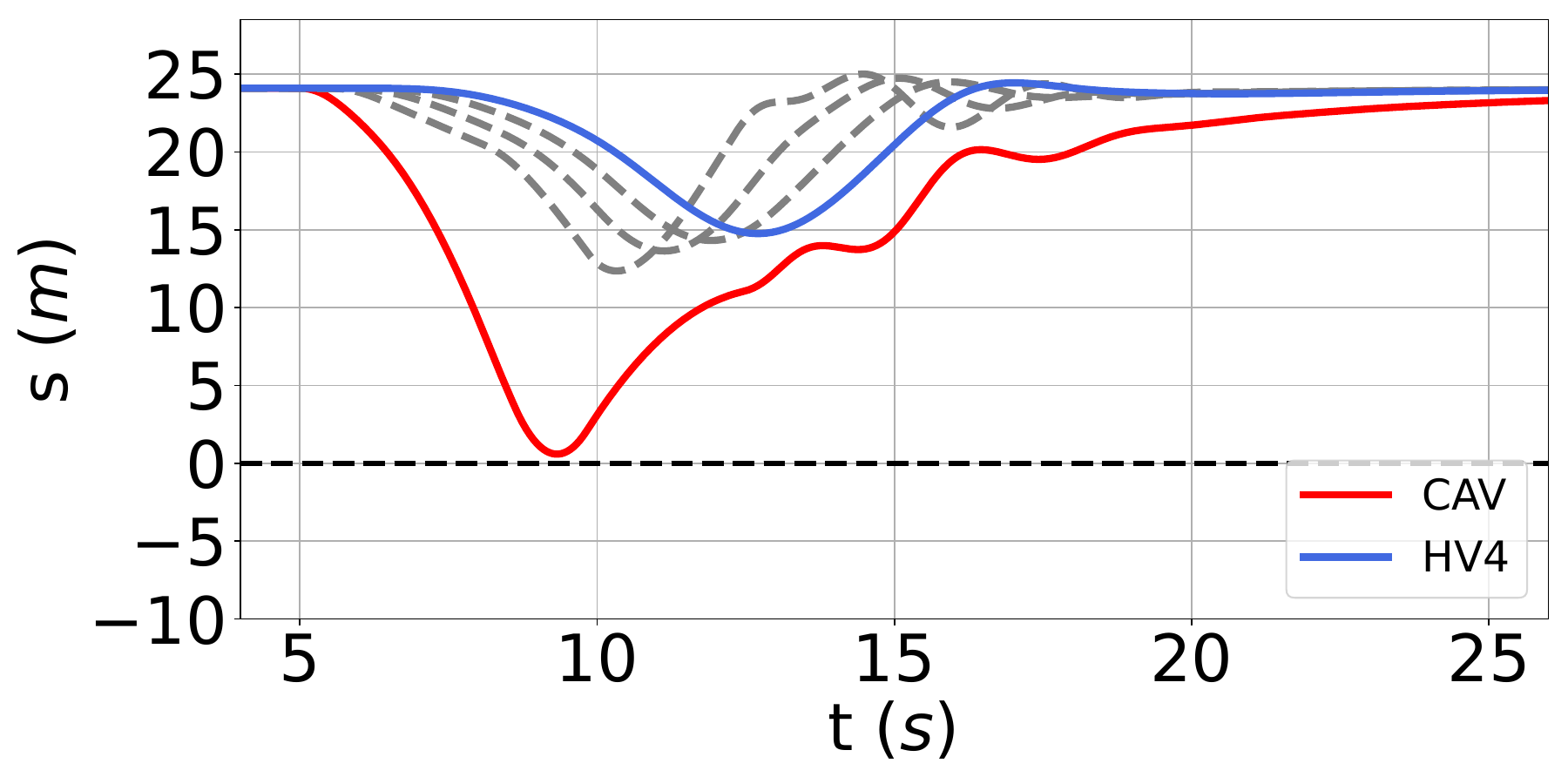}
    \includegraphics[width=0.24\linewidth]{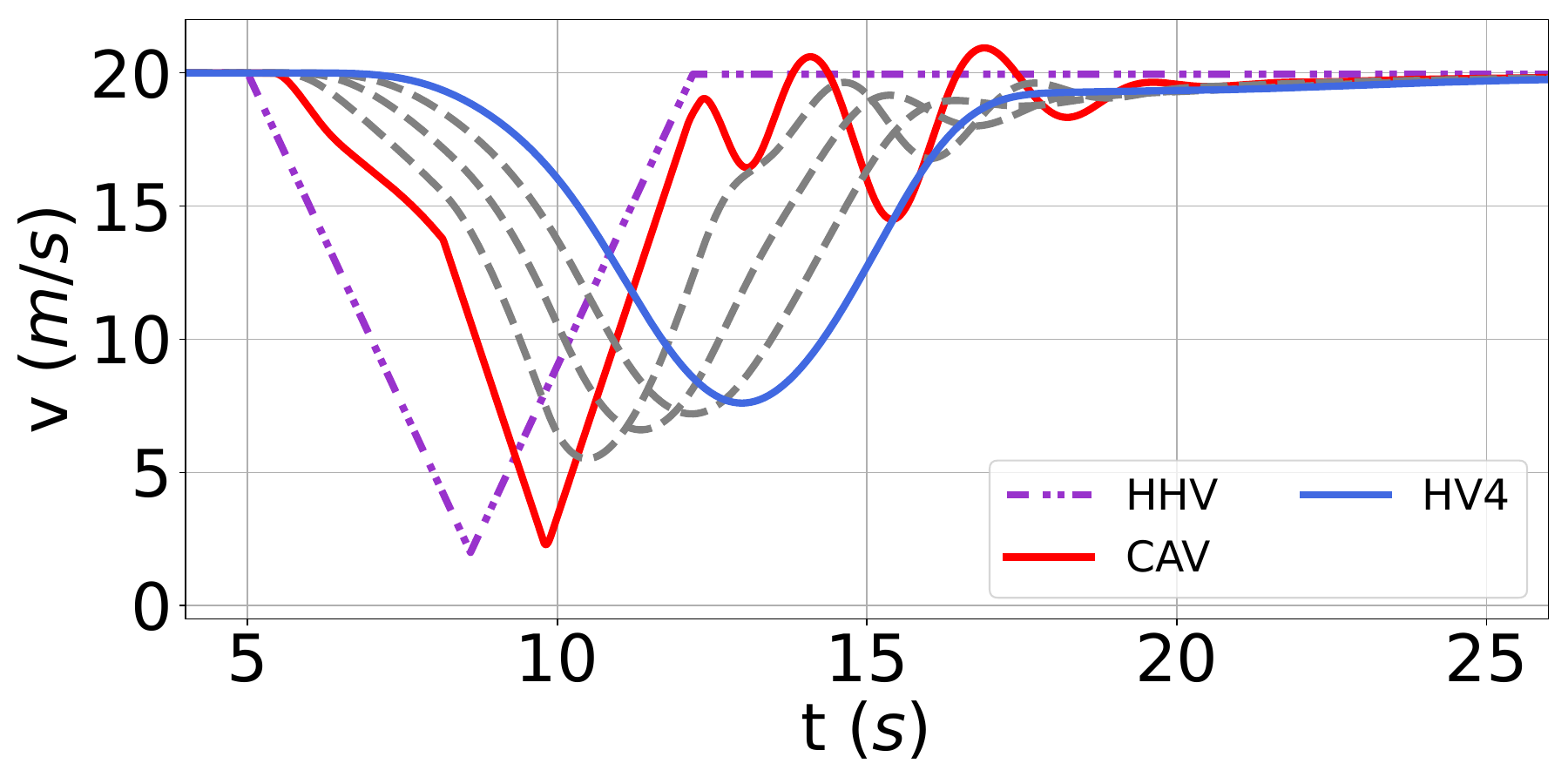}
    \includegraphics[width=0.24\linewidth]{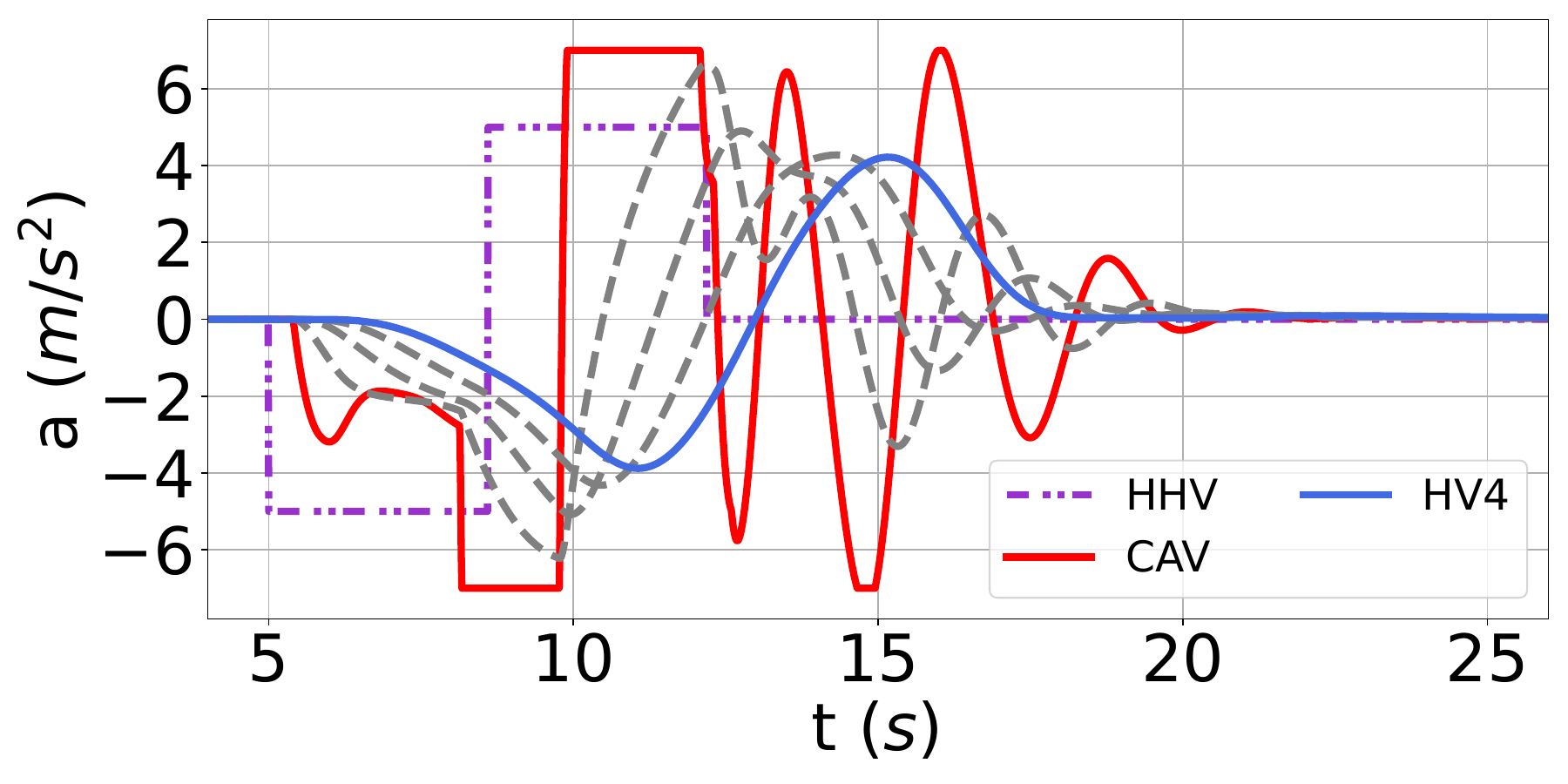}
    \includegraphics[width=0.24\linewidth]{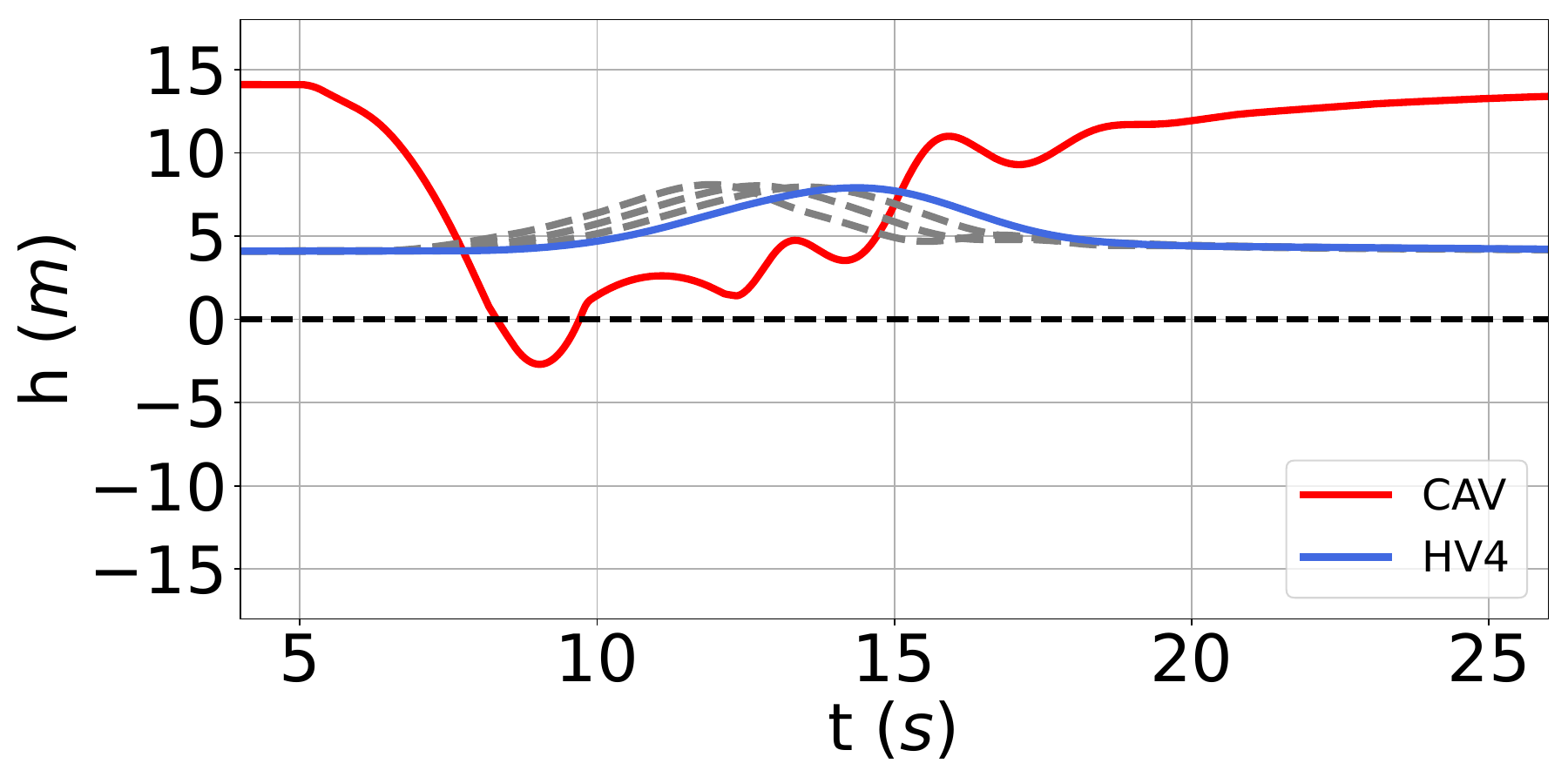}
    \caption{Numerical simulation of RSTC with actuator and sensor delays. We see that the formulated QP \eqref{eq:QP observer} guarantees safety, i.e., the gap keeps positive.}
    \label{fig:trajectory observer scenario 1}
\end{figure*}

\begin{figure*}[t!]
    \centering
    \includegraphics[width=0.24\linewidth]{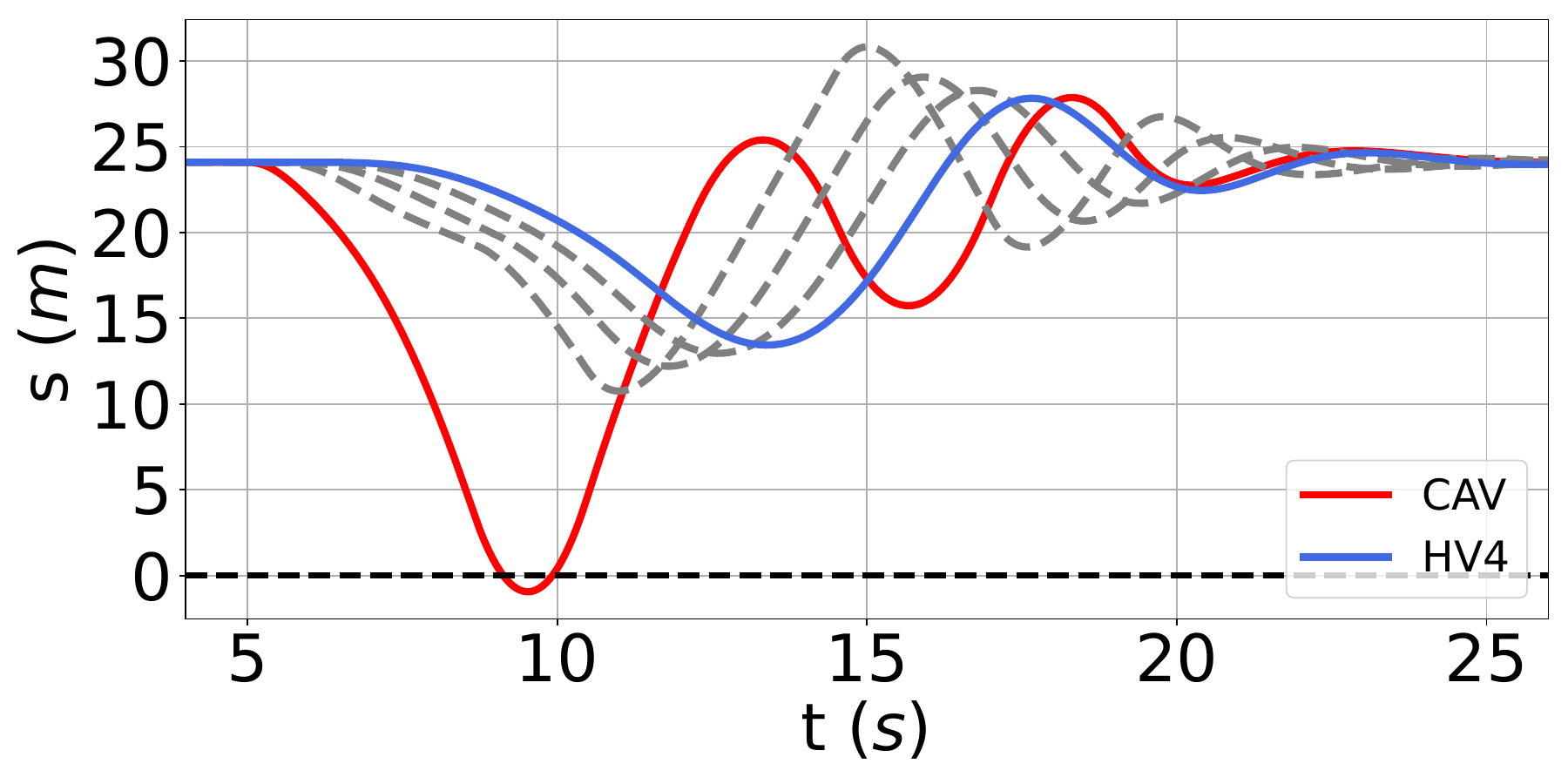}
    \includegraphics[width=0.24\linewidth]{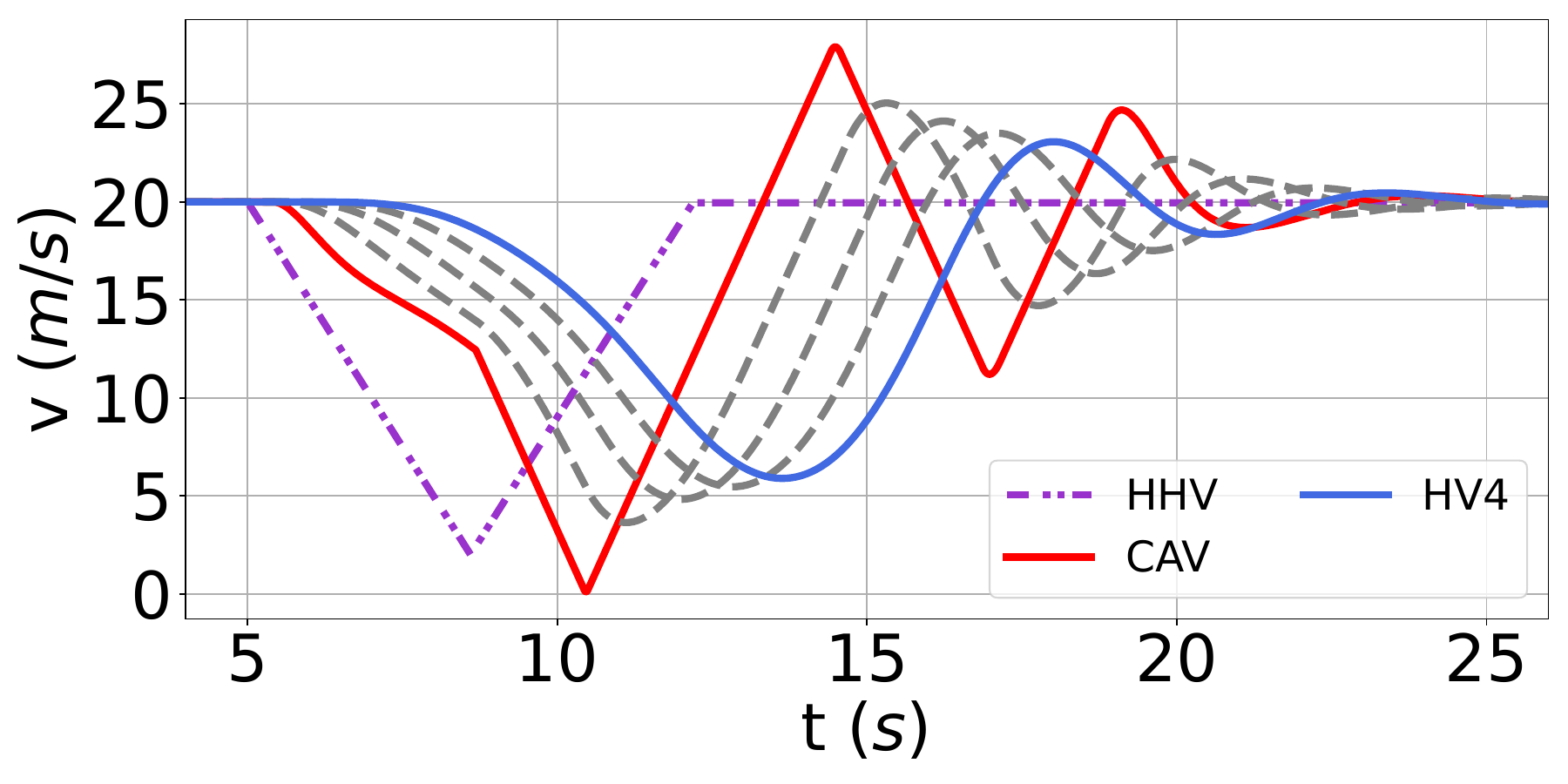}
    \includegraphics[width=0.24\linewidth]{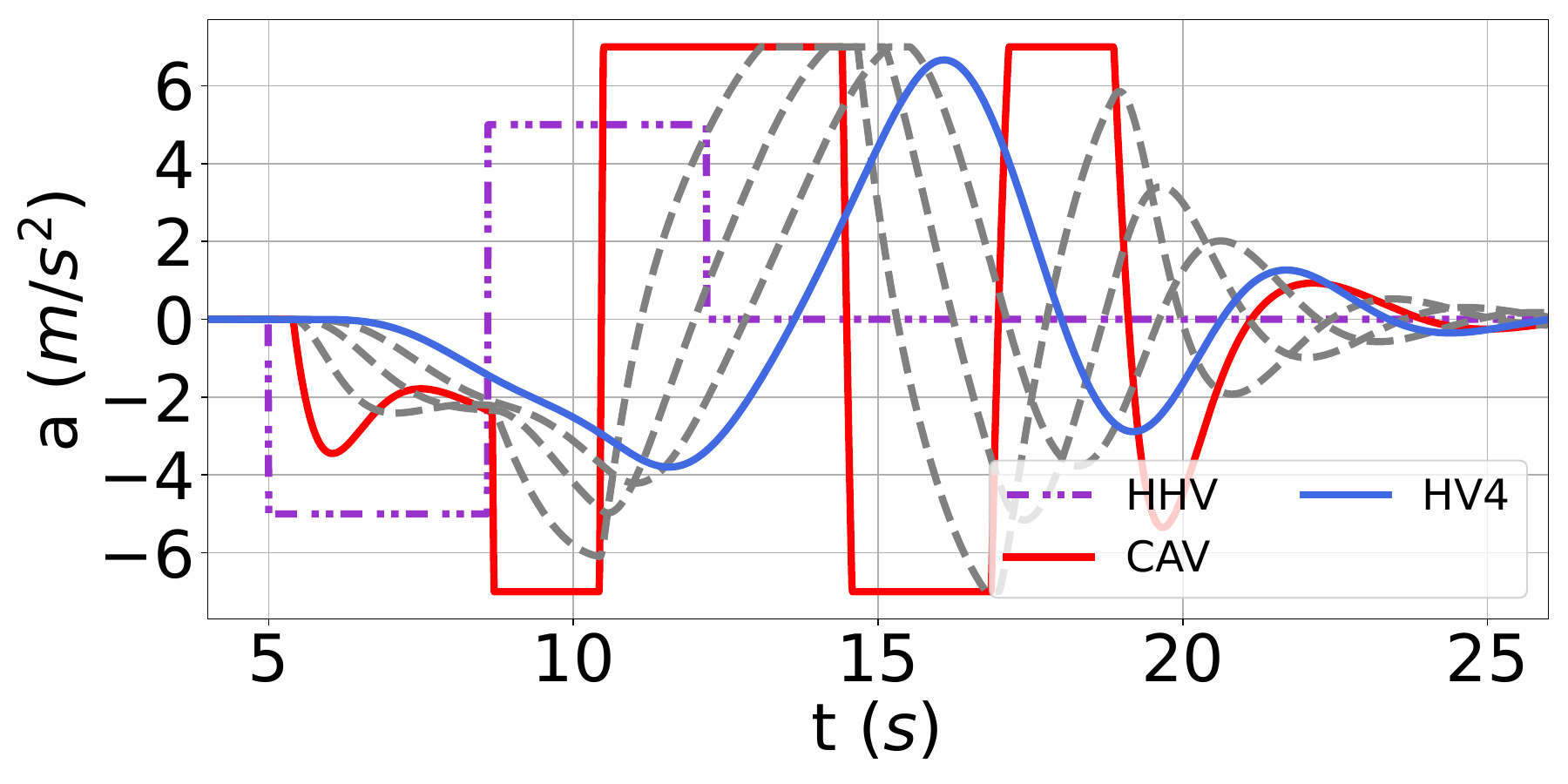}
    \includegraphics[width=0.24\linewidth]{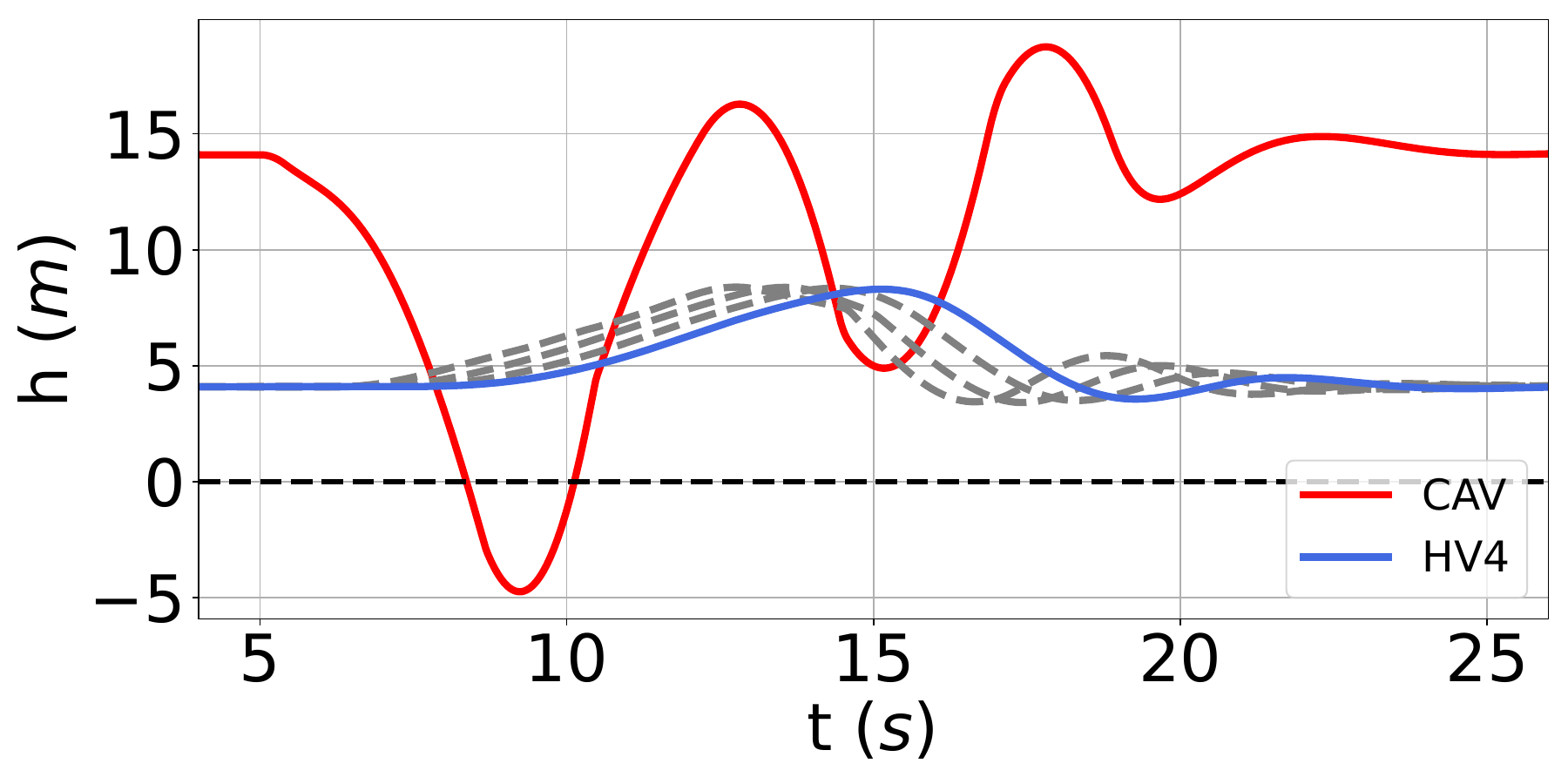}
    \caption{Numerical simulation of STC~\eqref{eq:QP CBFintro traffic} that ignores the actuator delay.}
    \label{fig:trajectory STC nodelay}
\end{figure*}

\begin{figure}
    \centering
    \subfloat[vehicle chain]{\includegraphics[width=0.48\linewidth]{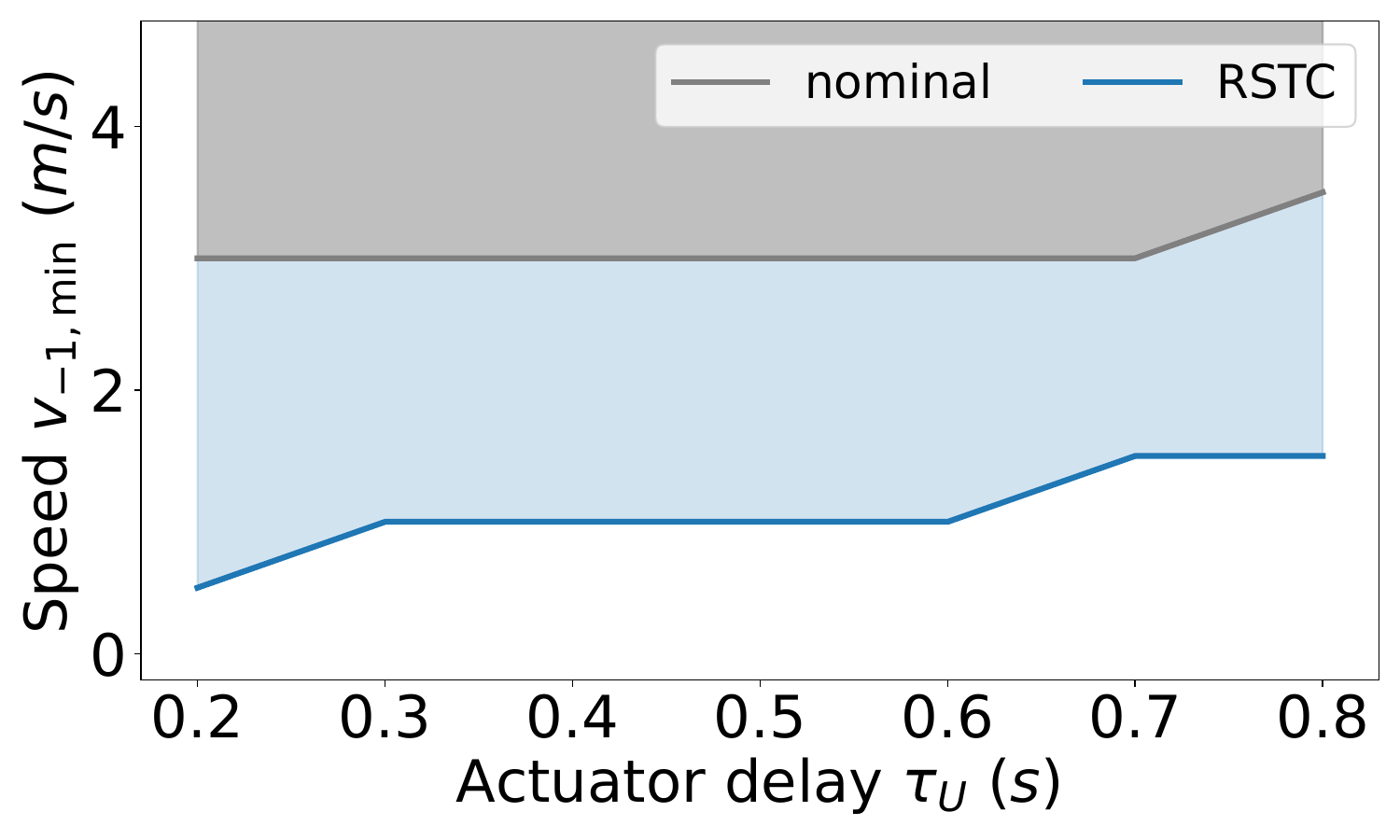}}
    \subfloat[CAV]{\includegraphics[width=0.48\linewidth]{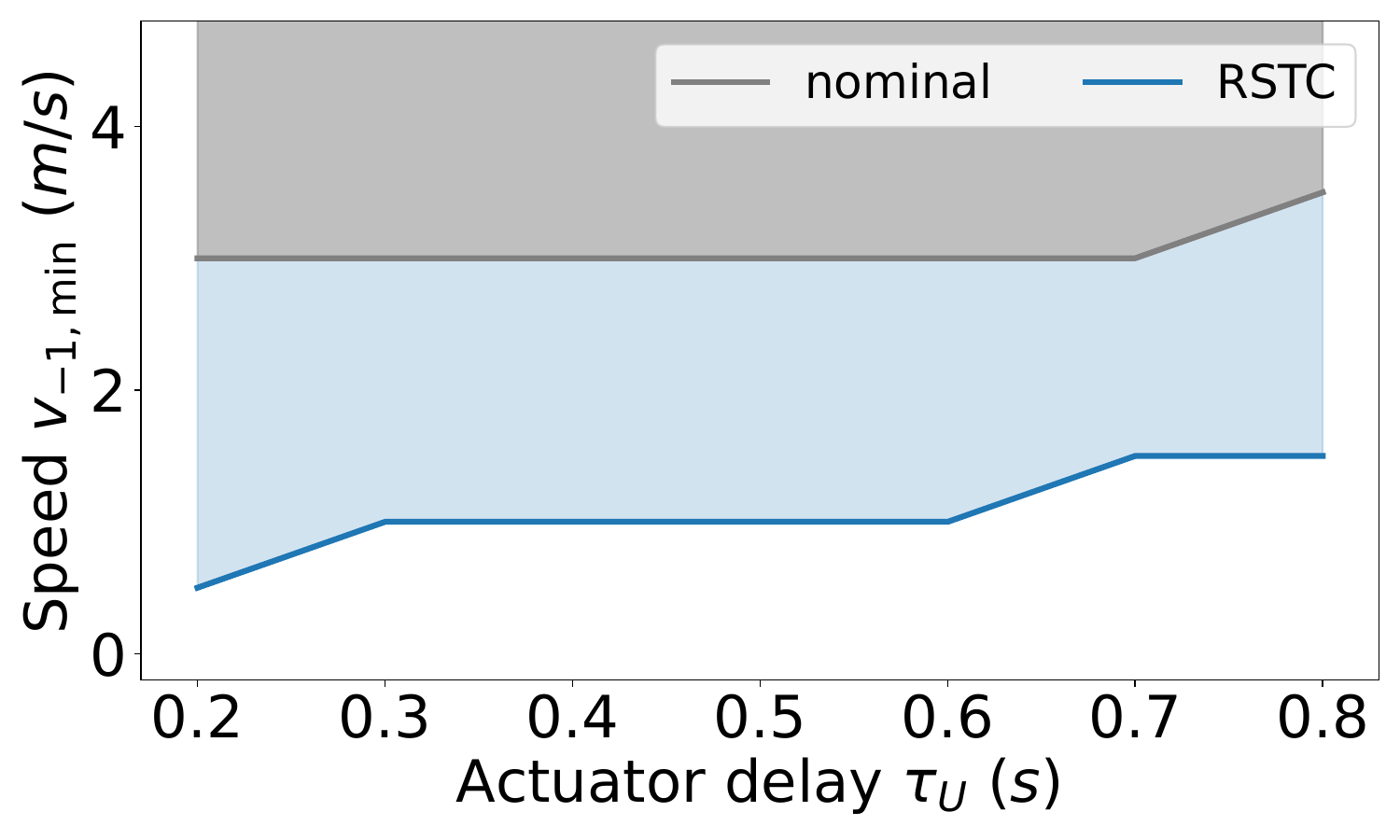}}
    \\
    \subfloat[HV-1]{\includegraphics[width=0.48\linewidth]{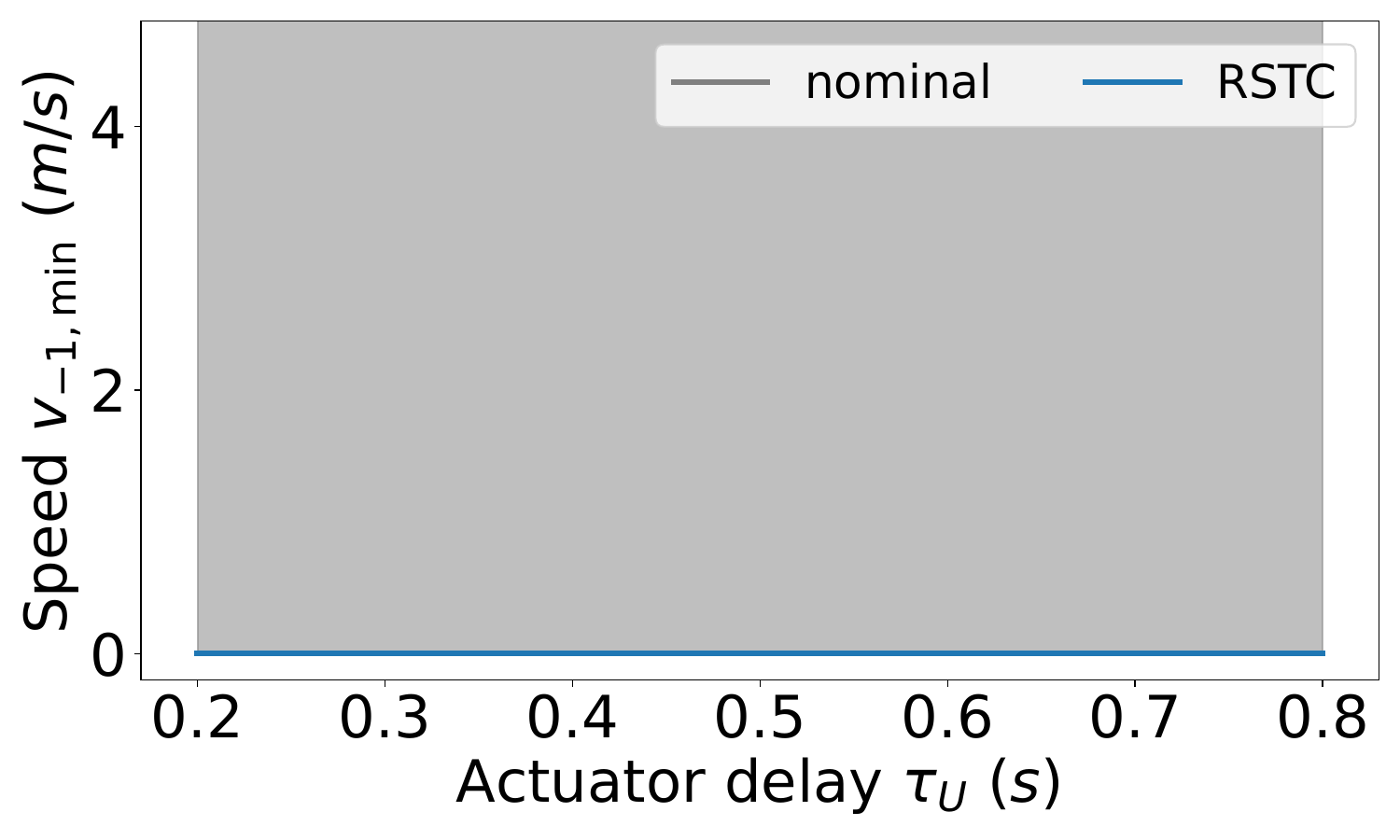}}
    \subfloat[HV-2]{\includegraphics[width=0.48\linewidth]{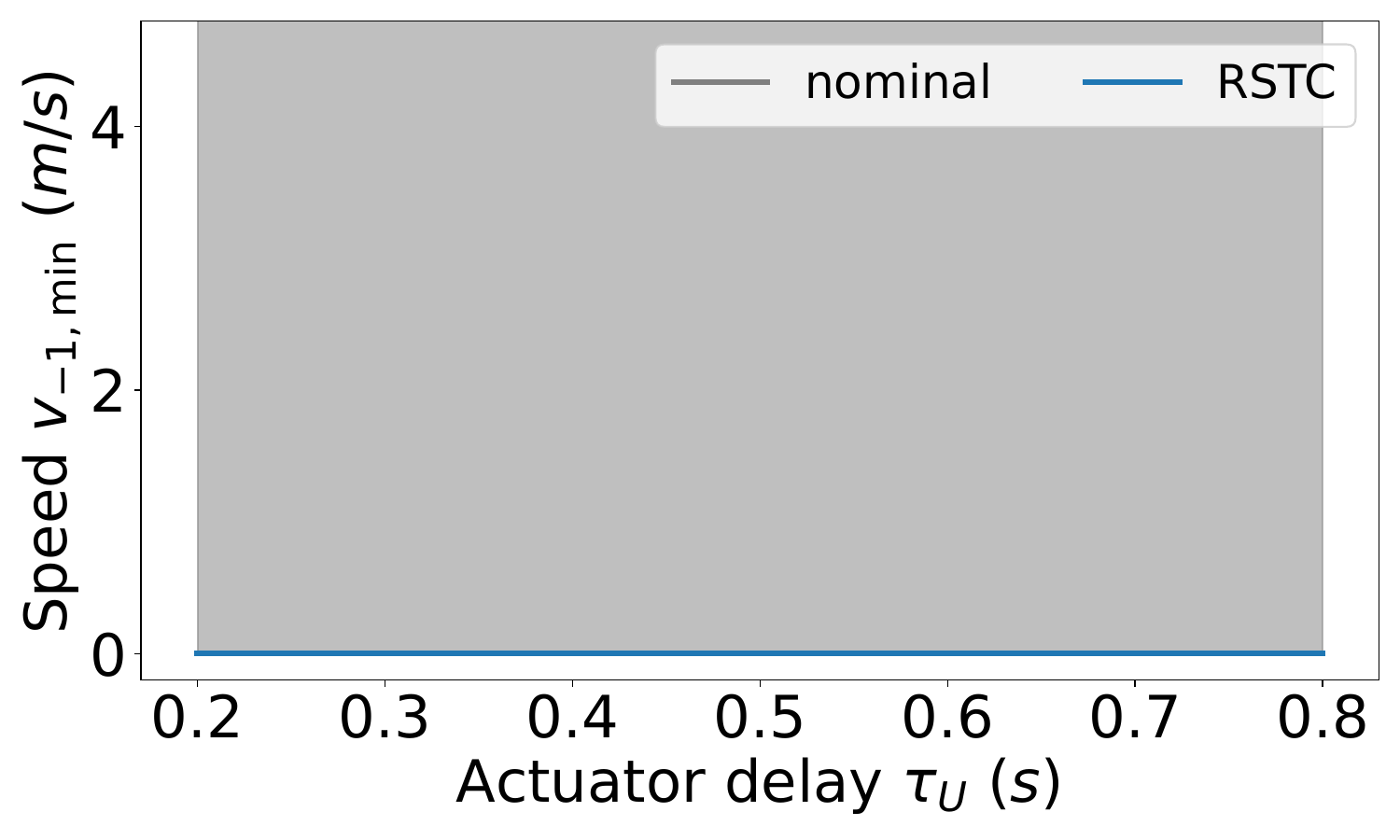}}
    \caption{Safety region, i.e., the range of speed perturbation of head vehicle that causes no rear-end collisions, of Scenario 1 under varying actuator delays $\delayu$. The grey line and area are the boundary and in-domain of the safety region of the nominal controller, respectively. The blue lines represent the boundaries of the safety regions of the proposed RSTC with corresponding time actuator delay $\delayu$. The blue area is the improvement in safety brought by the RSTC over the nominal controller. }
    \label{fig:safe region delayu scenario 1}
\end{figure}

\begin{figure}
    \centering
    \subfloat[vehicle chain]{\includegraphics[width=0.48\linewidth]{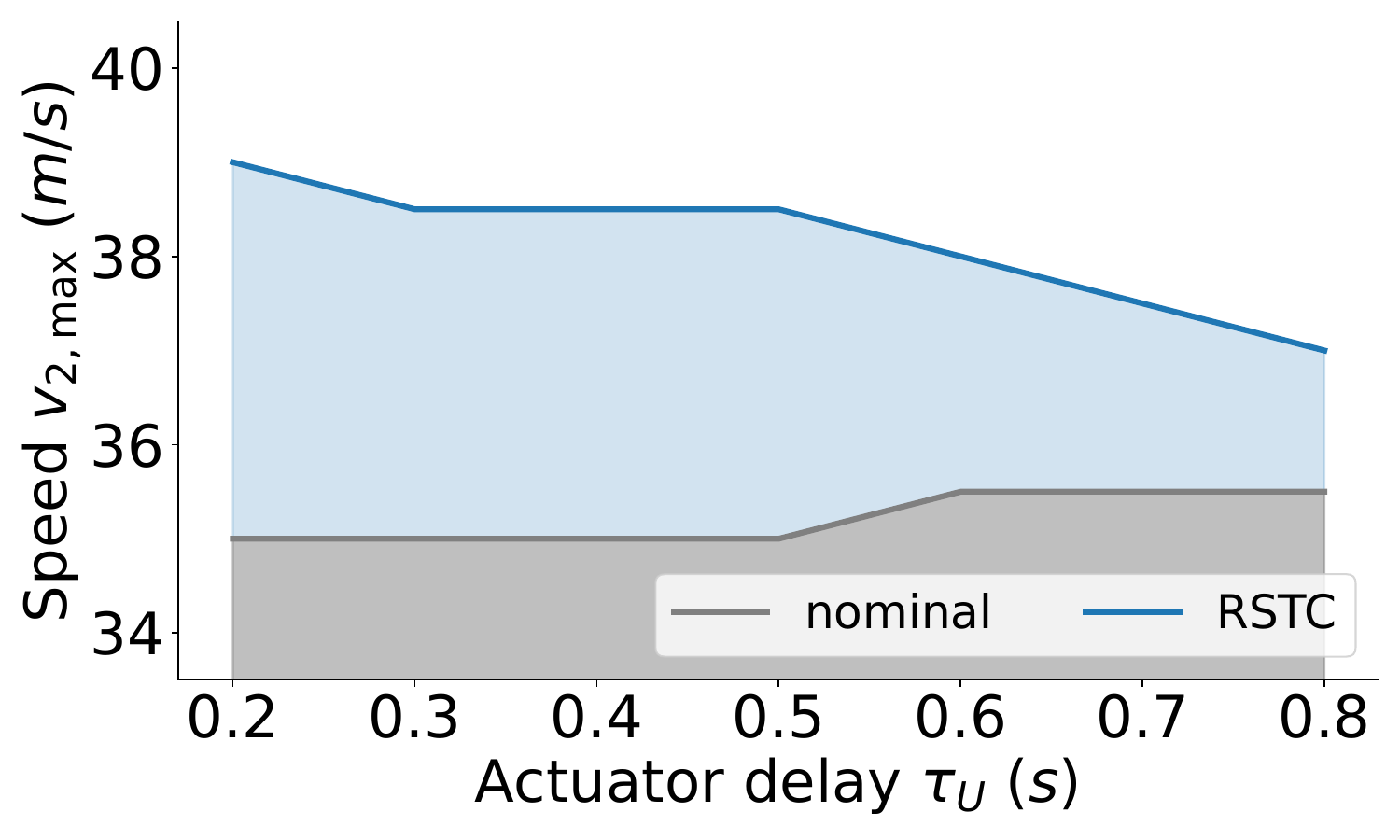}}
    \subfloat[CAV]{\includegraphics[width=0.48\linewidth]{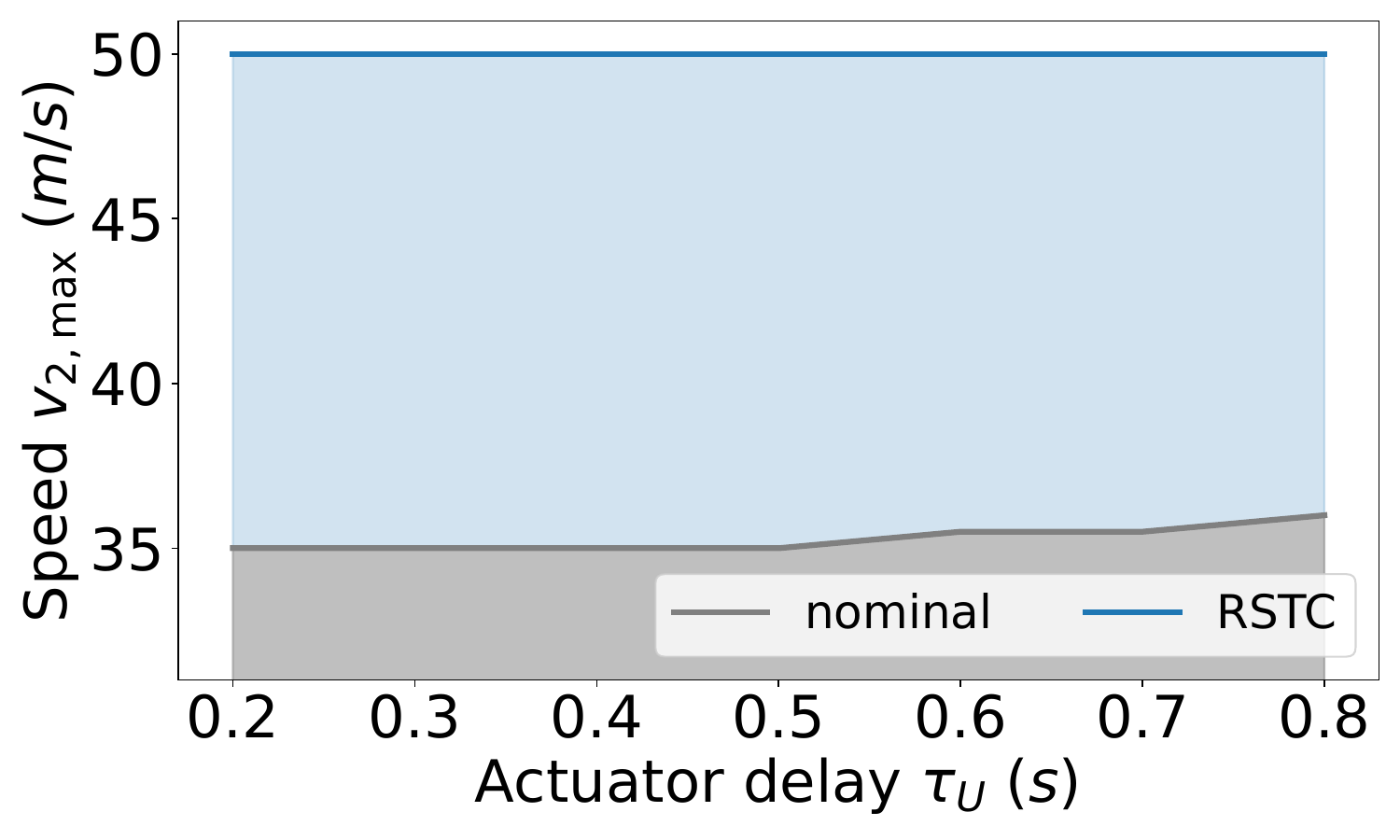}}
    \\
    \subfloat[HV-1]{\includegraphics[width=0.48\linewidth]{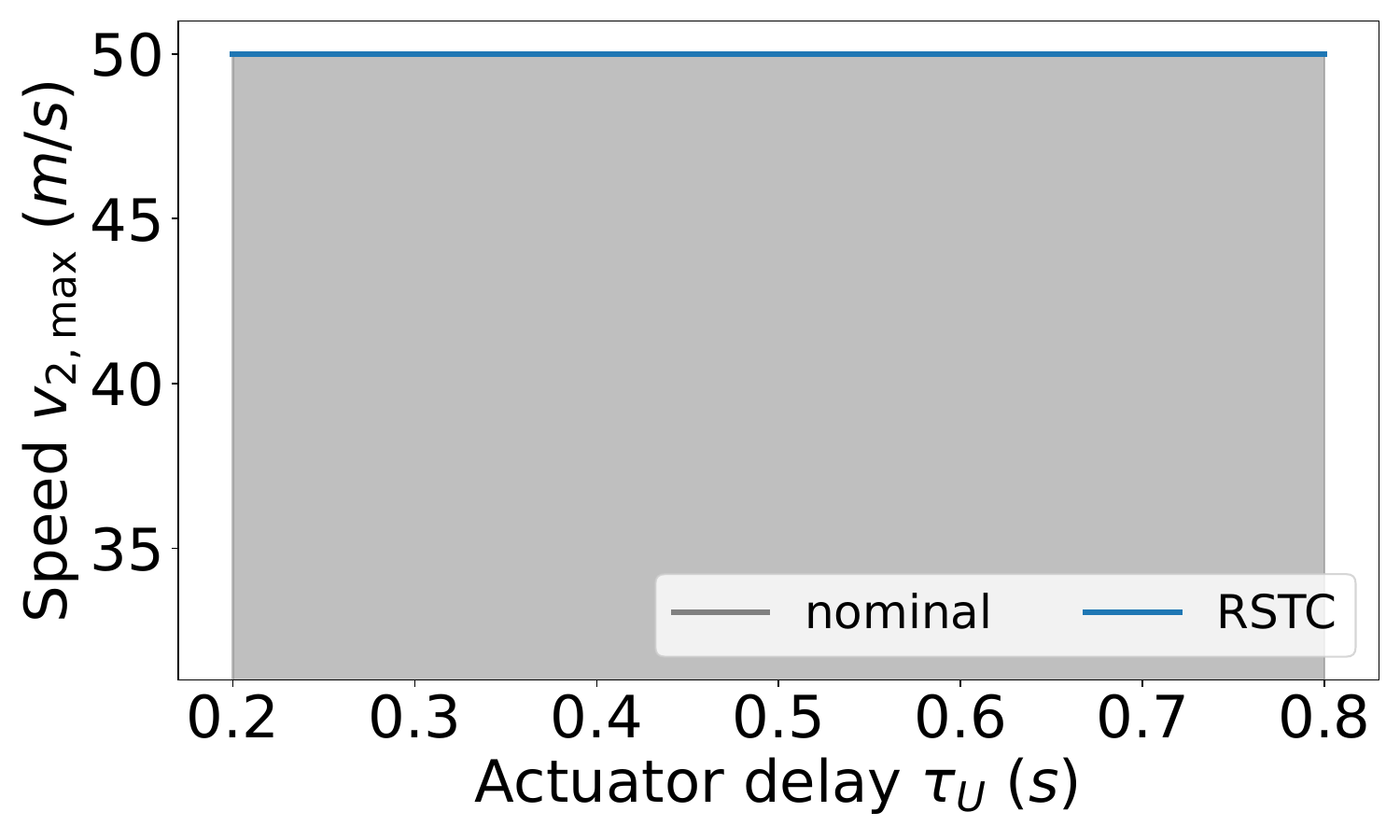}}
    \subfloat[HV-2]{\includegraphics[width=0.48\linewidth]{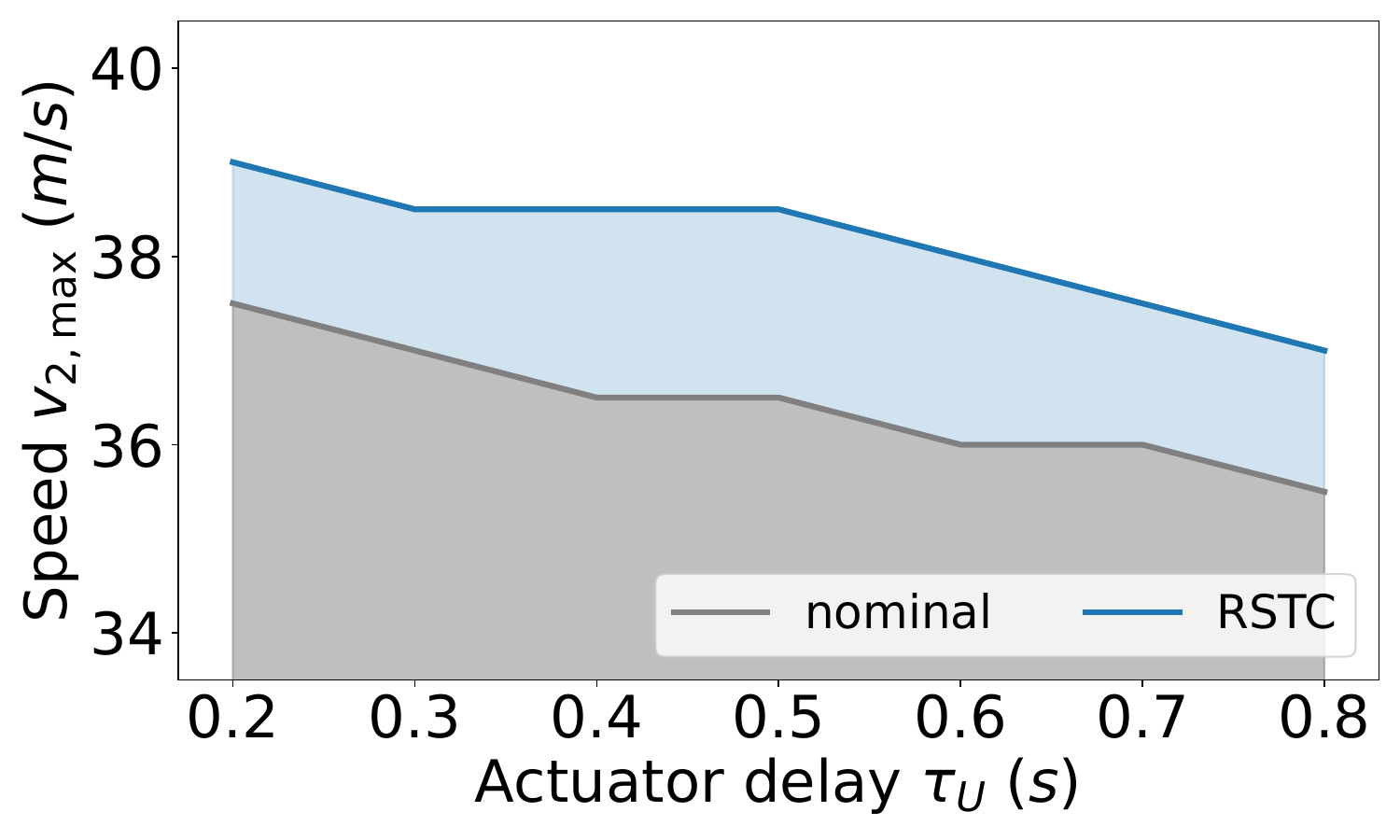}}
    \caption{Safety region of Scenario 2 under varying actuator delays $\delayu$.}
    \label{fig:safe region delayu scenario 2}
\end{figure}

For the nominal controller, we use the feedback controller 
\begin{align}\label{eq:nominal controller}
    u_0(t) = Kx_p(t) + \alpha_3 r(t),
\end{align}
with the feedback gain 
\begin{align}
    K = [\alpha_1,-\alpha_2,-2,0.2,-2,0.2,-2,0.2,-2,0.2].
\end{align}
In this paper, we focus on safety guarantee for the mixed autonomy system~\eqref{eq:system}. The stability of the system by the nominal controller can be referred to~\cite{wang2021leading,beregi2021connectivity}.
We take the actuator delay as $\delayu  = 0.4$ s and the sensor delay as $\delaysensor = 0.8$ s. For the safe spacing policy CTH \eqref{eq:safe s}, we take the time headway for CAV as $\psi_0 = 0.5$ s and for following HVs as $\psi_i = 1$ s. We solve the formulated QP \eqref{eq:QP} and \eqref{eq:QP observer} in  discrete time with the sampling time being 0.01 s.

We consider two safety-critical scenarios that can happen in real traffic.
\begin{itemize}
    \item Scenario 1: sudden deceleration of the head vehicle. This will happen when there is an aggressive cut-in from adjacent lanes. We make the head vehicle decelerate with an acceleration of $-a_{H}$ and duration of $t_{H}$. Then, the head vehicle accelerates to the equilibrium speed $v^*$ with acceleration $a_{H}$.  
    \item Scenario 2: sudden acceleration of the following vehicle. This may happen due to fatigue driving or wrong operations of human drivers. We set the last HV accelerate with an acceleration of $a_{F}$ and duration of $t_{F}$. After that, the HV decides its acceleration based on the model \eqref{eq:OVM}.
\end{itemize}
In Scenario 1 and 2, there is rear-end collision risk for the CAV and following HVs, respectively.

\subsection{Safety guarantee by RSTC}\label{sec:subsec:simulation main result}

We design CBF constraints when there is an actuator delay in Theorem \ref{theorem:safety} and construct a QP as \eqref{eq:QP}. 
In Fig. \ref{fig:trajectory scenario 1}, we give the profile of gap $s$, speed $v$, and acceleration $a$ generated by the nominal controller \eqref{eq:nominal controller} and the proposed RSTC controller \eqref{eq:QP} in Scenario 1 with $a_H = -5 \; \mathrm{m/s}$ and $t_H = 3.5$ s. For the nominal controller, we see from the second column of the speed $v$ profile that it stabilizes the mixed traffic so that the last vehicle (blue line) has a lower speed perturbation than the HHV (purple line). But this causes a rear-end collision between CAV and HHV ($s_0<0$, red line in the first column of the gap $s$ profile). The RSTC, on the other side,  generates a safety-critical controller   by setting a larger deceleration for CAV (red line in the third column of acceleration $a$ profile) to avoid collisions. We also give the profile of $h_i$ in the fourth column in Fig. \ref{fig:trajectory scenario 1}, which shows that RSTC guarantees $h_i\ge 0$.

In Fig. \ref{fig:trajectory scenario 2}, we compare the nominal controller \eqref{eq:nominal controller} and the proposed RSTC controller \eqref{eq:QP} in Scenario 2 with $a_F = 5 \; \mathrm{m/s}$ and $t_F = 2.6$ s. When HV-4 accelerates, we see that the RSTC synthesizes a larger acceleration for CAV, so that the following HV-1 to HV-3 also accelerates to avoid collisions between HV-4 and HV-3.

When only partial state is known with  sensor delay, we design CBF constraints in Theorem  \ref{theorem:safety observer}, and formulate the QP as in \eqref{eq:QP observer}. In Fig.   \ref{fig:trajectory observer scenario 1}, we give the trajectory by the RSTC \eqref{eq:QP observer}. We see that when there are actuator and sensor delays,   the RSTC still ensures safety and avoids collision.

\subsection{Analysis on RSTC}\label{sec:subsec:simulation analysis}

In Theorem \ref{theorem:safety}, we design safety constraints for mixed autonomy traffic when there is an actuator delay $\delayu$. The effect of the actuator delay is characterized by functions $M_i(t)$. To see the necessity of $M_i(t)$ in ensuring safety, we run simulations with the QP~\eqref{eq:QP CBFintro traffic} formulated for delay-free mixed autonomy systems. 
In Fig.~\ref{fig:trajectory STC nodelay}, we give the trajectory by~\eqref{eq:QP CBFintro traffic}. We see that the CBF designed for delay-free systems fails to ensure safety when there are delays.

We analyze how the actuator delay $\delayu$ affects the safety performance. We run simulations with $N=2$ following HVs, since we can see more clearly the safety of each vehicle. To evaluate a controller's performance in guaranteeing safety, we use the notation of ``safety region'', which is the range of speed perturbation with no rear-end collisions. 
\begin{itemize}
    \item For Scenario 1, it is the minimum speed of the head vehicle $v_{-1,\min}$ given as:
\begin{align}
   v_{-1,\min} = v^* - a_H t_H.
\end{align}
\item For Scenario 2, we use the maximum speed of HV-2 $v_{2,\max}$:
\begin{align}
    v_{2,\max} = v^* + a_F t_F.
\end{align}
\end{itemize}
In Fig.~\ref{fig:safe region delayu scenario 1} and Fig.~\ref{fig:safe region delayu scenario 2}, we give the safety region with the actuator delay $\delayu$ varying from $0.2$ s to $0.8$ s. From Fig.~\ref{fig:safe region delayu scenario 1}(a) and Fig.~\ref{fig:safe region delayu scenario 2}(a), we see that RSTC enlarges the safety region over the nominal controller with a wide range of actuator delays. As for the safety of each vehicle in Scenario 1, we see from Fig.~\ref{fig:safe region delayu scenario 1}(c) and Fig.~\ref{fig:safe region delayu scenario 1}(d) that the nominal controller already avoids collisions for the following HV-1 and HV-2, and the RSTC keeps this safety guarantee. For the CAV, the RSTC improves its safety region as shown in Fig.~\ref{fig:safe region delayu scenario 1}(b). For Scenario 2, we see from Fig.~\ref{fig:safe region delayu scenario 2}(b) that RSTC  avoids rear-end collisions for the CAV even when the following HV-2 accelerates to  50 m/s. As for the following HV-1, Fig.~\ref{fig:safe region delayu scenario 2}(c) shows that both the nominal controller and the RSTC ensure its safety. For following HV-2's safety, we see from Fig.~\ref{fig:safe region delayu scenario 2}(d) that RSTC enlarges its safety region under various $\delayu$.

\section{Conclusion}

In this paper, we develop robust-safety-critical traffic controller for mixed autonomy systems to guarantee a collision-free safety, in the presence of actuator and sensor delays, and disturbances from the leading HV.  Both full-state and partial-state feedback are considered for the stabilizing CAV, and safety constraints are designed for the two cases by predictor-based CBF  and predictor-observer-based CBF, respectively. For future work,  it is of interest to authors to address safety when control and coordination of multiple CAVs in mixed traffic is considered. Another extension is to consider safety-critical control in more complex and challenging traffic scenarios in addition to car-following behaviors, such as lane changing and merging.

\appendices
\section{Proof of Theorem \ref{theorem:safety}} \label{sec:appendix proof safety}

We first prove the safety constraints for CAV following the three steps. 

Step \uppercase\expandafter{\romannumeral1}: Bound the prediction error.

From the predictor \eqref{eq:predict x true} and the true future state \eqref{eq:predict x}, we have the prediction error is
\begin{align}
    x_e &= x(t+\delayu) - x_p(t) \notag\\
    &= \int_{0}^{\delayu} e^{A(\delayu-\theta)} D (r(t+\theta)-r(t)) \diff \theta. \label{eq:prediction error}
\end{align}
From \eqref{eq:model ABD}, we have that only the first element in $D$  is non-zero, so we have  
\begin{align}
    e^{A(\delayu-\theta)} D &= \left(e^{A(\delayu-\theta)}\right)_{\cdot,1} \notag\\
    &= \left(I + \sum_{k=1}^{\infty}(\delayu-\theta)^kA^k\right)_{\cdot,1} \notag\\
    & = I_{:,1} + \sum_{k=1}^{\infty}(\delayu-\theta)^k \left(A^k\right)_{\cdot,1},
\end{align}
where for a matrix $A$, we use $A_{\cdot,j}$ to denote the vector of its $j$-th column. 
From the coefficient definition in  \eqref{eq:model ABD}, we see that the first column of $A$ is zero, i.e., $A_{\cdot,1} = 0$. We prove by mathematical induction that
\begin{align}
    \left(A^k\right)_{\cdot,1} = 0, \quad \forall k\ge 1.
\end{align}
Therefore, we have
\begin{align}\label{eq:eAD=100}
    e^{A(\delayu-\theta)} D = I_{\cdot,1} = [1,0,\cdots,0]^{\top},
\end{align}
which indicates that there is only a prediction error for the CAV's gap. Denote $s_{0,p}(t)$  and $v_{0,p}(t)$ as the predicted gap and speed for CAV, we have the prediction error is
\begin{align}
    s_0(t+\delayu) - s_{0,p}(t) &= \int_0^{\delayu} \left(r(t+\theta) - r(t) \right)\diff \theta, \\
    v_0(t+\delayu) - v_{0,p}(t) &= 0. \label{eq:predict error bound v0}
\end{align}
To derive a bound on the prediction error $s_0(t+\delayu) - s_{0,p}(t)$, we note that
\begin{align}
    r(t+\theta) = r(t) + \int_0^{\theta} \dot{r} (t+\xi) \diff \xi.
\end{align}
From Assumption \ref{assumption:bound rdot}, we have
\begin{align}
    r(t) + \theta \drlower \le r(t+\theta) \le  r(t) + \theta \drupper,
\end{align}
which gives
\begin{align}
    s_0(t+\delayu) - s_{0,p}(t) &\le \int_0^{\tau} \theta \drupper \diff \theta = \frac{1}{2}\drupper\delayu^2, \\
    s_0(t+\delayu) - s_{0,p}(t) &\ge \int_0^{\tau} \theta \drlower \diff \theta = \frac{1}{2}\drlower\delayu^2. \label{eq:predict error bound s0 lower}
\end{align}

Step \uppercase\expandafter{\romannumeral2}: Construct a robust safety function.

By the prediction error \eqref{eq:predict error bound v0} and \eqref{eq:predict error bound s0 lower}, we see that to meet the safe gap condition in \eqref{eq:safe s},  it is sufficient to satisfy 
\begin{align}
    s_{0,p}\left(t\right)  + \frac{1}{2}\drlower\delayu^2 \ge \psi_0 v_{0,p} \left(t\right),
\end{align}
If we define a robust safety function $h_{0R}(x):\mathbb{R}^{n} \to \mathbb{R}$ as
\begin{align}
    h_{0R}(x) = h_0(x) +  \frac{1}{2}\drlower\delayu^2,
\end{align}
then we have $h_{0R}(x_p(t)) \ge 0 \implies h_0(x(t+\tau))\ge 0$.

Step \uppercase\expandafter{\romannumeral3}:
Derive CBF constraints.

From \eqref{eq:eAD=100} and the coefficient definition \eqref{eq:model ABD}, we have
\begin{align}\label{eq:eAD=D}
    e^{A(\delayu-\theta)} D = D, \quad \forall \theta.
\end{align}
So we have
\begin{align}
    \int_{0}^{\delayu} e^{A(\delayu-\theta)} D \diff \theta  = \int_{0}^{\delayu} D \diff \theta  = D \delayu,
\end{align}
which gives
\begin{align}
    \int_{0}^{\delayu} e^{A(\delayu-\theta)} D r(t) \diff \theta = D\delayu r(t).
\end{align}
The predictor \eqref{eq:predict x} is re-written as
\begin{align}
    x_p(t) =& e^{A\delayu}x(t) 
     +e^{At}\int_{t-\delayu}^{t} e^{-A\theta} B u(\theta) \diff \theta  + D\delayu r(t) .
\end{align}
The dynamics of the predictor is
\begin{align}
 \dot x_p(t)  =&e^{A\delayu}\dot x(t) + A e^{At}\int_{t-\delayu}^{t} e^{-A(\theta-t)} B u(\theta) \diff \theta \notag \\
    &+ Bu(t) - e^{A\delayu}Bu(t-\delayu) + D \delayu \dot{r}(t) \notag\\
    =& Ax_p(t) + Bu(t) \notag \\
    &+ e^{A\delayu}Dr(t) - AD\delayu r(t) + D \delayu \dot{r}(t).
\end{align}
Since $AD = 0$ and $e^{A\delayu}D = D$, we have
\begin{align}
    \dot{x}_p(t) = Ax_p(t) + Bu(t) + Dr(t) + D \delayu \dot{r}(t).
\end{align}

The time derivative of $h_{0R}(x_p(t))$ is 
\begin{align}
    \dot{h}_{0R}(x_p(t)) = &\frac{\partial h_{0R}(x_p(t))}{\partial x_p} Ax_p(t) + \frac{\partial h_{0R}(x_p(t))}{\partial x_p} Bu(t) \notag \\
    &+  \frac{\partial h_{0R}(x_p(t))}{\partial x_p} D(r (t) +  \delayu \dot{r}(t)).
\end{align}
By the definition of $h_{0R}(x)$ and $h_0(x)$, we have
\begin{align}\label{eq:ph0_px}
    \frac{\partial h_{0R}(x_p)}{\partial x_p} = \frac{\partial h_0(x_p)}{\partial x_p} = \begin{bmatrix}
    H_0 &0_{1\times 2} \cdots 0_{1\times 2}
\end{bmatrix}^{\top},
\end{align}
with $H_0 = [1,-\psi_0]$, which gives
\begin{align}
    \dot{h}_{0R}(x_p(t)) 
    =& L_fh_0(x_p(t)) + L_gh_0(x_p(t))u(t) \notag \\
    & +\frac{\partial h_{0}(x_p(t))}{\partial x_p} D(r (t) +  \delayu \dot{r}(t)).
\end{align}
To guarantee $h_{0R}(x_p(t))\ge 0$, the control input should satisfy
\begin{align}\label{eq:CBF constraint h0R}
    \dot{h}_{0R}(x_p(t)) \ge -\alpha_0(h_{0R}(x_p(t))),
\end{align}
with $\alpha_0$ being an \EKF. However, since $\dot r(t)$, the acceleration of the head vehicle is unknown, $\dot{h}_{0R}(x_p(t))$ is also unknown. 
From the model coefficient $D$ in~\eqref{eq:model ABD}, we have
\begin{align}
    \frac{\partial h_0(x_p(t))}{\partial x_p} D=1 >0.
\end{align}
So we have 
\begin{align}
    \dot{h}_{0R}(x_p(t)) \ge & L_fh_0(x_p(t)) + L_gh_0(x_p(t))u(t) \notag \\
    &+ \frac{\partial h_0(x_p(t))}{\partial x_p} D (r(t) + \delayu \drlower),
\end{align}
where $\drlower$ is a bound on $\dot r$ in Assumption~\ref{assumption:bound rdot}.
To meet \eqref{eq:CBF constraint h0R}, it is sufficient to  have
\begin{align}
    &L_fh_0(x_p(t)) + L_gh_0(x_p(t))u(t) + \frac{\partial h_0(x_p(t))}{\partial x_p} D (r(t) + \delayu \drlower)  \notag \\
    &\ge -\alpha_0(h_{0R}(x_p(t))),
\end{align}
which is \eqref{eq:safe constraint CBF CAV theorem}.

Now we prove the safety constraints for HVs. We also follow  the aforementioned three steps.

Step \uppercase\expandafter{\romannumeral1}: Bound the prediction error.

By the prediction error derived in \eqref{eq:prediction error} and \eqref{eq:eAD=100}, we have that there is no prediction error for HV-$i$'s gap and speed, 
\begin{align}
    s_{i,p}(t) &= s_i(t+\delayu), \\
    v_{i,p}(t) &= v_i(t+\delayu),
\end{align}
where $s_{i,p}(t)$ and $v_{i,p}(t)$ denotes the predicted gap and speed for HV-$i$ respectively.

Step \uppercase\expandafter{\romannumeral2}: Construct a robust safety function.

Since there is no prediction error for HV-$i$'s state, we have $h_i \left( x_p \left(t\right) \right) \ge 0 \implies h_i(x(t+\delayu))\ge 0 $. However, since $h_i$ has a high relative degree, we construct a reduced-degree safety function as    
\begin{align}
    h_{iR}^{\mathrm{r}}(x) &= h_{i}(x) - \eta_i h_{0R}(x) \notag \\
    &= h_{i}^{\mathrm{r}}(x) - \frac{1}{2}\eta_i\drlower\delayu^2,
\end{align}
It is straightforward that if $h_{0R}(x_p)\ge0$ and $ h_{iR}^{\mathrm{r}}(x_p)\ge0$, then $h_i(x_p)\ge 0$, which implies that the original safe criterion $h_i(x(t+\delayu))$ is met.

Step \uppercase\expandafter{\romannumeral3}:
Derive CBF constraints.

The time derivative of ${h}_{iR}^{\mathrm{r}}(x_p(t))$ is
\begin{align}
    \dot{h}_{iR}^{\mathrm{r}}(x_p(t)) = &\frac{\partial h_{iR}^{\mathrm{r}} (x_p(t))}{\partial x_p} Ax_p(t) + \frac{\partial h_{iR}^{\mathrm{r}} (x_p(t))}{\partial x_p} Bu(t) \notag \\
    &+  \frac{\partial h_{iR}^{\mathrm{r}} (x_p(t))}{\partial x_p} D(r (t) +  \delayu \dot{r}(t)),
\end{align}
where
\begin{align}\label{eq:phir_px}
    &\frac{\partial h_{iR}^{\mathrm{r}} (x_p)}{\partial x_p} = \frac{\partial h_{i}^{\mathrm{r}}(x_p)}{\partial x_p} = \begin{bmatrix}
        -\eta_i H_0 \! & \! 0_{1\times 2} \! & \! \cdots \! & \!  H_{i} \! &\! \cdots \! & \!0_{1\times 2}
    \end{bmatrix}^{\top},
\end{align}
with $H_i = [1,-\psi_i]$. So we have 
\begin{align}
    \dot{h}_{iR}^{\mathrm{r}} (x_p(t))= &L_f h_{i}^{\mathrm{r}}(x_p(t)) + L_g h_{i}^{\mathrm{r}}(x_p(t)) u(t) \notag \\
    & + \frac{\partial h_{i}^{\mathrm{r}}(x_p)}{\partial x_p} D ( r(t) + \delayu \dot{r}(t)).
\end{align}

To guarantee $h_{iR}^{\mathrm{r}}(x_p(t))\ge 0$, the control input $u$ should satisfy
\begin{align}\label{eq:CBF constraint hiRr}
    \dot{h}_{iR} \left( x_p\left(t\right) \right) \ge -\alpha_i \left(h_{iR}\left(x_p\left(t\right)\right)\right).
\end{align}
We note that
\begin{align}
    \frac{\partial h_{i}^{\mathrm{r}} (x_p)}{\partial x_p} D = -\eta_i <0,
\end{align}
so a sufficient condition for \eqref{eq:CBF constraint hiRr} is
\begin{align}
    &L_f h_{i}^{\mathrm{r}}(x_p(t)) + L_g h_{i}^{\mathrm{r}}(x_p(t)) u(t) +\frac{\partial h_{i}^{\mathrm{r}}(x_p)}{\partial x_p} D ( r(t) + \delayu \drupper) \notag \\
    &\ge -\alpha_i(h_{iR}^{\mathrm{r}} \left( x_p\left(t\right) \right)),
\end{align}
where $\drupper$ is the bound on $\dot r $ in Assumption~\ref{assumption:bound rdot}. 

\section{Proof of Theorem \ref{theorem:safety observer}}\label{sec:appendix proof safety observer}

We first prove the safe constraints for CAV following the three steps.

Step \uppercase\expandafter{\romannumeral1}: Bound the prediction error.

We re-write the predictor \eqref{eq:predict x observer} as 
\begin{align}
    \hat x_p(t) =& e^{A\delayu}  x(t) - e^{A\delayu}  \epsilon(t)  +\int_{-\delayu}^{0} e^{-A\theta} B u(t+\theta) \diff \theta  \notag \\
     &+ \int_{0}^{\delayu} e^{A(\delayu-\theta)} D r(t) \diff \theta \notag \\
     = & x_p (t) - e^{A\delayu}  \epsilon(t).
\end{align}
The prediction error for \eqref{eq:predict x observer} is
\begin{align}
    \hat x_{e} (t) &= x(t+\delayu)  - \hat{x}_p(t) \notag \\
    &= x(t+\delayu)  - x_p(t) + e^{A\delayu}\epsilon(t) \notag \\
    &= x_{e}(t) + e^{A\delayu}\epsilon(t).
\end{align}
We use $\hat{s}_{i,p}(t)$ and $\hat{v}_{i,p}(t)$ to denote the predicted gap and speed for vehicle $i$, and we have
\begin{align}
     -\eobound e^{-\lambda t} \le \hat s_{i,e}(t) - s_{i,e} (t)
     &\le \eobound e^{-\lambda t}, \label{eq:predict error bound si observer general}\\
     -\eobound e^{-\lambda t} \le \hat v_{i,e}(t) - v_{i,e} (t) &\le \eobound e^{-\lambda t}, \label{eq:predict error bound vi observer general}
\end{align}
with $\eobound = \Vert e^{A\delayu}\Vert \Upsilon \bar{\epsilon}$.

From the bound on $s_{0,e}$ and $v_{0,e}$ in \eqref{eq:predict error bound s0 lower} and \eqref{eq:predict error bound v0}, we have the bound on $\hat s_{0,e}$ and $\hat v_{0,e}$ as
\begin{align}
        \frac{1}{2}\drlower\delayu^2 - \eobound e^{-\lambda t}\le \hat{s}_{0,e} (t)  & \le \frac{1}{2}\drupper \delayu^2 + \eobound e^{-\lambda t}, \label{eq:predict error bound s0 observer}\\
        - \eobound e^{-\lambda t}\le \hat{v}_{0,e} (t) & \le \eobound e^{-\lambda t}. \label{eq:predict error bound v0 observer}
\end{align}

Step \uppercase\expandafter{\romannumeral2}: Construct a robust safety function.

By the prediction error bound in \eqref{eq:predict error bound s0 observer} and \eqref{eq:predict error bound v0 observer}, CAV safety \eqref{eq:safe s} is satisfied if
\begin{align}
    \hat s_{0,p}\left(t\right)  + \frac{1}{2}\drlower\delayu^2 - \eobound e^{-\lambda t}\ge \psi_0 \hat v_{0,p} \left(t\right)+ \psi_0  \eobound e^{-\lambda t}.
\end{align}
If we define an observer-based robust safety function as 
\begin{align}
    h_{0R}^o\left(\hat x_p \left(t\right),t\right) &= h_0(\hat x_p(t)) + \frac{1}{2}\drlower\delayu^2+W_0(t),
\end{align}
with
\begin{align}
    W_0(t) = -(1+\psi_0)\eobound e^{-\lambda t},
\end{align}
then we have
\begin{align}
    h_{0R}^o\left(\hat x_p \left(t\right),t\right)   \ge 0 \implies h_0(x(t+\delayu))\ge 0.
\end{align}

Step \uppercase\expandafter{\romannumeral3}:
Derive CBF constraints.

We re-write the predictor \eqref{eq:predict x observer} as 
\begin{align}
    \hat{x}_p(t) =& e^{A\delayu} \hat x(t) 
     +e^{At}\int_{t-\delayu}^{t} e^{-A\theta} B u(\theta) \diff \theta  + D\delayu r(t) .
\end{align}
The dynamics of $\hat x_p(t)$ is 
\begin{align}
    \dot{\hat{x}}_p(t) =&e^{A\delayu} \dot{\hat{x}} (t) + A e^{At}\int_{t-\delayu}^{t} e^{-A(\theta-t)} B u(\theta) \diff \theta \notag \\
    &+ Bu(t) - e^{A\delayu}Bu(t-\delayu) + D \delayu \dot{r}(t) \notag \\
    =& A\hat{x}_p(t) + Bu(t) + Dr(t) + D  \dot{r}(t)\delayu \notag \\
    &+ e^{A\delayu}L(Y(t)-\bar{C}\hat x(t) ).
\end{align}

By the CBF constraints for time varying functions $h(x,t)$~\cite{xu2018constrained}, to guarantee $h_{0R}^o\left(\hat x_p \left(t\right),t\right)   \ge 0$, the control input $u$ should satisfy 
\begin{align}\label{eq:CBF constraint h0R observer}
    \dot{h}_{0R}^o\left(\hat x_p \left(t\right),t\right) \ge -\alpha_0( h_{0R}^o \left(\hat x_p \left(t\right)\right)). 
\end{align}
The time derivative of $h_{0R}^o\left(\hat x_p \left(t\right),t\right)$ is 
\begin{align}
    \dot{h}_{0R}^o\left(\hat x_p \left(t\right),t\right) =& \dot{h}_{0} \left(\hat x_p \left(t\right)\right)+\dot{W}_0(t) \notag \\
    =& \frac{\partial h_{0}(\hat x_p(t))}{\partial \hat x_p} A\hat x_p(t) + \frac{\partial h_{0}(\hat x_p(t))}{\partial \hat x_p} Bu(t) \notag \\
    &+  \frac{\partial h_{0}(\hat x_p(t))}{\partial \hat x_p} D(r (t) +  \delayu \dot{r}(t)) \notag\\
    &+\frac{\partial h_{0}(\hat x_p(t))}{\partial \hat x_p} e^{A\delayu}L(Y(t)-\bar{C}\hat x(t) ) \notag \\
    &+\dot{W}_0(t).
\end{align}
From the gradient of $h_{0R}$ in \eqref{eq:ph0_px} and the bound on $\dot r$ in Assumption \ref{assumption:bound rdot}, we have
\begin{align}
    \dot{h}_{0R}^o\left(\hat x_p \left(t\right),t\right) \ge & L_fh_0(\hat x_p(t))  + L_gh_0(\hat x_p(t)) \notag \\
    &+ \frac{\partial h_{0}(\hat x_p(t))}{\partial \hat x_p} D (r(t) + \delayu \drlower) \notag\\
    &+\frac{\partial h_{0}(\hat x_p(t))}{\partial \hat x_p} e^{A\delayu}L(Y(t)-\bar{C}\hat x(t) ) \notag \\
    &+\dot{W}_0(t).
\end{align}
Therefore, to meet \eqref{eq:CBF constraint h0R observer}, it is sufficient to have \eqref{eq:safe constraint CBF CAV theorem observer}.

Now we prove the safe constraints for HVs.

Step \uppercase\expandafter{\romannumeral1}: Bound the prediction error.

From the prediction error on \eqref{eq:predict error bound si observer general} and \eqref{eq:predict error bound vi observer general}, since $s_{i,e}=0$ and $v_{i,e} = 0$ as \eqref{eq:prediction error} shows, we have 
\begin{align}
        - \eobound e^{-\lambda t}\le \hat{s}_{i,e} (t) & \le \eobound e^{-\lambda t},\quad \forall i=1,\cdots,N, \label{eq:predict error bound si observer}\\
        - \eobound e^{-\lambda t}\le \hat{v}_{i,e} (t) & \le \eobound e^{-\lambda t},\quad \forall i=1,\cdots,N. \label{eq:predict error bound vi observer}
\end{align}

Step 
\uppercase\expandafter{\romannumeral2}: Construct a robust safety function.

By the prediction error bound on gap and speeed in \eqref{eq:predict error bound si observer}  and \eqref{eq:predict error bound vi observer}, the safety of HVs is guaranteed if
\begin{align}
    \hat s_{i,p}  \left(t\right) - \eobound e^{-\lambda t} \ge \hat v_{i,p}  \left(t\right) + \psi_i \eobound e^{-\lambda t}.
\end{align}
If we define an observer-based robust safety function for HV-$i$ as
\begin{align}
    h_{iR}^o(x) = h_i (x) + W_i(t),
\end{align}
with
\begin{align}
    W_i(t) = - (1+\psi_i) \eobound e^{-\lambda t},
\end{align}
then 
\begin{align}
    h_{iR}^o(x_p(t)) \ge 0\implies h_i(x(t+\delayu)) \ge 0.
\end{align}
We note that $h_{iR}^o(x)$  also has a high relative degree, so we introduce reduced order CBF candidates as 
\begin{align}
    h_{iR}^{o\mathrm{r}}(x) &= h_{iR}^o(x) - \eta_i h_{0R}^o(x) \notag \\
    &=h_{i}^{\mathrm{r}}(x) -\frac{1}{2}\drlower\eta_i\delayu^2 +W_i(t) - \eta_i W_0(t).
\end{align}

Step \uppercase\expandafter{\romannumeral3}:
Derive CBF constraints.

The safety condition ${h}_{iR}^{o\mathrm{r}}(\hat x_p(t)) \ge 0 $ is satisfied if
\begin{align}\label{eq:CBF constraint hiR observer}
    \dot{h}_{iR}^{o\mathrm{r}}(\hat x_p(t)) \ge -\alpha_i({h}_{iR}^{o\mathrm{r}}(\hat x_p(t))).
\end{align}
The time derivative of $\dot{h}_{iR}^{o\mathrm{r}}(\hat x_p(t),t)$ is 
\begin{align}
    \dot{h}_{iR}^{o\mathrm{r}}(\hat x_p(t),t) &= \dot{h}_{i}^{\mathrm{r}}(\hat x_p(t)) +\dot W_i(t) - \eta_i \dot{W}_0(t) \notag \\
    &+\frac{\partial h_{i}^{\mathrm{r}}(\hat x_p(t))}{\partial \hat x_p} e^{A\delayu}L(Y(t)-\bar{C}\hat x(t) ).
\end{align}
From the gradient of $h_{i}^{\mathrm{r}}$ in \eqref{eq:phir_px}, we have 
\begin{align}
    \dot{h}_{iR}^{o\mathrm{r}}(\hat x_p(t),t) \ge & L_f h_{i}^{\mathrm{r}}(\hat x_p(t)) + L_gh_{i}^{\mathrm{r}}(\hat x_p(t)) u(t) \notag \\
    & + \frac{\partial h_{i}^{\mathrm{r}}(\hat x_p(t))}{\partial \hat x_p}  ( r(t)  + \delayu \drupper )\notag \\
    & +\frac{\partial h_{i}^{\mathrm{r}}(\hat x_p(t))}{\partial \hat x_p} e^{A\delayu}L(Y(t)-\bar{C}\hat x(t) ) \notag \\
    & +\dot W_i(t) - \eta_i \dot{W}_0(t).
\end{align}
So \eqref{eq:CBF constraint hiR observer} is satisfied if we have \eqref{eq:safe constraint CBF HDV theorem observer}.

\bibliographystyle{IEEEtranS}
\bibliography{ref}

\end{document}